\documentclass[acmsmall,screen]{acmart}

\setcopyright{acmcopyright}
\acmPrice{}
\acmDOI{10.1145/3408992}
\acmYear{2020}
\copyrightyear{2020}
\acmSubmissionID{icfp20main-p78-p}
\acmJournal{PACMPL}
\acmVolume{4}
\acmNumber{ICFP}
\acmArticle{110}
\acmMonth{8}

\usepackage{booktabs}   
\usepackage{subcaption} 

\usepackage{proof}
\usepackage{mathpartir}
\usepackage{listings}
\usepackage{xfrac}
\usepackage{mathtools}
\usepackage{ifthen}
\usepackage{stmaryrd}
\usepackage{wrapfig}

\usepackage{cleveref}

\crefname{theorem}{Thm.}{Thms.}
\crefname{lemma}{Lem.}{Lemmas}
\crefname{corollary}{Cor.}{Cors.}
\crefname{figure}{Fig.}{Figs.}
\crefname{definition}{Defn.}{Defns.}
\crefname{table}{Tab.}{Tabs.}
\crefformat{section}{\S#2#1#3}
\crefmultiformat{section}{\S#2#1#3}{ and~\S#2#1#3}{, \S#2#1#3}{ and~\S#2#1#3}
\crefname{section}{\S}{\S\S}
\crefname{example}{Ex.}{Exs.}
\crefname{item}{item}{item}
\crefname{footnote}{footnote}{footnote}
\crefname{observation}{Obs.}{Obs.}
\crefname{remark}{Remark}{Remarks}
\crefname{proposition}{Prop.}{Props.}
\crefname{counterexample}{Counterexample}{Counterexamples}
\usepackage{shortcuts}

\usepackage{array}
\newcolumntype{H}{>{\setbox0=\hbox\bgroup}c<{\egroup}@{}}

\AtBeginDocument{%
\setlength\abovedisplayskip{0.5\abovedisplayskip}%
\setlength\belowdisplayskip{0.5\belowdisplayskip}%
\setlength\abovedisplayshortskip{0.5\abovedisplayshortskip}%
\setlength\belowdisplayshortskip{0.5\belowdisplayshortskip}%
\setlength\floatsep{0.5\floatsep}%
\setlength{\belowcaptionskip}{-5pt}%
}

\newcommand{\word}[1]{\ensuremath{\mathop{\,\mathsf{#1}\,}}}
\newcommand{\abt}[1]{\ensuremath{\mathrm{#1}}}
\newcommand{\code}[1]{\text{\sl #1}}

\newcommand{\li}[2][XX]{%
   \ifthenelse{\equal{#1}{XX}}%
      {\ensuremath{L(#2)}}%
            {\ensuremath{L^{#1}(#2)}}}
            
\newcommand{\share}[3]{\ensuremath{#1 \mathop{\curlyveedownarrow} (#2, #3)}}
\newcommand{\zero}[1]{\ensuremath{|#1|}}
\newcommand{\isZero}[1]{\ensuremath{#1 = \zero{#1}}}
\newcommand{\subtype}[0]{\ensuremath{<:}}

\newcommand{\unitS}[0]{\ensuremath{\abt{triv}}}
\newcommand{\funS}[3]{\ensuremath{\abt{fun}(#1,#2.#3)}}
\newcommand{\appS}[2]{\ensuremath{\abt{app}(#1;#2)}}

\newcommand{\nilS}[0]{\ensuremath{\abt{nil}}}
\newcommand{\consS}[2]{\ensuremath{\abt{cons}(#1;#2)}}
\newcommand{\matLS}[5]{\ensuremath{\abt{mat_L}\{#1;#2,#3.#4\}(#5)}}
\newcommand{\tickS}[1]{\ensuremath{\abt{tick}\{#1\}}}
\newcommand{\letS}[3]{\ensuremath{\abt{let}(#1;#2.#3)}}
\newcommand{\shareS}[4]{\ensuremath{\abt{share}(#1;#2,#3.#4)}}
\newcommand{\flipS}[3]{\ensuremath{\abt{flip}\set{#2;#3}(#1)}}
\newcommand{\probS}[1]{\ensuremath{\abt{prob}\{#1\}}}
\newcommand{\sflipS}[3]{\ensuremath{\abt{flip_\bbS}(#1;#2,#3)}}

\newcommand{\nilC}[0]{\ensuremath{[]}}
\newcommand{\consC}[2]{\ensuremath{#1 \dblcolon #2}}
\newcommand{\matLC}[5]{\ensuremath{%
\word{case} #5\; \{ \word{nil} \hookrightarrow #1 %
\mid \word{cons}(#2,#3) \hookrightarrow #4%
\}}}
\newcommand{\unitC}[0]{\ensuremath{\tuple{}}}
\newcommand{\funC}[3]{\ensuremath{\word{fun } #1 #2 = #3}}
\newcommand{\appC}[2]{\ensuremath{#1(#2)}}

\newcommand{\tickC}[1]{\ensuremath{\word{tick } #1}}
\newcommand{\letC}[3]{\ensuremath{\word{let } #2 = #1 \word{ in } #3}}
\newcommand{\shareC}[4]{\ensuremath{\word{share } #1 \word{ as } #2,#3 \word{ in } #4}}
\newcommand{\flipC}[3]{\ensuremath{\word{flip} #1\;\{  \word{H} \hookrightarrow #2  \mid \word{T} \hookrightarrow #3 \}}}
\newcommand{\condC}[3]{\word{if } #1\word{then } #2\word{else } #3}
\newcommand{\probC}[1]{\ensuremath{#1}}
\newcommand{\sflipC}[3]{\ensuremath{\word{flip_\bbS} #1\; \{ \word{H} \hookrightarrow #2 \mid \word{T} \hookrightarrow #3 \}}}

\newcommand{\arrT}[2]{\ensuremath{#1 \to #2}}
\newcommand{\listT}[2]{\ensuremath{\li[#2]{#1}}}
\newcommand{\unitT}[0]{\ensuremath{\vvmathbb{1}}}
\newcommand{\boolT}[0]{\ensuremath{\mathsf{bool}}}
\newcommand{\intT}[0]{\ensuremath{\mathsf{int}}}
\newcommand{\prodT}[2]{\ensuremath{#1 \times #2}}
\newcommand{\annoT}[2]{\ensuremath{\langle #1,#2 \rangle}}
\newcommand{\probT}[2][XX]{%
  \ifthenelse{\equal{#1}{XX}}%
    {\ensuremath{\vvmathbb{P}}}%
    {\ensuremath{\vvmathbb{P}^{#1}_{#2}}}}

\newcommand{\arrA}[2]{\ensuremath{\abt{arr}(#1;#2)}}

\newcommand{\listA}[1]{\ensuremath{\abt{list}(#1)}}
\newcommand{\unitA}[0]{\ensuremath{\abt{unit}}}

\newcommand{\annoA}[2]{\ensuremath{\abt{pot}(#1;#2)}}
\newcommand{\probA}[2]{\ensuremath{\abt{prob}\{#1;#2\}}}

\newcommand{\unitV}{\ensuremath{\tuple{}}}

\newcommand{\nilV}{\ensuremath{[]}}
\newcommand{\consV}[2]{\ensuremath{#1 \dblcolon #2}}
\newcommand{\cloV}[4]{\ensuremath{\mathsf{clo}(#1;#2,#3.#4)}}
\newcommand{\probV}[1]{\ensuremath{\mathsf{prob}(#1)}}

\newcommand{\typingrel}[4]{\ensuremath{#1;#2 \vdash #3 : #4}}
\newcommand{\valuerel}[2]{\ensuremath{#1 : #2}}
\newcommand{\evalrel}[6]{\ensuremath{#1;#6 \vdash #2 \Downarrow^{#5} #3 \mid #4 }}
\newcommand{\evalrelnotr}[5]{\ensuremath{#1 \vdash #2 \Downarrow^{#5} #3 \mid #4}}
\newcommand{\steprel}[4]{\ensuremath{#1 \vdash #2 \Rightarrow^{#4} #3 }}

\newcommand{\pot}[2]{\ensuremath{\Phi(#1 : #2})}

\renewcommand{\Rule}[4][]{\ensuremath{\inferrule*[right={(#2)},#1]{#3}{#4}}}

\newcommand{\jan}[1]{{\color{ACMGreen}Jan: #1}}

\newcommand{\versioning}[2]{#1}

\usepackage[color]{changebar}
\cbcolor{ACMRed}
\nochangebars

\begin{document}

\title[Raising Expectation]{Raising Expectations: Automating Expected Cost Analysis with Types}



\author{Di Wang}
\orcid{nnnn-nnnn-nnnn-nnnn}             
\affiliation{
  \position{PhD Student}
  \department{Computer Science Department}             
  \institution{Carnegie Mellon University}           
  \streetaddress{5000 Forbes Ave}
  \city{Pittsburgh}
  \state{PA}
  \postcode{15213}
  \country{USA}                   
}
\email{diw3@andrew.cmu.edu}          

\author{David M. Kahn}
\affiliation{
  \position{PhD Student}
  \department{Computer Science Department}             
  \institution{Carnegie Mellon University}           
  \streetaddress{5000 Forbes Ave}
  \city{Pittsburgh}
  \state{PA}
  \postcode{15213}
  \country{USA}                   
}
\email{davidkah@andrew.cmu.edu}         

\author{Jan Hoffmann}
\affiliation{
  \position{Assistant Professor}
  \department{Computer Science Department}             
  \institution{Carnegie Mellon University}           
  \streetaddress{5000 Forbes Ave}
  \city{Pittsburgh}
  \state{PA}
  \postcode{15213}
  \country{USA}                   
}
\email{jhoffmann@cmu.edu}         

\begin{abstract}

This article presents a type-based analysis for deriving upper bounds
on the expected execution cost of probabilistic programs.
The analysis is naturally compositional, parametric in the cost
model, and supports higher-order functions and inductive data types.
The derived bounds are multivariate polynomials that are functions of
data structures.
Bound inference is enabled by local type rules that reduce type
inference to linear constraint solving.
The type system is based on the potential method of amortized analysis
and extends automatic amortized resource
analysis (AARA) for deterministic programs.
A main innovation is that bounds can contain symbolic probabilities,
which may appear in data structures and function arguments.
Another contribution is a novel soundness proof that establishes the
correctness of the derived bounds with respect to a distribution-based
operational cost semantics that also includes nontrivial diverging
behavior.
For cost models like time, derived bounds imply termination with
probability one.
To highlight the novel ideas, the presentation focuses on linear
potential and a core language.
However, the analysis is implemented as an extension of Resource Aware
ML and supports polynomial bounds and user defined data structures.
The effectiveness of the technique is evaluated by analyzing the
sample complexity of discrete distributions 
and with a novel average-case estimation for deterministic programs
that combines expected cost analysis with statistical methods.


\end{abstract}



\maketitle

\section{Introduction}
\label{sec:intro}

Probabilistic programming~\cite{JCSS:Kozen81,book:MM05} is an
effective tool for customizing probabilistic inference~\cite{kn:CGH17,misc:dippl,PLDI:MSH18}
as well as for modeling and analyzing randomized
algorithms~\cite{TassarottiH19}, cryptographic
protocols~\cite{POPL:BGB09}, and privacy mechanisms~\cite{POPL:BKO12}.
In this paper, we study probabilistic programs as models of the
execution cost (or resource use) of programs. \Omit{ \jan{I wanted to say that the programs are models.}}
Execution cost can be defined by a cost semantics or a
programmer-defined metric.
For such a cost model, a probabilistic program defines a distribution
of cost that depends on the distribution of the inputs as well as the
probabilistic choices that are made in the code.

The problem of statically analyzing the cost distribution of
probabilistic programs has attracted growing attention in recent
years.
Kaminski et al.~\cite{ESOP:KKM16,LICS:OKK16} have built on the work of Kozen~\cite{JCSS:Kozen81}, studying
weakest-precondition calculi for deriving upper bounds on the expected
worst-case cost of imperative programs, as well as reasoning about lower
bounds~\cite{POPL:HKG20}.
%
It has been shown that this calculus can be specialized to
automatically infer constant bounds on the sampling cost of
non-recursive Bayesian networks~\cite{ESOP:BKKM18} and polynomial
bounds on the worst-case expected cost of arithmetic programs~\cite{PLDI:NCH18,POPL:CFN16,CAV:CFG16}.
The key innovation that enables the inference of symbolic bounds is a
template-based approach that reduces bound inference to efficient \emph{linear-program} (LP)
solving, a reduction which has been previously applied non-probabilistic
programs~\cite{PLDI:CHS15,CAV:CHR17}.
This technique has been extended to best-case bounds and
non-monotone cost~\cite{PLDI:WFG19} as well as to incorporate higher-moment reasoning for
deriving tail bounds using linear~\cite{WangHR20} and
non-linear~\cite{TACAS:KUH19} constraint solving.

The only existing technique for analyzing the expected cost of
probabilistic (higher-order) functional programs, is the recent work
of Avanzini et al.~\cite{LICS:ALG19}.
It applies an affine refinement type system,
called $\ell$RPCF, to derive bounds on the expected worst-case cost
for an affine version of PCF~\cite{Plotkin77}.
$\ell$RPCF can be seen as a probabilistic version of
d$\ell$PCF~\cite{LICS:LG11}.
While the refinement types of $\ell$RPCF are expressive and flexible,
a disadvantage is that the complexity of the corresponding refinement
constraints hampers type inference. It seems unclear if type
checking $\ell$RPCF is decidable.

This article presents the first automatic analysis of worst-case
bounds on the expected cost of probabilistic functional
programs. It is based on \emph{automatic
amortized resource analysis}
(AARA)~\cite{POPL:HJ03,POPL:HAH11}, a type system for
inferring worst-case bounds.
The expressivity of AARA's type-based approach for probabilistic programs goes
beyond existing techniques for imperative integer programs in the
following ways:
\begin{enumerate}

\item The analysis infers expected cost bounds for higher-order programs.

\item Bounds can be functions of the sizes of values of (potentially nested) inductive types

\item Bounds can be functions of symbolic probabilities.

\end{enumerate}

In addition, AARA for probabilistic programs preserves many advantageous features
of classical AARA for deterministic programs, which include
\begin{itemize}

\item[-] efficient type checking (linear in the size of the type derivation),

\item[-] reduction of type inference for \emph{polynomial bounds} to
  \emph{linear programming},

\item[-] use of the potential method to amortize operations with varying expected cost, and

\item[-] natural compositionality, as types summarize the cost behavior of functions.

\end{itemize}

Nonetheless, while AARA for deterministic programs naturally derives bounds on the
high-water mark of non-monotone resources that can become available
during evaluation (like memory), this is not the case for AARA for
probabilistic programs.
Reasoning about high-watermark resource usage of probabilistic
programs is in fact an open problem even for manual reasoning systems
for first-order languages.
This problem is out of the scope of this article and we limit the
development to monotone resources like time. The technical
difficulties with non-monotone resources are discussed in more detail
in \cref{sec:semantics}.

To focus on the novel ideas, we present the analysis for a simple
probabilistic functional language with probabilistic branching and
lists (\cref{sec:semantics}) \Omit{\jan{and binary trees?}} with linear
potential functions (\cref{sec:type}).
However, the results carry over to multivariate polynomial potential
functions and user-defined inductive data structures. We
implemented the analysis as an extension of Resource Aware
ML (RaML)~\cite{POPL:HDW17} that we call pRaML (\cref{sec:goat}).

The main technical innovations are the introduction of a type rule for
probabilistic branching, and a new type for symbolic probabilities
(\cref{sec:overview,sec:type}). While these new features are fairly intuitive,
proving their soundness with respect to a cost semantics is not. 
The existing proof method for deterministic AARA does not directly generalize
to the probabilistic setting because of the complexities introduced by
a probabilistic cost semantics. 
To address the challenges of the probabilistic setting, we present a
novel soundness proof with respect to a probabilistic operational cost
semantics based on Borgstr{\"o}m et al.'s trace-based and
step-indexed-distribution-based semantics \cite{ICFP:BLG16}
(\cref{sec:sound}).
The details are discussed in \cref{sec:semantics}.  

We evaluate the effectiveness of pRaML by analyzing textbook
examples (\cref{sec:goat}) and by exploring novel problem domains (\cref{sec:app}).
The first domain (\cref{sec:sample}) \Omit{\jan{@David: in you version it
  says ``section 7.1'' here. Please update your latex or fix whatever
  causes this.}} is the implementation and analysis of discrete
probability distributions. Specifically, we use pRaML to analyze the
\emph{sample complexity} of the distributions, i.e., on average, how
many steps a program needs to produce a sample from the target
distribution.  Low sample complexity has recently become an important
criterion for efficient sampler implementations, as many probabilistic
inference methods require billions of random samples~\cite{misc:Djuric19}.  We
also verify some more complex fractional bounds in pRaML using a
scaled model.
The second domain (\cref{sec:model}) is the estimation of average-case
cost of functional programs on a specific input distribution as a
three step process. First, we gather
statistics on the branching behavior of conditional branches by evaluating the program on small inputs
that are representative for the input distribution. Second,
the conditionals are replaced with probabilistic branches that mirror
the observed branching behavior on the small inputs. Third, the
resulting program is analyzed with pRaML to determine a symbolic bound
on the expected cost of the resulting probabilistic program for all
input sizes.

In summary, we make the following \emph{contributions}:
\begin{enumerate}
\item Design of a novel type-based AARA for probabilistic programs
\item Type soundness proof with respect to a probabilistic operational cost semantics
\item Implementation as an extension of RaML
\item Application of RaML to automatically analyze sample complexity
\item Automatic average-case analysis that combines the use of RaML with empirical statistics
\end{enumerate}

\section{Topic Overview}
\label{sec:overview}

\paragraph{AARA} 

The type system of \emph{automatic amortized resource analysis} (AARA)
is a pre-existing framework for inferring cost bounds for
deterministic functional
programs~\cite{POPL:HJ03,POPL:JHL10,POPL:HDW17}. It imbues its types
with potential energy so as to perform the \textit{physicist's method}
(or \textit{potential method}) of amortized
analysis~\cite{kn:Tarjan85}. When performing type inference, the
system generates linear constraints on this potential that, when
solved, provide the coefficients of polynomials or other
functions. These functions express concrete (non-asymptotic) bounds on
worst- or best-case~\cite{SP:NDF17} execution costs, parameterized by
input size. 

In more detail, the potential method works as follows.
We say that $\Phi : \mathsf{State} \to \bbQ_{\ge 0}$ is a valid potential function if, for all states $S \in \mathsf{State}$ and operations $o:S \to S$, the following holds.
\[
\Phi(S) \geq 0 \hspace{3em} \text{and} \hspace{3em} \Phi(S) \ge \mathit{cost}(S,o(S)) + \Phi(o(S)).
\]
The second inequality states that the potential of the current state is
sufficient to pay for the cost of the transition from $S$ to $o(S)$ and potential of the
next state.
It then follows that the potential of the initial state establishes an
\emph{upper} bound on the \emph{worst-case} cost of a sequence of operations.

The AARA type system is designed to automatically assign such
potential functions to functional programs, where we view evaluation
steps as operations on machine states of an abstract machine.
Automation is enabled by fixing the format potential functions to linear combinations of base functions, and then 
incorporating them into the types of values.
Consider for example the function \code{exists} from the OCaml List
module in \cref{fig:exists}.
We model its cost behaviour using explicit \code{tick(q)} expressions that
consume $q \geq 0 \in \mathbb{Q}$ when evaluated.
The function \code{exists pred lst} has a cost of $1$ for in every
recursive call, and therefore the worst-case cost is equal to the
length of \code{lst} in addition to the cost of the calls to the function
\code{pred}.

To automatically derive this bound in linear AARA we assign the
following type template where $q_0,q_1,q,p,r$ and $r'$ are yet unknown
non-negative coefficients.
\[
\code{exists} : \arrT{\annoT{\arrT{\annoT{\tau}{r}}{\annoT{\boolT}{r'}}} {q_0}}{\arrT{\annoT{\listT \tau p} {q_1}}{\annoT {\boolT} {q}}}
\]

A valid instantiation of the potential annotation would for instance
be the following type.
\[
\code{exists} : \arrT{\annoT{\arrT{\annoT{\tau}{0}}{\annoT{\boolT}{0}}} {0}}{\arrT{\annoT{\listT \tau {1}} {0}}{\annoT {\boolT} {0}}}
\]

If we ignore the potential annotations in $\tau$ and the cost of evaluating the function \code{pred}, then this type expresses that the cost of evaluating \code{exists pred lst} is $1 \cdot |\code{lst}|$, as marked by requiring a list argument with 1 unit of potential per element. Another valid typing is
\[
\code{exists} : \arrT{\annoT{\arrT{\annoT{\tau}{2}}{\annoT{\boolT}{0}}} {0}}{\arrT{\annoT{\listT \tau {3}} {0}}{\annoT {\boolT} {0}}} \; .
\]
It now expresses that the cost of evaluating \code{exists pred
  lst} is $3 \cdot |\code{lst}|$ if the cost of evaluating \code{pred} is raised to $2$. The $\code{pred}$ function here is typed to take 2 units of potential to run, but is balanced by each element of the list argument being paired with 3 units of potential, 2 more than previously. 
  
  In general, type inference constrains this type's annotation variables with $p \geq r +1$ and $q_1 \geq q$, and leaves the other annotations unconstrained. This aids in the compositionality of the approach, as the specific constants chosen can be adapted to the arguments, including arguments that are themselves functions like $\code{pred}$ here. 
  
  To exemplify such compositionality, consider some function \code{f} that merely iterates over a list, consumes 1 resource every iteration, and then returns the list. It can be typed $\arrT{\annoT{\listT \tau {p}} {0}}{\annoT{\listT \tau {q}} {0}}$ where $p \geq q+1$. If we chain its application to some list \code{lst} as \code{f (f lst)}, then we might instantiate the type of the inner application with $p=2,q=1$, and the outer with $p=1,q=0$, composing the costs naturally. In this case, we would also type \code{lst} as $\listT \tau {2}$.
  
  Of course, AARA cannot do the impossible of successfully analyzing all programs. AARA uses structural reasoning methods that cannot pick up on semantical properties that the program may depend on, like Peano arithmetic. Further, not all resource usage can be accurately expressed in a given class of resource functions. For instance, polynomials will over-approximate logarithms, and simply cannot express exponentials. The resource functions we present in this paper are linear, but we make use of polynomial resource functions in our implementation.
  

\begin{figure}
  \centering
  \begin{subfigure}{0.32\textwidth}
\begin{lstlisting}[xleftmargin=0pt]
let rec exists pred lst =
	match lst with
	| [] -> false
	| hd::tl ->
		if pred hd
		then true
		else let _ = tick 1 in
			exists pred tl
\end{lstlisting}
\caption{\label{fig:exists}}
  \end{subfigure}
  \begin{subfigure}{0.32\textwidth}
\begin{lstlisting}[xleftmargin=0pt]
let rec bernoulli lst =
	match lst with
	| [] -> false
	| hd::tl ->
		match flip 0.5 with
		| H -> true
		| T ->  let _ = tick 1 in
			bernoulli tl
\end{lstlisting}
\caption{\label{fig:bernoulli}}
  \end{subfigure}
  \begin{subfigure}{0.34\textwidth}
\begin{lstlisting}[xleftmargin=0pt]
let rec rdwalk lst =
  match lst with
  | [] -> ()
  | p::ps ->
    let _ = tick 1 in
    match flip p with
    | H -> rdwalk (0.2::0.4::ps)
    | T -> rdwalk ps
\end{lstlisting}
\caption{\label{fig:rdwalk}}
  \end{subfigure}

\caption{Implementations of probabilistic programs in pRaML.}
\end{figure}

\paragraph{Probabilistic programming}

In this paper, we extend AARA to deriving bounds on the expected cost
of probabilistic programs.
In contrast to a deterministic program, a probabilistic program may
not always evaluate to the same value (if any), but rather to a distribution over values and divergence. 
Similarly, the evaluation cost of a probabilistic program is given
by a distribution. 

Consider for example the function \code{bernoulli} in
\cref{fig:bernoulli}. 
It is similar to the function \code{exists}, but the conditional is
replaced with the probabilistic construct \code{match flip 0.5}.
Intuitively, this construct means that we flip a coin
and evaluate the heads or tails branch based on the outcome.
In probabilistic programming, we assume that such flips are truly
random (as opposed to an implementation that may rely on a
pseudorandom number generator).
As a result, function \code{bernoulli} describes a Bernoulli process across the
elements of an input list. It terminates with probability $1$ and has
the same linear worst-case cost as $\mathit{exists}$, namely
$1 \cdot |\code{lst}|$.  However, the expected cost of \code{bernoulli} is
only $1$.

For an example with a more interesting expected cost, consider the
function \code{rdwalk} in \cref{fig:rdwalk}. Its argument is a list of
probabilities that are used, one after another, to determine the odds
in a probabilistic branch that either pops the head off the list (in
the tails case) or adds two new probabilities to the list (in the heads case). The random
walk consumes $1$ \code{tick} in each iteration and terminates if the
argument list is empty. One can show that the function \code{rdwalk}
terminates with probability 1 and the expected cost is a function of the
argument $[p_1,\ldots,p_n]$ as
\[
n + \sum_{1 \leq i \leq n} 5 p_i \; .
\]

This is an example of a program with non-terminating execution that may nonetheless have expected costs that can be bounded. If only finite cost is accrued on non-terminating execution, nontermination may even occur with positive probability and still yield a finite bound. Conversely, programs that
terminate with probability $1$ may still have unbounded expected cost, e.g., a symmetric random walk over natural numbers that stops at 0~\cite{book:MM05}. 

\paragraph{AARA for Expected Cost}

Now reconsider the potential method in the presence of probabilistic operations,
that is, the cost and the next state of an operation are given by distributions.
Let $o(S)$ denote the probability distribution over possible next states induced by $o$ operating on $S$.
One can derive bounds on the worst-case \textit{expected} cost by requiring that the following inequality for the potential function holds over all states $S$ and operations $o$. We use the notation $\bbE_{S' \sim o(S)}$ (defined in \cref{sec:semantics}) to weight expected cost over states $S'$ by the probability given by $o(S)(S')$.
\[
\Phi(S) \ge \bbE_{S' \sim o(S)}(\mathit{cost}(S,S') + \Phi(S')) = \bbE_{S' \sim o(S)}(\mathit{cost}(S,S')) + \bbE_{S' \sim o(S)} (\Phi(S')) ,
\]
The intuitive meaning of the inequality is that the potential $\Phi(S) \geq 0$ is sufficient to pay for the expected cost of the operation $o$ from the state $S$, and the expected potential of the next state $S'$ with respect to the probability distribution $o(S)$.

Further, if for some operation $o'$ we have $\Phi(S') \ge \bbE_{S'' \sim o'(S')}(\mathit{cost}(S',S'')) + \bbE_{S'' \sim o'(S')}(\Phi(S''))$ for each state $S'$ the could succeed $S$ under $o$, then we can \emph{compose} the reasoning for $o$ and $o'$ as follows.
\begin{align*}
  \Phi(S) & \ge   \bbE_{S' \sim o(S)}(\mathit{cost}(S,S')) + \bbE_{S' \sim o(S)} (\Phi(S')) \\
  & \ge \bbE_{S' \sim o(S)}(\mathit{cost}(S,S')) + \bbE_{S' \sim o(S)} \lrsq{ \bbE_{S'' \sim o'(S')}(\mathit{cost}(S',S'')) + \bbE_{S'' \sim o'(S')}(\Phi(S'')) } \\
  & = \bbE_{S' \sim o(S),S''\sim o'(S')}(\mathit{cost}(S,S')+\mathit{cost}(S',S'')) + \bbE_{S' \sim o(S), S'' \sim o'(S')}(\Phi(S'')).
\end{align*}
Thus, the potential $\Phi(S)$ is sufficient to cover the expected cost of operations $o$ and $o'$, as well as the expected potential of the final state. This can be sequenced indefinitely to cover all operations of an entire program. A valid potential assignment for the initial state of the program then provides an \emph{upper} bound on the \emph{expected} total cost of running the program.

In \cref{sec:type}, we extend the AARA type system to support this
kind of potential-method reasoning while preserving the benefits of
AARA such as compositionality and reduction of type inference to LP
solving.
For example, our probabilistic extension to AARA can type the code of
the function \code{bernoulli} in \cref{fig:bernoulli} as
\[
\code{bernoulli} : \arrT{\annoT{\listT \tau 0} 1}{\annoT {\boolT} 0}
\]
where the input can be typed as a list with $0$ units of potential per
element (assuming $\tau$ does not assign potential).
To cover the expected cost, it only needs 1 available potential unit
per run, indicated by the 1 paired with the input type. When typing
the probabilistic $\code{flip}$, this single unit of potential can
pay for the expected cost of the two equally-likely branches: The $H$
branch costs 0, the $T$ branch costs 2 (1 each for the recursive call
and for $\code{bernoulli}$ to consume), and they average to 1. As
$\code{bernoulli}$ can be typed to consume 1 unit of potential, the upper bound AARA
finds is exact.

The functions \code{bernoulli} and \code{exists} form an example of the
automatic average-case estimation algorithm that we introduce in
\cref{sec:model}. Assume that you want to run \code{exists} on a
certain distribution of inputs and you want to determine the average
cost of \code{exists} on this distribution. To approximately answer this question,
we collapse code like $\mathit{exists}$ into code like
$\mathit{bernoulli}$ and use pRaML to estimate that the average cost
is $1$. In this case, such a collapse would be justified by finding
empirically that $\mathit{pred}$ holds with probability $0.5$.

The technical innovation that makes possible the typing of
\code{bernoulli} is a new typing rule for probabilistic branching.
Another innovation is the introduction of the type $\probT{}{}$ for
probabilities. The introduction form for values of type $\probT{}{}$
simply takes a rational number $0 \leq p \leq 1$ and the elimination form
is a probabilistic branch.
We can assign potential 
$$
\pot{p}{\probT[q_H]{q_T}} \defeq q_H \cdot p + q_T \cdot (1-p)
$$
to a value $p$ of type $\probT{}{}$. The potential $q_H$ and $q_T$
then becomes available in the head and tails cases, respectively, of
the probabilistic branching.

Consider for example the function \code{rdwalk} in \cref{fig:rdwalk}
again. Our probabilistic analysis can automatically derive the
typing
\[
\code{rdwalk} : \arrT{\annoT{\listT {\probT[5]{0}} 1} 0}{\annoT{ \mathsf{unit} } 0} \; .
\]
The potential of the argument
\[
\pot{[p_1,\ldots,p_n]}{\annoT{\listT {\probT[5]{0}} 1} 0} = n + \sum_{1 \leq i \leq n} 5 p_i \; ,
\]
corresponds to the exact bounds on the expected cost.

Here we present these novel ideas for a simple functional
language with lists and linear potential functions. However, the
results carry over to user-defined inductive types and multivariate
polynomial potential functions of RaML~\cite{POPL:HDW17} that we use
in the implementation.
The main theorem of this paper (see \cref{sec:sound}) states that the
expected cost bounds are sound, with respect to a step-indexed distribution-based
operational semantics inspired by Borgstr{\"o}m et al.'s semantics for the
probabilistic lambda calculus~\cite{ICFP:BLG16}.
We then extend the semantics with \emph{partial evaluations} to capture the resource
behavior of non-terminating executions of a probabilistic program.
This novel extension enables an improved soundness result, which implies that expected
bounds on run-times ensure termination with probability 1.


\section{Language and Semantics}
\label{sec:semantics}

In this section, we introduce a subset of pRaML as a functional ML-like language that includes units, lists, recursion, pattern match, and a new \emph{flip} expression for probabilistic branching.
We then present an initial form of our operational cost semantics for probabilistic programs, which keeps track of both the probability and the cost of executions.
We will use this language and semantics to formalize and justify our type-based expected cost analysis in \cref{sec:type,sec:sound}.

\paragraph{Syntax}
We only consider expressions in \emph{share-let-normal-form}~\cite{POPL:HAH11}. This is a syntactic form that uses variables instead of arbitrary terms whenever possible, without loss of expressivity. This is done through maximizing the use of let-expressions. The syntax also must use $\shareS{x}{x_1}{x_2}{e}$ to allow multiple uses of a variable $x$ in an expression $e$, due to linear properties of the type system.
The abstract and concrete syntax of our probabilistic programming language is given by the grammar in \cref{fig:syntax}.
Abstract syntax is given via abstract binding trees~\cite{book:PFPL16}. While the concrete syntax matches the intuitive meaning of each expression, the abstract syntax conveys the same information and compacts some overly large expressions, allowing them to be written down more succinctly.

The syntactic form $\flipS{p}{e_1}{e_2}$ is introduced to execute $e_1$ or $e_2$ at random.
The intuitive meaning of the flip expression is to flip a biased coin, which shows heads with probability $p$ and tails with probability $(1-p)$, then execute $e_1$ if the coin shows heads, or execute $e_2$ if the coin shows tails.
Additionally, the introduction form $\probS{p}$ and the elimination form $\sflipS{x}{e_1}{e_2}$ are provided for the new \emph{probability} type:
$\probS{p}$ encapsulates a rational number $0 \le p \le 1$ for probability, and $\sflipS{x}{e_1}{e_2}$ is essentially the same as $\abt{flip}$ expressions except that the branching probability is specified by a variable $x$ of probability type.
The syntactic form $\shareS{x}{x_1}{x_2}{e}$ has to be used to allow multiple uses of a variable $x$ in an expression $e$.

\begin{figure}\small
\centering
\begin{tabular}{ccllll}
  & & \textbf{Abstract} & \textbf{Concrete} \\ \hline
  $e$ & $\Coloneqq$ & $x$ & $x$  & variable \\
  & & $\unitS$ & $\unitC$  & null tuple \\
  & & $\nilS$ & $\nilC$ & empty list \\
  & & $\consS{x_1}{x_2}$ & $\consC{x_1}{x_2}$ & cons list \\
  & & $\matLS{e_0}{x_1}{x_2}{e_2}{x}$ & $\matLC{e_0}{x_1}{x_2}{e_2}{x}$ & pattern match \\
  & & $\funS{f}{x}{e}$ & $\funC{f}{x}{e}$ & function \\
  & & $\appS{x_1}{x_2}$ & $\appC{x_1}{x_2}$ & application \\
  & & $\tickS{q}$ & $\tickC{q}$ & cost \\
  & & $\letS{e_1}{x}{e_2}$ & $\letC{e_1}{x}{e_2}$ & definition \\
  & & $\shareS{x}{x_1}{x_2}{e}$ & $\shareC{x}{x_1}{x_2}{e}$ &  sharing \\
  & & $\flipS{p}{e_1}{e_2}$ & $\flipC{p}{e_1}{e_2}$ & coin flip \\
  & & $\probS{p}$ & $\probC{p}$ & probability \\
  & & $\sflipS{x}{e_1}{e_2}$ & $\sflipC{x}{e_1}{e_2}$ & symbolic flip
\end{tabular}
\caption{Syntax of the language\label{fig:syntax}}
\end{figure}

\paragraph{Elementary probability theory} We recount some essential concepts from elementary probability theory.
You can find more serious mathematical development of probabilities in textbooks on measure theory~\cite{book:Williams91,book:Billingsley12}.

Consider a random experiment.
Let $\Omega$ denote the set of all the possible outcomes, called the \emph{sample space}.
A discrete \emph{probability space} is a pair $(\Omega,\bbP)$, where $\bbP : \Omega \to [0,1]$ is a \emph{probability distribution} on $\Omega$, i.e., $\sum_{\omega \in \Omega} \bbP(\omega) = 1$.
The probability of an \emph{event} $E \subseteq \Omega$, written $\bbP(E)$, is defined as $\sum_{\omega \in E} \bbP(\omega)$.
We often write $\bbP(\theta)$ for the probability of a statement $\theta$, i.e., $\bbP(\{ \omega \mid \theta(\omega)~\text{is}~\text{true}\})$.
A \emph{random variable} $X:\Omega \to \bbR \cup \{{-\infty},{+\infty}\}$ is a function from a probability space to the extended real numbers.
The \emph{expected value} of a random variable $X$ is the weighted average $\bbE_{\omega \sim (\Omega,\bbP)}(X) \defeq \sum_{\omega \in \Omega} X(\omega) \cdot \bbP(\omega)$.
We often write $\bbE(X)$ if there is no ambiguity in the choice of the probability space.
An important property of expected value is \emph{linearity}:
If $X$ and $Y$ are random variables and $a,b\in \bbR$, then $(aX + bY)$ is a random variable and $\bbE(aX + bY) = a\bbE(X) + b\bbE(Y)$.

\paragraph{Obstacles for probabilistic semantics} 

To define the expected resource usage of probabilistic programs, we formulate a cost semantics based on an evaluation dynamics. This turns out to be challenging. Previous work on AARA cost semantics for non-probabilistic programs lack the infrastructure to reason about certain effects of probabilistic phenomena. One such example is the poor behaviour of high-water marks: the well-known probabilistic Martingale betting strategies have an unbounded expected high-water mark, even while having finite expected net gain. In this section we describe the sorts of problems faced from the perspective of the cost semantics.

The notion of values in the cost semantics can proceed unchanged: \emph{value} $v \in \mathsf{Val}$ is either a null tuple $\unitV$, an empty list $\nilV$, a cons list $\consV{v_1}{v_2}$, or a \emph{function closure} $\cloV{V}{f}{x}{e}$ that consists of an \emph{environment} $V : \mathsf{Var} \to \mathsf{Val}$ and a function definition $\funC{f}{x}{e}$. However, the evaluation cost dynamics surrounding such values must be altered to deal with probability. In prior work on AARA~\cite{APLAS:HH10}, the cost semantics is defined by a judgment of the form
\begin{align*}
\evalrelnotr{V}{e}{v}{(q,q')}{}
\end{align*}
This judgment means that, under an evaluation environment $V$, the expression $e$ evaluates to the value $v$ using a high-water mark of $q \in \bbQ_{\ge 0}$ resources and leaving $q' \in \bbQ_{\ge 0}$ resources leftover. By tracking both the high-water mark and leftover resources, non-probabilistic AARA was able to reason about resources that might be returned after use, like space. This tracking is performed by the \emph{resource monoid}~\cite{APLAS:HH10}, which algebraically composes the high-water mark/leftover pairs $(q,q')$. 

Unfortunately, this operational judgment does not adapt to the probabilistic domain. Firstly, it distinguishes a particular value $v$ for evaluation, rather than a distribution. Further, the resource monoid does not compose under probability. Both points must be remedied to soundly model cost.

To illustrate the resource monoid problem, we first define it. The following accounts for how the high-water mark and leftover resource constraints change under non-probabilistic composition.
$$(a,b)\cdot(c,d) \defeq (a+\max(c-d,0),d+\max(d-c,0))$$

Now consider the following two expressions. Letting $e_1$ have the associated resource monoid term $(0,0)$, $e_2$ have $(4,2)$ and $e_3$ have $(2,1)$, we see both have an expected high-water mark resource usage of 2 and expected leftover of 1. 
$$\flipC{\frac 1 2}{e_1}{e_2} \;\;\; \textit{vs} \;\;\; e_3$$

However, we cannot represent the resource constraints of both expressions uniformly with the term $(2,1)$. This becomes apparent if we precede both expressions by a copy of $e_3$, as then the expected high-water mark of each differs. The latter can be correctly calculated to be 3 with $(2,1)\cdot(2,1)=(3,1)$. However, the high-water mark of the former would be 3.5 since half the time it would be 2 and half the time 5. There is no way to get two different results as a function of the same input $(2,1)$, so two-place resource monoid terms cannot be salvaged for probabilistic use.

To avoid this problem, we forgo the high-water mark/leftover resource distinction, and reason only about resources that \emph{monotonically} decrease, like time. 
It then suffices to track only the net cost with a more well-behaved one-place term. As a result, the AARA system described here only consumes resources, and never provides them.

\begin{changebar}
This restriction to monotonically consumed resources solves an additional problem for the cost semantics concerning the well-definedness of expected cost in the presence of nontermination. 
%
Even programs with finite expected cost may have nonterminating executions.
However, if the execution can be non-terminating, there can be an infinite number of execution traces, and thus the expected value of their cost is defined over an infinite sum. Such a sum must converge absolutely to represent an expected value, and if the costs for operations can have different signs this is not clearly the case.
Recent work~\cite{PLDI:WFG19} has proposed techniques to reason about non-monotone resources for imperative programs; adapting these techniques to analyze functional programs is beyond the scope of this paper, but is an interesting future research direction.
\end{changebar}

Besides the cost, a probabilistic semantics must also account for probabilistic execution resulting in a distribution of values, rather than 1 particular value. To solve this problem, one might first think to reason about individual executions separately by adding a component that tracks the probability of a particular value resulting. By collecting such judgments with probabilities adding to 1, one could then recover the desired value distributions.
For this approach, one might create the judgment $\evalrelnotr{V}{e}{v}{q}{p}$ which would mean that there exists \emph{an} execution where the expression $e$ evaluates to the value $v$ with net cost $q$ and probability $p$.
However, this approach has a subtle problem: There might be multiple different executions with the same evaluation result, cost, and probability.
For example, consider the following program
\[
e \equiv \flipC{\frac{1}{2}}{\tickC{2}}{\letC{\tickC{1}}{\_}{\tickC{1}}}.
\]
Although the program has two possible syntactically-distinct executions, there is only one valid evaluation relation derivable from the given rules, which is
\[
\evalrelnotr{\cdot}{e}{\unitC}{2}{\sfrac{1}{2}}.
\]
This thwarts the idea of collecting relations with probabilities summing to 1, as some relations would need to be counted multiple times, and the present components to the judgment leave no way to determine the multiplicity. To solve these problems, we present the following cost semantics.

\paragraph{Trace-based cost semantics} We deal with obstacles surrounding cost semantics by adapting Borgstr{\"o}m et al.'s trace-based semantics for lambda calculus~\cite{ICFP:BLG16} to our setting.
The key observation is that an execution is uniquely determined by the \emph{trace} of outcomes of the coin flips in the execution.
We augment the evaluation relation with a component for traces, i.e., a finite sequence of elements in $\{\mathsf{H},\mathsf{T}\}$.
The trace-based evaluation judgment then has the form
\[
\evalrel{V}{e}{v}{q}{p}{\sigma},
\]
The intuitive meaning is that under the environment $V$, with a sequence $\sigma$ of coin-flip outcomes, the expression $e$ evaluates to a value $v$ with cost $q$ and probability $p$.

\cref{Fi:TraceSemantics} presents the rules for this trace-based evaluation dynamics.
We write $[]$ for empty traces, $\sigma_1 \mathbin{@} \sigma_2$ for trace concatenation, and $\mathsf{H} \dblcolon \sigma$ or $\mathsf{T} \dblcolon \sigma$ to observe a new coin flip and prepend the outcome to $\sigma$.
In the rule \textsc{E:Let}, we multiply the probabilities of an execution of $e_1$ and an execution of $e_2$, as well as concatenate their traces of coin flips.

\begin{figure}\footnotesize
\fbox{$\evalrel{V}{e}{v}{q}{p}{\sigma}$ \quad ``in environment $V$, with trace $\sigma$, expression $e$ evaluates to value $v$ with cost $q$ and probability $p$''}
\begin{mathpar}
  \Rule{E:Var}{ }{ \evalrel{V}{x}{V(x)}{0}{1}{[]} }  
  \and
  \Rule{E:Triv}{ }{ \evalrel{V}{\unitS}{\unitV}{0}{1}{[]} }
  \and
  \Rule{E:Nil}{ }{ \evalrel{V}{\nilS}{\nilV}{0}{1}{[]} }
  \and
  \Rule{E:Cons}{ V(x_1) = v_1 \\ V(x_2) = v_2 }{ \evalrel{V}{\consS{x_1}{x_2}}{\consV{v_1}{v_2}}{0}{1}{[]} }
  \and
  \Rule{E:MatL-1}{ V(x) = \nilV \\ \evalrel{V}{e_0}{v}{q}{p}{\sigma} }{ \evalrel{V}{\matLS{x}{e_0}{x_1}{x_2}{e_1}}{v}{q}{p}{\sigma} }
  \and
  \Rule{E:MatL-2}{ V(x) = \consV{v_1}{v_2} \\ \evalrel{V,x_1 \mapsto v_1,x_2 \mapsto v_2}{e_1}{v }{q}{p}{\sigma} }{ \evalrel{V}{\matLS{x}{e_0}{x_1}{x_2}{e_1}}{v}{q}{p}{\sigma} }
  \and
  \Rule{E:Tick}{ }{ \evalrel{V}{\tickS{q}}{\unitV}{q}{1}{[]} }
  \and
  \Rule{E:Let}{ \evalrel{V}{e_1}{v_1}{q_1}{p_1}{\sigma_1} \\  \evalrel{V,x\mapsto v_1}{e_2}{v_2}{q_2}{p_2}{\sigma_2}  }{ \evalrel{V}{\letS{e_1}{x}{e_2}}{v_2}{q_1+q_2}{p_1 \cdot p_2}{\sigma_1 \mathbin{@} \sigma_2} }
  \and
  \Rule{E:Fun}{ }{ \evalrel{V}{\funS{f}{x}{e}}{\cloV{V}{f}{x}{e}}{0}{1}{[]} }
  \and
  \Rule{E:App}{ V(x_1) = \cloV{V'}{f}{x}{e} \\ V(x_2) = v_2 \\ \evalrel{V',f\mapsto \cloV{V'}{f}{x}{e},x\mapsto v_2}{e}{v}{q}{p}{\sigma}  }{ \evalrel{V}{\appS{x_1}{x_2}}{v}{q}{p}{\sigma} }
  \and
  \Rule{E:Flip-1}{ \evalrel{V}{e_1}{v_1}{q_1}{p_1}{\sigma} }{ \evalrel{V}{\flipS{p}{e_1}{e_2}}{ v_1}{q_1}{p \cdot p_1}{\mathsf{H} \dblcolon \sigma} }
  \and
  \Rule{E:Flip-2}{ \evalrel{V}{e_2}{v_2}{q_2}{p_2}{\sigma} }{ \evalrel{V}{\flipS{p}{e_1}{e_2}}{v_2}{q_2}{(1-p)\cdot p_2}{ \mathsf{T} \dblcolon \sigma} }
  \and
  \Rule{E:Share}{ V(x) = v \\ \evalrel{V,x_1 \mapsto v,x_2 \mapsto v}{e}{v'}{q}{p}{\sigma} }{ \evalrel{V}{\shareS{x}{x_1}{x_2}{e}}{v'}{q}{p}{\sigma} }
  \and
  \Rule{E:Prob}
  { }
  { \evalrel{V}{\probS{p}}{\probV{p}}{0}{1}{[]} }
  \and
  \Rule{E:FlipS-1}
  { V(x) = \probV{p} \\ \evalrel{V}{e_1}{v_1}{q_1}{p_1}{\sigma} }
  { \evalrel{V}{\sflipS{x}{e_1}{e_2}}{v_1}{q_1}{p \cdot p_1}{\mathsf{H} \dblcolon \sigma} }
  \and
  \Rule{E:FlipS-2}
  { V(x) = \probV{p} \\ \evalrel{V}{e_2}{v_2}{q_2}{p_2}{\sigma} }
  { \evalrel{V}{\sflipS{x}{e_1}{e_2}}{v_2}{q_2}{(1-p) \cdot p_2}{\mathsf{T} \dblcolon \sigma} }
\end{mathpar}
\caption{Evaluation rules of the trace-based cost semantics}
\label{Fi:TraceSemantics}
\end{figure}

Recall that in order to reason about expected resource usage, we need a notion of \emph{probability distributions} over executions, and found that accounting for the multiplicity of operational judgments made this difficult.
With the trace-based dynamics, we can now capture all the terminating executions uniquely.
This is because the result value $v$, the net cost $q$, and the probability $p$, are determined uniquely by the environment $V$, the expressions $e$, and the trace of coin flips $\sigma$. 

By induction on the structure of expression $e$, we prove the lemma below.
\begin{lemma}\label{lem:uniqueness}
  For all $V$, $e$ and $\sigma$, there is at most one combination of $v$, $q$ and $p$ s.t. $\evalrel{V}{e}{v}{q}{p}{\sigma}$.
\end{lemma}

Therefore, for fixed $V$ and $e$, the set of all finite traces induces a ``distribution'' over terminating executions.
We can extract a ``distribution'' $\interp{e}^V_{\Downarrow}$ on values $v$ and costs $q$ as follows:
\[
\interp{e}^V_{\Downarrow}(v,q) \defeq
   \sum_{\sigma} p_\sigma \quad \text{where $\sigma$'s are finite traces satisfying $\evalrel{V}{e}{v}{q}{p_\sigma}{\sigma}$}.
\]
Note that if there are non-terminating executions with non-zero probabilities, the map defined above is a subprobability distribution in the sense that the probabilities do \emph{not} sum up to one.
In other words, the probability that $e$ diverges under environment $V$ is $(1-\sum_{(v,q)} \interp{e}^V_{\Downarrow}(v,q))$.

With this trace-based cost semantics in hand, we can finally define the expected cost of evaluating some terminating expression $e$ with variable bindings given the values of $V$. The expected cost is just the sum of costs $q$ weighted by probability $p$ over all execution traces $\sigma$.
\[
	\sum_{\sigma : \evalrel{V}{e}{v}{q}{p}{\sigma}} p \cdot  q = \sum_{V,v,q} \interp{e}^V_{\Downarrow}(v,q) \cdot q
\]

However, generalizing this definition for non-termination would be nontrivial. As probability is only countably additive, and the set of infinite traces of a non-terminating execution may be uncountable, the above sum could no longer be used. It would appear that a more complicated summation mechanism like integration over a cost density function would be required to deal with such divergence, and we do not deign to develop that here. Instead, to deal with this concern and others, the cost semantics will be revisited in \cref{sec:sound}. There we will do like \cite{ICFP:BLG16} and convert from trace-based to distribution-based semantics.

\section{Type System}
\label{sec:type}

In this section, we develop an AARA type system to carry out expected cost analysis for probabilistic programs.
To focus on the changes that probabilistic choice induces on the type system, we describe its action here in \textit{linear} AARA, where all potential functions are linear in terms of list sizes. In other work, potential functions have been expanded to cover polynomials~\cite{POPL:HAH11,ESOP:HH10} and exponentials~\cite{FoSSaCS:KH20}, but this extension to AARA is orthogonal to probabilistic choice. Indeed, we have carried over the implementation and soundness of probabilistic AARA to support multivariate-polynomial potential functions and user-defined datatypes without problem, which we use to perform analyses in \cref{sec:goat} and beyond. 

\begin{wrapfigure}{r}{6.5cm}\small
\begin{tabular}{@{\hspace{0.1em}}c@{\hspace{0.5em}}c@{\hspace{0.5em}}l@{\hspace{0.5em}}l@{\hspace{0.5em}}l@{\hspace{0.1em}}}
  & & \textbf{Abstract} & \textbf{Concrete} \\ \hline
  $\tau$ & $\Coloneqq$ & $\unitA$ & $\unitT$ & nullary product \\
  & & $\listA{A}$ & $\listT{\tau}{q}$ & list \\
  & & $\arrA{A}{B}$ & $\arrT{A}{B}$ & arrow \\
  & & $\probA{q_H}{q_T}$ & $\probT[q_H]{q_T}$ & probability \\
  $A,B$ & $\Coloneqq$ & $\annoA{\tau}{q}$ & $\annoT{\tau}{q}$ & potential
\end{tabular}
\caption{Syntax of the type system\label{fig:types}}
\end{wrapfigure}

\paragraph{Types and potentials}
\cref{fig:types} presents the types that are supported in linear AARA. Aside from usual types like the nullary $\unitT$ and binary product $\prodT{\tau_1}{\tau_2}$, there are three special types that have potential-related components. The first is the potential pairing $\annoT{\tau}{q}$, which represents storing a constant $q \in \bbQ_{\ge 0}$ units of potential alongside a value of type $\tau$. The second is the list type $\listT{\tau}{q}$---a compact representation of $\listA{\annoT{\tau}{q}}$---which represents a list with $q \in \bbQ_{\ge 0}$ units of potential per element. The combination is sufficient to express potential functions that are linear combinations of input list lengths and constants.
The last is the probability type $\probA{q_H}{q_T}$. As introduced in \cref{sec:overview}, it represents $q_H$ units of potential for head cases and $q_T$ units for tail cases after a coin flip.

Formally, the \emph{potential function} $\pot{\cdot}{\tau}$ or $\pot{\cdot}{A}$, which maps values of type $\tau$ or $A$ to non-negative rational numbers, is defined as follows.  
\begin{align*}
  \pot{\unitV}{\unitA}  & \defeq 0, & \pot{v}{\annoA{\tau}{q}} & \defeq \pot{v}{\tau} + q, \\
  \pot{\nilV}{\listA{A}} & \defeq 0, & \pot{\consV{v_1}{v_2}}{\listA{A}} & \defeq \pot{v_1}{A} + \pot{v_2}{\listA{A}}, \\
  \pot{\cloV{V}{f}{x}{e}}{\arrA{A}{B}} & \defeq 0, & \pot{\probV{p}}{\probA{q_H}{q_T}} & \defeq q_H \cdot p + q_T \cdot (1-p).
\end{align*}

From the inductive definition above, we can derive the following closed form for the potential of a list $\ell = [v_1,\cdots,v_n]$ with respect to a type $\listT{\tau}{q}$, which is linear in the length of the list $\ell$.
\begin{equation*}
  \pot{\ell}{\listT{\tau}{q}} =  q \cdot n + \sum_{i=1}^n \pot{v_i}{\tau}.
\end{equation*}

Note that these definitions leave potential as a function of both type and value. Different values of the same type may differ in their total potential. For instance, in the case of lists, one term in the above closed form for potential depends on the length of the list, so lists of differing lengths but the same type may differ in total potential. 

\paragraph{Static semantics}
The typing judgment for linear AARA the form $\typingrel{\Gamma}{q}{e}{A}$, the intuitive meaning of which is that the potential given by $\Gamma$ and $q$ is sufficient to cover the \emph{expected} evaluation cost of $e$ and the \emph{expected} potential of the evaluation result with respect to $A$.

As existing AARA type systems, our typing rules form an \emph{affine} linear type system, which ensures that every variable is used \emph{at most} once~\cite{kn:Walker02}.
\cref{fig:typing-rules} lists the typing rules.
It turns out that most of the rules coincide with those of non-probabilistic linear AARA systems.
This fact indicates that our type system is a conservative extension of non-probabilistic AARA for monotonic resources, and our type system is able to derive worst-case cost bounds for deterministic programs.

To understand the new rule \textsc{L:Flip} for probabilistic branching, consider the expression $\flipS{p}{e_1}{e_2}$, where $e_1$ requires $\Phi_1$ units of potential and $e_2$ requires $\Phi_2$. The evaluation of the flip expression should expect to require a weighted average of $\Phi_1$ and $\Phi_2$, specifically $p \cdot \Phi_1 + (1-p) \cdot \Phi_2$. This should be paid out of the typing context $\Gamma$ and constant potential $q$, both of which are shared between branches. The distribution of this sharing is expressed using a \emph{sharing relation} $\share{\tau}{\tau_1}{\tau_2}$, which apportions the potential indicated by $\tau$ into two parts to be associated with $\tau_1$ and $\tau_2$, alongside a \emph{potential-scaling} operation. We formally define these relations and prove they capture the correct intuition with \cref{lemma:split,lemma:scale}, but first we explain why the rule does \textit{not} also perform expected value calculations for the type $A$.

One might think that a similar weighted average could be used to combine the types of $e_1$ and $e_2$ to get \textit{expected} type $A$, rather than require both expressions have type $A$ exactly. Perhaps equally likely types $\listT{\tau}{3}$ and $\listT{\tau}{1}$ could convert to expected type $\listT{\tau}{2}$. However, the value produced in each branch might differ, and for lists of type $\listT{\tau}{q}$ total potential is a scalar $q$ of length; taking the expected value of the scalars without accounting for length does not succeed in finding the expected potential. Thus, the rule \textsc{L:Flip} cannot be made more permissive in that manner.

Nonetheless, note that the same return type for both branches in \textsc{L:Flip} still can leave differing potential after each branch, which is necessary for expected cost reasoning. For example, consider the following program where the function \code{append} requires 1 unit of potential per element in its first argument. 

\begin{lstlisting}
append (flip 0.5 | H ->  [1;2;3] | T ->  [0]) [5;6]
\end{lstlisting}

The return types of the two branches of the flip expression are the same ($\listT{\intT}{2}$), but the actual potential in the results is different: the heads branch returns a list with 3 units of potential, and the tails branch returns 1 unit. This shows that the analysis properly composes and correctly reasons that expected cost of the program is 2. This also works for symbolic lists and can derive the bound $|x|+|y|$ for the function 

\begin{lstlisting} 
fun x y -> append (flip 0.5 | H -> y | T -> (append x y)) [] 
\end{lstlisting}

Now we formalize the sharing and scaling relations.
The sharing relation for types is defined as follows. Note that the sharing relation is also used in \textsc{L:Share} to make ``copies'' of a variable, while ensuring that the total potential over copies is preserved.
\begin{mathpar}\footnotesize
  \Rule{Sh:Unit}{ }{ \share{\unitA}{\unitA}{\unitA} }
  \and
  \Rule{Sh:List}{ \share{A}{A_1}{A_2} }{ \share{\listA{A}}{\listA{A_1}}{\listA{A_2}} }
  \and
  \Rule{Sh:Arrow}{ }{ \share{\arrA{A}{B}}{\arrA{A}{B}}{\arrA{A}{B}} }
  \and
  \Rule{Sh:Prob}{ q_H = q_H^{(1)} + q_H^{(2)} \\ q_T = q_T^{(1)} + q_T^{(2)} }{ \share{ \probA{q_H}{q_T} }{ \probA{q_H^{(1)}}{q_T^{(1)}} }{ \probA{q_H^{(2)}}{q_T^{(2)}} } }
  \and
  \Rule{Sh:Pot}{ q=q_1+q_2 \\ \share{\tau}{\tau_1}{\tau_2} }{ \share{\annoA{\tau}{q}}{ \annoA{\tau_1}{q_1} }{ \annoA{\tau_2}{q_2} } }
\end{mathpar}
We extend the sharing relation to typing contexts, as it has previously only been used on a per-type basis. This splits the potential across all types in $\Gamma$ across 2 new contexts of the same base types.  
\begin{mathpar}\footnotesize
  \Rule{Sh:Empty}{ }{ \share{\cdot}{\cdot}{\cdot} }
  \and
  \Rule{Sh:Extend}{ \share{\Gamma}{\Gamma_1}{\Gamma_2} \\ \share{\tau}{\tau_1}{\tau_2} }{ \share{\Gamma,x:\tau}{\Gamma_1,x:\tau_1}{\Gamma_2,x:\tau_2} }
\end{mathpar}

Potential-scaling can be defined syntactically as follows.
Intuitively, $p \times \tau$ (resp., $p \times A$) produces a type with as much potential as that of the original type $\tau$ (resp., $A$) scaled by the factor $p$.
\begin{align*}
  p \times \unitA & \defeq \unitA, & p \times \annoA{\tau}{q} & \defeq \annoA{p \times \tau}{p \cdot q}, \\
  p \times \listA{A} & \defeq \listA{p \times A}, & p \times \arrA{A}{B} & = \arrA{A}{B}, \\
  p \times \probA{q_H}{q_T} & \defeq \probA{p \cdot q_H}{p \cdot q_T}. & 
\end{align*}
Also, we extend the scaling operation to typing contexts.
\begin{align*}
  p \times (\cdot) & \defeq \cdot, & p \times (\Gamma,x:\tau) & \defeq p \times \Gamma, x : (p \times \tau)  .
\end{align*}


By induction on the structure of value $v$, we prove the following lemmas that ensure the sharing and scaling relations are consistent with their intuitive meaning.

\begin{lemma}
\label{lemma:split}
    For any value $v$ of type $\tau$ (resp., $A$), if $\share{\tau}{\tau_1}{\tau_2}$ (resp., $\share{A}{A_1}{A_2}$), then $\pot{v}{\tau} = \pot{v}{\tau_1} + \pot{v}{\tau_2}$ (resp., $\pot{v}{A} = \pot{v}{A_1}+\pot{v}{A_2}$) .
\end{lemma}

\begin{lemma}\label{lemma:scale}
For any probability $p$ and value $v$ of type $\tau$ (resp., $A$),
  $\pot{v}{p \times \tau} = p \cdot \pot{v}{\tau}$ (resp., $\pot{v}{p \times A} = p \cdot \pot{v}{A}$).
\end{lemma}

We now discuss the rule \textsc{L:Prob} and \textsc{L:FlipS} for the new probability type $\probT[q_H]{q_T}$.
To type a probability encapsulation $\probS{p}$, we need $\pot{\probV{p}}{\probT[q_H]{q_T}} = p \cdot q_H + (1-p) \cdot q_T$ units of potential in the context to cover its expected value.
Then to type an expression $\sflipS{x}{e_1}{e_2}$ that flips a variable $x$ with type $\probT[q_H]{q_T}$, one might want to use the potential-scaling operation as the rule \textsc{L:Flip} does.
However, the probability $x$ here is \emph{symbolic}, thus we cannot define the scaling operation in linear AARA.\footnote{
In our implementation of pRaML, we use multivariate polynomial AARA to support symbolic scaling, which unifies the two distinct flip operations presented here. We also incorporate the ability to multiply and complement symbolic probabilities. 
}
The rule \textsc{L:FlipS} avoids the problem by forcing $e_1$ and $e_2$ to be typed under the same context, appealing to the equality $\Phi = x \cdot \Phi + (1-x) \cdot \Phi$ for any $x$ and $\Phi$.
Note that it assigns $q_H$ units of potential to type $e_1$, and $q_T$ units to type $e_2$; this assignment is sound because we pay $x \cdot q_H + (1-x) \cdot q_T$ to create $x : \probT[q_H]{q_T}$.

Finally, we briefly explain other typing rules.
In the rule \textsc{L:Cons}, we have to provide potential $p$ to account for the potential of the new list element. Conversely, the potential of the head $x_1$ of the list $x : \listT{\tau}{p}$ becomes available in the cons branch of the pattern match in the rule \textsc{L:MatL}. As a result, we have constant potential $p + q$ available when typing $e_1$.
In the rule \textsc{L:App}, we require that we have the exact potential annotations ($x_2:\tau$ and $q$) that are required by the argument. The resulting potential is given by the result type $B$.
In the rule \textsc{L:Fun} for (recursive) function abstraction, we
require that the potential of the variables captured in the context
$\Gamma$ is zero. We write $\zero{\Gamma}$ for the context $\Gamma$ in
which every potential annotation $q$ is replaced by $0$. This is
formally defined below. The reason for this requirement is that we
allow functions to be used an arbitrary number of times (recall the
definition of sharing).
If $\Gamma$ would carry potential then we could use
this potential multiple times to account for cost, which is not sound. 
Since functions do not carry potential, we do not have to restrict the
type of the recursively defined function $f$ in a similar way.
An alternative
would be to remove the premise $\isZero{\Gamma}$ and to treat
functions in an affine way.
\begin{align*}
  \zero{\unitA} & \defeq \unitA, &  \zero{\annoA{\tau}{q}} & \defeq \annoA{\zero{\tau}}{0}, \\
  \zero{\listA{A}} & \defeq \listA{\zero{A}}, & \zero{\arrA{A}{ B}} & \defeq \arrA{A}{ B},\\
  \zero{\probA{q_H}{q_T}} & \defeq \probA{0}{0}. &
\end{align*}
Note that for function types, we do not have to recursively eliminate potential with $|\cdot |$
since the potential of a function is already $0$. 
The definition is then lifted point-wise to annotated contexts $\Gamma$.
\begin{align*}
  \zero{\cdot} & \defeq \cdot, & \zero{\Gamma, x:\tau} & \defeq  \zero{\Gamma}, x:\zero{\tau}.
\end{align*}

The structural rules \textsc{L:Sub}, \textsc{L:Sup}, \textsc{L:Weak}, and \textsc{L:Relax} can be applied to every expression.
The weakening rule \textsc{L:Weak} is standard. However, there is
another form of weakening: The rule \textsc{L:Relax}, states that,
given a judgment $\Gamma; p \vdash e : \annoT{\tau}{p'}$, we can also
have more potential $q$ in the context and give up some of the
potential $p'$. Additionally, the rule also covers the case in which
we pass through additional potential $c \geq 0$ yielding the judgment
$\Gamma; p+c \vdash e : \annoT{\tau}{p'+c}$.
The subtyping rules \textsc{L:Sub} and \textsc{L:Sup} enable
us to relax the potential requirements for potential in data
structures in the same way as \textsc{T:Relax} does for constant
potential. The subtyping relation for types is defined by the
following rules.
  \begin{mathpar}\footnotesize
    \inferrule
    {
    }
    { \unitA \subtype \unitA }
    \and
    \inferrule
    { A \subtype B }
    { \listA{A} \subtype \listA{B} }
    \\
    \inferrule
    { A_2 \subtype A_1 \\ B_1 \subtype B_2
    }
    { {\arrA{A_1}{ B_1}} \subtype {\arrA{A_2}{ B_2}} }
    \and
    \inferrule
    { q_H^{(1)} \ge q_H^{(2)} \\ q_T^{(1)} \ge q_T^{(2)} }
    { \probA{q_H^{(1)}}{q_T^{(1)}} \subtype \probA{q_H^{(2)}}{q_T^{(2)}} }
    \and
    \inferrule
    { q_1 \geq q_2
    \\ \tau_1 \subtype \tau_2
    }
    { \annoA{\tau_1}{q_1} \subtype \annoA{\tau_2}{q_2}}
\end{mathpar}

By induction on the structure of value $v$ followed by inversion on the subtyping judgment, we prove the following lemma.
\begin{lemma}
  If $\tau \subtype \tau'$ then
  $\Phi(v:\tau') \leq \Phi(v:\tau)$ for any value $v$ of type $\tau$.
\end{lemma}

\begin{figure}\footnotesize
\fbox{$\typingrel{\Gamma}{q}{e}{A}$ \quad ``in context $\Gamma$ with constant potential $q$, expression $e$ has potential-annotated type $A$''}
\begin{mathpar}
    \Rule{L:Var}
    { 
    }
    {x : \tau;0 \vdash x : \annoT{\tau}{0}}

    \Rule{L:Unit}
    { 
    }
    {\cdot;0 \vdash \unitS : \annoT{\unitA}{0}}
    
    \Rule{L:Nil}
    {
    }
    { \cdot; 0 \vdash \nilS{} : \annoT{\listA{A}}{0}
    }

    \Rule{L:Cons}
    { A = \annoT{\tau}{p}
    }
    { x_1: \tau, x_2: \listA{A} ; p \vdash \consS{x_1}{x_2} : \annoT{\listA{A}}{0}
    }

    \Rule{L:MatL}
    { A = \annoT{\tau}{p}
    \\ \Gamma; q \vdash e_0 : B
    \\ \Gamma,x_1:\tau,x_2:\listA{A}; q+p \vdash e_1 : B
    }
    { \Gamma,x:\listA{A}; q \vdash \matLS{e_0}{x_1}{x_2}{e_1}{x} : B
    }
    
    \Rule{L:Tick}
    {
    }
    { \cdot; q \vdash \tickS{q} : \annoT{\unitA}{0}
    }
    
    \Rule{L:Let}
    { \Gamma_1; q \vdash e_1 : \annoT{\tau}{p}
    \\ \Gamma_2, x:\tau ; p \vdash e_2 : B
    }
    { \Gamma_1,\Gamma_2; q \vdash \letS{e_1}{x}{e_2} : B
    }

    \Rule{L:App}
    {  
    A = \annoT{\tau}{q}
    }
    { x_1: \arrA{ A }{ B }, x_2:\tau ;q \vdash \appS{x_1}{x_2} : B}

    \Rule{L:Fun}
    { A = \annoT{\tau}{q}
    \\ \isZero{\Gamma}
    \\\\ \Gamma, f: \arrA{A}{B}, x:\tau; q \vdash e : B
    }
    { \Gamma;0 \vdash \funS{f}{x}{e} : \annoT{\arrA{A}{B}}{0}
    }

    \Rule{L:Share}
    { \share{\tau}{\tau_1}{\tau_2}
    \\ \Gamma,x_1:\tau_1,x_2:\tau_2; q \vdash e : B
    }
    {\Gamma, x:\tau; q \vdash \shareS{x}{x_1}{x_2}{e} : B}
    
    \Rule{L:Flip}{ \share{\Gamma}{p \times \Gamma_1}{(1-p) \times \Gamma_2} \\ q = p \cdot q_1 + (1-p) \cdot q_2 \\ \typingrel{\Gamma_1}{q_1}{e_1}{A} \\ \typingrel{\Gamma_2}{q_2}{e_2}{A} }{ \typingrel{\Gamma}{q}{\flipS{p}{e_1}{e_2}}{A} }
    
    \Rule{L:Prob}
    { q = p \cdot q_H + (1-p) \cdot q_T
    }
    { \typingrel{\cdot}{ q }{\probS{p}}{\annoT{\probA{q_H}{q_T}}{0}} }
    
    \Rule{L:FlipS}
    { \typingrel{\Gamma}{q+q_H}{e_1}{A} \\
    \typingrel{\Gamma}{q+q_T}{e_2}{A}
    }
    { \typingrel{\Gamma,x:\probA{q_H}{q_T}}{q}{\sflipS{x}{e_1}{e_2}}{A} }
    
    \\
    
    \Rule{L:Sub}
    {\Gamma;q \vdash e : \annoT{\tau'}{q'}
    \\ \tau' \subtype \tau
    }
    {\Gamma;q \vdash e : \annoT{\tau}{q'} }

    \Rule{L:Sup}
    {\Gamma,x:\tau; q \vdash e : B
    \\ \tau' \subtype \tau
    }
    {\Gamma,x:\tau'; q \vdash e : B}
    
    \\

    \Rule{L:Weak}
    {\Gamma; q \vdash e : B
    }
    {\Gamma,x:\tau; q \vdash e : B}

    \Rule{L:Relax}
    { \Gamma; p \vdash e : \annoT{\tau}{p'}
    \\ q \geq p
    \\ q - q' \geq p - p'
    }
    {\Gamma; q \vdash e : \annoT{\tau}{q'}}
\end{mathpar}
\caption{Typing rules\label{fig:typing-rules}}
\end{figure}

\paragraph{Example}
To illustrate the type system in action, we apply it to a random-walk program in concrete syntax below.
Consider the function \code{brdwalk} which performs a biased random walk over the length of its input list, stopping whenever the list is empty. With $\frac 3 4$ probability of shrinking the list and $\frac 1 4$ of growing it, we expect the list length to shrink by $\frac 1 2$ per iteration of the walk. Thus, we expect a stopping time of twice the input list's length.
\begin{eqnarray*}
  \code{brdwalk} & \equiv & \word{fun } f \ell = \\
  & & \word{case} \ell\;\{ \nilC \hookrightarrow \unitC \\
  & & \hphantom{\word{case} \ell\;\{} \consC{\_}{x_2} \hookrightarrow \word{let } \_ = \tickC{1} \word{ in }  \word{flip} \sfrac{3}{4}\;\{  \word{H} \hookrightarrow \appC{f}{x_2} \mid \word{T} \hookrightarrow \appC{f}{\consC{\unitC}{\consC{\unitC}{x_2}}}  \}  \} 
\end{eqnarray*}
We now derive the type $\arrT{\annoT{\listT{\unitT}{2}}{0}}{\annoT{\unitT}{0}}$
for \code{brdwalk}, which indicates using twice the input list's length for initial potential. This amount of potential provides an upper bound on expected stopping time which happens to be exact.

In the \emph{tails} case, we need four units of extra constant potential to construct the argument list.
\begin{mathpar}\footnotesize
  \Rule{L:App}{ \Rule{L:Cons}{ \Rule{L:Cons}{ \Rule{L:Var}{ }{ \typingrel{f:\arrT{\annoT{\listT{\unitT}{2}}{0}}{\annoT{\unitT}{0}},x_2:\listT{\unitT}{2}}{0}{x_2}{\annoT{\listT{\unitT}{2}}{0}} } }{ \typingrel{f:\arrT{\annoT{\listT{\unitT}{2}}{0}}{\annoT{\unitT}{0}},x_2:\listT{\unitT}{2} }{2}{\consC{\unitC}{x_2}}{\annoT{\listT{\unitT}{2}}{0}} } }{ \typingrel{f:\arrT{\annoT{\listT{\unitT}{2}}{0}}{\annoT{\unitT}{0}},x_2:\listT{\unitT}{2}}{4}{\consC{\unitC}{\consC{\unitC}{x_2}}}{\annoT{\listT{\unitT}{2}}{0}} } }{ \typingrel{f:\arrT{\annoT{\listT{\unitT}{2}}{0}}{\annoT{\unitT}{0}},x_2:\listT{\unitT}{2}}{4}{\appC{f}{\consC{\unitC}{\consC{\unitC}{x_2}}}}{\annoT{\unitT}{0}} }
\end{mathpar}
Otherwise, if the coin flip shows \emph{heads}, the type derivation goes as follows:
\begin{mathpar}\footnotesize
  \Rule{L:App}{ }{ \typingrel{f:\arrT{\annoT{\listT{\unitT}{2}}{0}}{\annoT{\unitT}{0}},x_2:\listT{\unitT}{2}}{0}{\appC{f}{x_2}}{\annoT{\unitT}{0}} }
\end{mathpar}
Via the definition of potential scaling, we find that
\[
\share{\listT{\unitT}{2}}{\sfrac{3}{4} \times \listT{\unitT}{2} }{ (1-\sfrac{3}{4}) \times \listT{\unitT}{2}} \quad \text{and} \quad 1 = \sfrac{3}{4} \cdot 0 + \sfrac{1}{4} \cdot 4.
\]
Then we apply the rule \textsc{L:Flip}, deriving the desired type $\arrT{\annoT{\listT{\unitT}{2}}{0}}{\annoT{\unitT}{0}}$ for \code{brdwalk}.
\begin{mathpar}\footnotesize
  \Rule{L:Flip}{ \typingrel{f:\arrT{\annoT{\listT{\unitT}{2}}{0}}{\annoT{\unitT}{0}}, x_2:\listT{\unitT}{2} }{0}{\appC{f}{x_2}}{\annoT{\unitT}{0}}  \\
  \typingrel{f:\arrT{\annoT{\listT{\unitT}{2}}{0}}{\annoT{\unitT}{0}},x_2:\listT{\unitT}{2}}{4}{\appC{f}{\consC{\unitC}{\consC{\unitC}{x_2}}}}{\annoT{\unitT}{0}}
   }{ \typingrel{f:\arrT{\annoT{\listT{\unitT}{2}}{0}}{\annoT{\unitT}{0}}, x_2:\listT{\unitT}{2} }{1}{\flipC{\sfrac{3}{4}}{\appC{f}{x_2}}{\appC{f}{\consC{\unitC}{\consC{\unitC}{x_2}}}}}{\annoT{\unitT}{0}} }
\end{mathpar}

\section{Soundness}
\label{sec:sound}

In this section, we formalize our intuition that our type system derives expected cost bounds and sketch a soundness proof (\cref{the:soundness}).
We also study nontrivial non-termination behavior of probabilistic programs, and prove a stronger result (\cref{the:soundness:improved}) which implies that derived expected bounds on resources like time imply that the analyzed program terminates with probability one (\cref{cor:ast}).
Proofs are included in \cref{sec:appendix:proofs}.

\paragraph{Values}
Before we can state the theorem, we need to properly extend the definition of potential to typing contexts and evaluation environments.
We introduce a type judgment $\valuerel{v}{\tau}$ (or $\valuerel{v}{A}$) for values, which is defined in \cref{fig:value-typing}.
This relation ignores potential annotations and checks only the values are well-typed.
An evaluation environment $V$ is said to have type context $\Gamma$, written $V : \Gamma$, if for all $x$ bound in $\Gamma$, we have $\valuerel{V(x)}{\Gamma(x)}$.
The most interesting rule is the
rule \textsc{V:Fun} for function closures. It uses the type
rule \textsc{L:Fun} for expressions and existentially quantifies over the context $\Gamma$. This rule ensures that we only consider functions that are well-formed with respect to the type system, which is necessary to prove the soundness of the analysis.

\begin{figure}\footnotesize
\fbox{$\valuerel{v}{\tau}$ (or $\valuerel{v}{A}$) \quad ``value $v$ has type $\tau$ or $A$''}
\begin{mathpar}
 \Rule{V:Unit}
 { }
 { \unitV : \unitA}
 
 \Rule{V:Prob}
 { }
 { \probV{p} : \probA{q_H}{q_T} }

 \Rule{V:Nil}
 { }
 { \nilV : \listA{A} }

 \Rule{V:Cons}
 { v_1 : A 
 \\ v_2 : \listA{A}
 }
 { \consV{v_1}{v_2} : \listA{A} }

 \Rule{V:Anno}
 { v: \tau }
 { v : \annoA{\tau}{q} }

 \Rule{V:Fun}
 {A = \annoA{\tau}{q}
    \\ V : \Gamma
    \\\\ \zero{\Gamma}, f: \arrA{A }{ B}, x:\tau; q \vdash e : B
 }
 { \cloV{V}{f}{x}{e} : \arrA{A }{ B} }
  \end{mathpar}
  \caption{Typing rules for values\label{fig:value-typing}}  
\end{figure}

Let $V : \Gamma$. We define the potential of $V$ with respect to $\Gamma$ as follows.
\begin{equation*}
  \pot{V}{\Gamma} \defeq \sum_{x \in \mathrm{dom}(\Gamma)} \pot{V(x)}{\Gamma(x)}.
\end{equation*}

\paragraph{A first attempt} With the trace-based evaluation dynamics, we might state the soundness theorem for probabilistic programs as follows. Intuitively, it says that the initial potential is sufficient to pay for the expected evaluation cost and the typing of the result. 

  Let $\typingrel{\Gamma}{q}{e}{A}$ and $\valuerel{V}{\Gamma}$.
  Then
  \[
  \pot{V }{ \Gamma} + q \ge \sum_{\sigma_0 : \evalrel{V}{e}{v_0}{q_0}{p_0}{\sigma_0}} p_0 \cdot  ( \pot{v_0 }{ A} + q_0).
  \]
  Note that the summation is taken over traces $\sigma_0$, and by \cref{lem:uniqueness}, the tuple $(p_0,v_0,q_0)$ is uniquely determined by $V$, $e$, and $\sigma_0$.
However, it is unclear how to prove the theorem by induction on the evaluation judgment.
The reason is that we now have to deal with a collection of evaluation judgments, instead of one.
Intuitively, the trace-based evaluation dynamics talks about \emph{individual} executions, while the goal of our resource analysis for probabilistic programs is to reason about \emph{aggregated} information over all possible executions.
We therefore develop another evaluation dynamics that deals with \emph{distributions} of executions more directly, and show that it agrees with our previous semantics.

First we illustrate why a naive approach here will not work. One might start with a new judgment $\steprel{V}{e}{\mu}{}$ where $\mu$ is a distribution over pairs $(v,q)$, $v$ is the evaluation result, and $q$ is the net cost.
Then one might use the following rule for composition under probabilistic branching.
\[
\Rule{Bad:Flip}{ \steprel{V}{e_1}{\mu_1}{} \\ \steprel{V}{e_2}{\mu_2}{} }{ \steprel{V}{\flipS{p}{e_1}{e_2}}{p \cdot \mu_1 + (1-p) \cdot \mu_2}{} }
\]
Here, we denote the weighted sum of two distributions $\mu_1$ and $\mu_2$ by $p \cdot \mu_1 + (1-p) \cdot \mu_2$, defined as $\lambda\omega. p \cdot \mu_1(\omega) + (1-p) \cdot \mu_2(\omega)$.

For the leaf cases, such as unit values, one might then introduce the rule where $\delta(\omega) = \lambda\omega'.[\omega = \omega']$ denotes the \emph{point distribution} on $\omega$, and where the \emph{Iverson brackets} $[\cdot]$ are defined by $[\varphi]=1$ if $\varphi$ is true and otherwise $[\varphi]=0$.
\[
\Rule{Bad:Triv}{ }{ \steprel{V}{\unitS}{\delta(\unitV,0)}{} }
\]

However, the attempt does not work well for \emph{almost-sure} termination, i.e., terminating with probability 1.
The issue is that the inductive definition of such a distribution dynamics will fail if there is a non-terminating execution.
Consider the following program
\[
f \equiv \funC{f}{\_}{\flipC{\frac{1}{2}}{\unitC}{\appC{f}{\unitC}}}
\]
and suppose that we want to derive an evaluation judgment for $\appC{f}{\unitC}$.
There does not exist a distribution $\mu$ such that $\steprel{V}{\appC{f}{\unitC}}{\mu}{}$, because if we try to apply the rules inductively, we will end up with a derivation tree with an infinite depth.
\begin{mathpar}\footnotesize
  \Rule{Bad:App}{ \Rule{Bad:Flip}{ \Rule{Bad:Triv}{ }{ \steprel{V}{\unitC}{\delta(\unitV,0)}{} } \\ \Rule{Bad:App}{ \vdots }{ \steprel{V}{\appC{f}{\unitC}}{???}{} } }{ \steprel{V}{\flipC{\sfrac{1}{2}}{\unitC}{\appC{f}{\unitC}}}{???}{} } }{ \steprel{V}{\appC{f}{\unitC}}{???}{} }
\end{mathpar}

\paragraph{Distribution-based semantics}
To cope with possible non-terminating executions, we develop a partial-evaluation-like dynamics equivalent to our trace-based one. Unlike partial evaluation dynamics used in AARA literature to deal with non-termination~\cite{APLAS:HH10}, we \emph{do} care about the evaluation results. For our new dynamics, we adapt the distribution-based semantics of~\cite{ICFP:BLG16}, which index judgments by their derivation depth to be able to construct a ``complete'' semantics from the ``partial'' ones. To this end, we need only modify the subprobability distributions to be over value-cost pairs, resulting in judgments of the following form: 
\[
\steprel{V}{e}{\mu}{n}.
\]
The meaning is that the expression $e$ reduces to a subprobability distribution with an at-most-$n$ derivation depth.
We use \textit{sub}probability distributions, whose probabilities sum to possibly less than one, because there could be terminating executions with a derivation tree whose depth is more than $n$.
%
%
\cref{Fi:DistributionSemantics} presents the rules for this distribution-based semantics.
In addition to the syntax-directed rules, we introduce a special base case where $n=0$ and $\mu$ is set to a zero distribution $\mathbf{0} \defeq \lambda\omega. 0$.


\begin{figure}\footnotesize
\fbox{$\steprel{V}{e}{\mu}{n}$ \qquad ``in environment $V$, expression $e$ reduces to result distribution $\mu$ within $n$ steps}
\begin{mathpar}
  \Rule{DE:Base}{ }{ \steprel{V}{e}{\mathbf{0}}{0} }
  \and
  \Rule{DE:Var}{ n > 0 }{ \steprel{V}{x}{\delta(V(x),0)}{n} }
  \and
  \Rule{DE:Triv}{n > 0 }{ \steprel{V}{\unitS}{\delta(\unitV,0)}{n} }
  \and
  \Rule{DE:Nil}{ n > 0 }{ \steprel{V}{\nilS}{\delta(\nilV,0)}{n} }
  \and
  \Rule{DE:Cons}{ n > 0 \\ V(x_1) = v_1 \\ V(x_2) = v_2}{ \steprel{V}{\consS{x_1}{x_2}}{\delta(\consV{v_1}{v_2},0)}{n} }
  \and
  \Rule{DE:MatL-1}{ V(x) = \nilV \\ \steprel{V}{e_0}{\mu}{n} }{ \steprel{V}{\matLS{e_0}{x_1}{x_2}{e_1}{x}}{\mu}{n+1} }
  \and
  \Rule{DE:MatL-2}{ V(x) = \consV{v_1}{v_2} \\\\ \steprel{V,x_1\mapsto v_1,x_2\mapsto v_2}{e_1}{\mu}{n} }{ \steprel{V}{\matLS{e_0}{x_1}{x_2}{e_1}{x}}{\mu}{n+1} }
  \and
  \Rule{DE:Tick}{ n > 0 }{ \steprel{V}{\tickS{q}}{\delta(\unitV,q)}{n} }
  \and
  \Rule{DE:Fun}{ n > 0}{ \steprel{V}{\funS{f}{x}{e}}{\delta(\cloV{V}{f}{x}{e},0 )}{n} }
  \and
  \Rule{DE:App}{ V(x_1) = \cloV{V'}{f}{x}{e} \\ V(x_2) = v_2 \\ \steprel{V',f\mapsto \cloV{V'}{f}{x}{e},x\mapsto v_2}{e}{\mu}{n} }{ \steprel{V}{\appS{x_1}{x_2}}{\mu}{n+1} }
  \and
  \Rule{DE:Let}{ \steprel{V}{e_1}{\mu}{n} \\ \Forall{(v_1,q_1) \in \mathrm{supp}(\mu)} \steprel{V,x\mapsto v_1}{e_2}{\mu_{(v_1,q_1)}}{n} }{ \steprel{V}{\letS{e_1}{x}{e_2}}{  \textstyle\sum_{(v_1,q_1) } \sum_{(v_2,q_2) } \mu(v_1,q_1) \cdot \mu_{(v_1,q_1)}(v_2,q_2) \cdot \delta(v_2,q_1+q_2) }{n+1} }
  \and
  \Rule{DE:Share}{ V(x) = v \\\\ \steprel{V,x_1\mapsto v, x_2\mapsto v}{e}{\mu}{n} }{ \steprel{V}{\shareS{x}{x_1}{x_2}{e}}{\mu}{n+1} }
  \and
  \Rule{DE:Flip}{ \steprel{V}{e_1}{\mu_1}{n} \\ \steprel{V}{e_2}{\mu_2}{n} }{ \steprel{V}{\flipS{p}{e_1}{e_2}}{p \cdot \mu_1 + (1-p) \cdot \mu_2}{n+1} }
  \and
  \Rule{DE:Prob}
  { n > 0 }
  { \steprel{V}{\probS{p}}{\delta( \probV{p}, 0 )}{n} }
  \and
  \Rule{DE:FlipS}
  { V(x) = \probV{p} \\ \steprel{V}{e_1}{\mu_1}{n} \\ \steprel{V}{e_2}{\mu_2}{n} }
  { \steprel{V}{\sflipS{x}{e_1}{e_2}}{ p \cdot \mu_1 + (1-p) \cdot \mu_2 }{ n+1} }
\end{mathpar}
\caption{Evaluation rules of the distribution-based cost semantics}
\label{Fi:DistributionSemantics}
\end{figure}

We can now approximate the distribution over terminating executions using the depth-indexed distributions by making use of the following lemma.

\begin{lemma}\label{lem:dist-wcpo}
  If $\steprel{V}{e}{\mu_1}{n}$, $\steprel{V}{e}{\mu_2}{m}$ and $n \le m$, then $\mu_1 \le \mu_2$ pointwise.
  As a consequence, we can define $\interp{e}^V_{\Rightarrow} \defeq \sup\{ \mu_n : \steprel{V}{e}{\mu_n}{n}\} = \lim_{n \to \infty} \mu_n$ as the subprobability distribution of all possible terminating executions of a probabilistic program $e$ under environment $V$.
\end{lemma}
\begin{proof}
  By induction on the derivation of $\steprel{V}{e}{\mu_2}{m}$, followed by inversion on $\steprel{V}{e}{\mu_1}{n}$.
  The existence of the sequence appeals to the Monotone Convergence Theorem.
\end{proof}

Recall the problem case from attempting a non-indexed distribution-based operational semantics:
\[
f \equiv \funC{f}{\_}{\flipC{\frac{1}{2}}{\unitC}{\appC{f}{\unitC}}}
\]
With the depth-indexed distribution-based dynamics, we can now derive the following judgments:
\begin{align*}
    & \steprel{V}{\appC{f}{\unitC}}{\mathbf{0}}{0}, & & \steprel{V}{\appC{f}{\unitC}}{\sfrac{1}{2} \cdot \delta(\unitV, 0)}{3}, \\
    & \steprel{V}{\appC{f}{\unitC}}{\sfrac{1}{2} \cdot \delta(\unitV,0) + \sfrac{1}{4} \cdot \delta(\unitV,0)}{5}, & \cdots,\qquad & \steprel{V}{\appC{f}{\unitC}}{\textstyle\sum_{i=1}^k (\sfrac{1}{2})^i \cdot \delta(\unitV,0)}{2k+1}.
\end{align*}
Letting $k$ approach infinity, we derive that $\interp{\appC{f}{\unitC}}^V_{\Rightarrow} = \delta(\unitC,0)$, i.e., the program terminates with probability one. Further, the evaluation result is always unit, and the net cost is always zero.

Finally, we show that the distribution-based dynamics is equivalent to the trace-based one, so we can proceed to prove soundness with respect to the distribution-based semantics.

\begin{proposition}
  Let $V$ be an environment and $e$ be an expression. Then $\interp{e}^V_{\Rightarrow} = \interp{e}^V_{\Downarrow}$.
\end{proposition}
\begin{proof}
  We proceed by proving both $\interp{e}^V_{\Rightarrow} \le \interp{e}^V_{\Downarrow}$ and $\interp{e}^V_{\Downarrow} \le \interp{e}^V_{\Rightarrow}$.
  For the first inequality, it is sufficient to show that $\mu_n \le \interp{e}^V_{\Downarrow}$ for all $n \in \bbN$ where $\steprel{V}{e}{\mu_n}{n}$.
  For the second one, it suffices to show that $\nu_n \le \interp{e}^V_{\Rightarrow}$ for all $n \in \bbN$ where $\nu_n$ is a sub-distribution of executions in $\interp{e}^V_{\Downarrow}$ whose trace has length at most $n$.
  Both cases are done by induction on $n$.
\end{proof}

\paragraph{Soundness} We now restate and prove the soundness theorem using the distribution-based semantics. Again, it states that the initial potential can pay for the expected evaluation cost and the typing of the result.

\begin{theorem}[Soundness of AARA]\label{the:soundness}
  Let $\typingrel{\Gamma}{q}{e}{A}$ and $\valuerel{V}{\Gamma}$. Then
  \[
  \pot{V }{\Gamma} + q \ge \sum_{(v_0,q_0)} \interp{e}^V_{\Rightarrow}(v_0,q_0) \cdot (\pot{v_0}{ A} + q_0).
  \]
\end{theorem}
\begin{proof}
  It suffices to prove for every $n \in \bbN$, if $\steprel{V}{e}{\mu}{n}$, then
  \[
  \pot{V }{ \Gamma} + q \ge \sum_{(v_0,q_0) } \mu(v_0,q_0) \cdot (\pot{v_0 }{ A} + q_0).
  \]
  Proceed by induction on $n$ with inversion on $\steprel{V}{e}{\mu}{n}$ then inner induction on $\typingrel{\Gamma}{q}{e}{A}$. 
\end{proof}

\paragraph{Non-termination} So far we have only considered terminating executions in the evaluation dynamics, dealing with non-termination indirectly.
Recall that the distribution over $e$'s evaluations in environment $V$ is defined as
\[
\interp{e}^V_{\Downarrow}(v,q) \defeq
   \sum_{\sigma} p_\sigma \quad \text{where $\sigma$'s are \emph{finite} traces satisfying $\evalrel{V}{e}{v}{q}{p_\sigma}{\sigma}$},
\]
thus \emph{infinite} traces (e.g., non-terminating executions) are totally ignored.
Hence, the soundness theorem (\cref{the:soundness}) does \emph{not} imply that the typing judgment $\typingrel{\Gamma}{q}{e}{A}$ (where $e$ is instrumented with ticks to count evaluation steps) entails that the expected termination time of $e$ is finite.
We therefore now extend the dynamics to account for non-terminating behavior directly.

To deal with non-termination, we first introduce a dummy value $\circ$ to represent some partial evaluation.
We can then enrich the distribution-based dynamics with partial evaluation by forcing the result distribution $\mu$ in the judgment $\steprel{V}{e}{\mu}{n}$ to be a \emph{full} probability distribution instead of a subprobability one.
To achieve this, we extend $\mu$'s distributions to be over $(\mathsf{Val} \cup \{\circ\}) \times (\bbQ_{\ge 0} \cup \{ \infty \})$, including this new dummy value.
Most of the rules stay unchanged, except the following two:
\begin{mathpar}\scriptsize
  \Rule{PE:Base}{ }{ \steprel{V}{e}{\delta(\circ,0)}{0} }
  \and
  \Rule{PE:Let}{ \steprel{V}{e_1}{\mu}{n} \\ \Forall{(v_1,q_1) \in \mathrm{supp}(\mu)} (v_1 \neq \circ) \implies \steprel{V,x\mapsto v_1}{e_2}{\mu_{(v_1,q_1)}}{n} }{ \steprel{V}{\letS{e_1}{x}{e_2}}{\textstyle \sum_{q_1} \mu(\circ,q_1) \cdot \delta(\circ,q_1) + \sum_{(v_1,q_1) : v_1 \neq \circ} \sum_{(v_2,q_2)} \mu(v_1,q_1) \cdot \mu_{(v_1,q_1)}(v_2,q_2) \cdot \delta(v_2,q_1+q_2)  }{n+1} }
\end{mathpar}

However, we can no longer take the previous approach of defining $\interp{e}^V_{\Rightarrow}$ by the limit of $\{\mu_n\}_{n \in \bbN}$ where $\steprel{V}{e}{\mu_n}{n}$, because it no longer holds that, if $n \le m$ ,then $\mu_n \le \mu_m$ pointwise. To get around this, we define a new ordering on complete distributions, extending it to cover the dummy value differently.
We define $\mu_1 \sqsubseteq \mu_2$ as
\begin{itemize}
  \item $\Forall{v,q} (v \neq \circ) \implies \mu_1(v,q) \le \mu_2(v,q)$, and
  \item $\Forall{q} \mu_1((\mathsf{Val} \cup \{\circ\}) \times [0,q]) \ge \mu_2((\mathsf{Val} \cup \{\circ\}) \times [0,q])$.
\end{itemize}
For concrete values, the order above is the same as the pointwise order on subprobability distributions, but for divergence, we take the other direction---the property above implies that $\mu_1(\{\circ\} \times[0,q]) \ge \mu_2(\{\circ\} \times [0,q])$ for all $q \in \bbQ_{\ge 0} \cup \{\infty\}$.
Since we assume non-negative ticks, the probability that the cost is smaller than any $q$ with respect to $\mu_1$ should be greater than or equal to that with respect to $\mu_2$.
Formally, we prove that $\sqsubseteq$ defines an \emph{$\omega$-complete partial order} on distributions.

\begin{lemma}\label{lem:partialorder}
    The relation $\sqsubseteq$ defines a partial order on the distributions.
    Further, let $\{\mu_n\}_{n \in \bbN}$ be a sequence such that $\mu_1 \sqsubseteq \mu_2 \sqsubseteq \cdots \sqsubseteq \mu_n \sqsubseteq \cdots$.
    Then there exists a least distribution $\mu$ such that for all $n \in \bbN$, $\mu_n \sqsubseteq \mu$. Further, we denote $\mu$ by $\bigsqcup_{n \in \bbN} \mu_n$.
\end{lemma}

We now restate \cref{lem:dist-wcpo} in terms of the partial order $\sqsubseteq$ over distributions.

\begin{lemma}
	\label{lemma:dist def}
  If $\steprel{V}{e}{\mu_1}{n}$, $\steprel{V}{e}{\mu_2}{m}$ and $n \le m$, then $\mu_1 \sqsubseteq \mu_2$ pointwise.
  As a consequence, we can define $\interp{e}^V_{\Rightarrow} \defeq \bigsqcup_{n \in \bbN} \mu_n$ as the distribution of all possible terminating and non-terminating executions of a probabilistic program $e$ under environment $V$.
\end{lemma}
\begin{proof}
  By induction on the derivation of $\steprel{V}{e}{\mu_2}{m}$, followed by inversion on $\steprel{V}{e}{\mu_1}{n}$.
  The existence of the sequence appeals to \cref{lem:partialorder}.
\end{proof}

Recall that in the soundness proof, we induct on the index $n$ of $\steprel{V}{e}{\mu}{n}$.
The reason why this approach works is that the expected cost with respect to $\mu$ is \emph{$\omega$-continuous}, i.e., monotone and interchangeable with a limit operator.
Although it is unclear whether the continuity still holds for $\sqsubseteq$ or not, we can prove the following weaker result that is sufficient for our soundness proof.

\begin{lemma}
	\label{lemma:dist bound}
  Let $h(\mu) \defeq \sum_{q} \mu(\circ,q) \cdot q + \sum_{(v,q) : v \neq \circ} \mu(v,q) \cdot (\pot{v}{A} + q)$.
  Let $\{\mu_n\}_{n \in \bbN}$ be a sequence such that $\mu_1 \sqsubseteq \mu_2 \sqsubseteq \cdots \sqsubseteq \mu_n \sqsubseteq \cdots$.
  Let $M \in \bbR_{\ge 0}$.
  If $h(\mu_n) \le M$ for all $n \in \bbN$, then $h(\bigsqcup_{n \in \bbN} \mu_n) \le M$.
\end{lemma}

Now we can strengthen the soundness theorem to capture both termination and non-termination.

\begin{theorem}[Soundness of AARA, improved]\label{the:soundness:improved}
  Let $\typingrel{\Gamma}{q}{e}{A}$ and $\valuerel{V}{\Gamma}$. Then
  \[
  \pot{V}{\Gamma} + q \ge \sum_{q_0} \interp{e}^V_{\Rightarrow}(\circ,q_0) \cdot q_0 + \sum_{(v_0,q_0) : v_0 \neq \circ} \interp{e}^V_{\Rightarrow}(v_0,q_0) \cdot (\pot{v_0}{A} + q_0).
  \]
\end{theorem}
\begin{proof}
  By \cref{lemma:dist def,lemma:dist bound}
  it suffices to prove for every $n \in \bbN$, if $\steprel{V}{e}{\mu}{n}$, then
  \[
  \pot{V }{ \Gamma} + q \ge \sum_{q_0} \mu(\circ,q_0) \cdot q_0 + \sum_{(v_0,q_0) } \mu(v_0,q_0) \cdot (\pot{v_0 }{ A} + q_0).
  \]
  Again proved by induction on $n$ with inversion on $\steprel{V}{e}{\mu}{n}$, then $\typingrel{\Gamma}{q}{e}{A}$ inner induction.
\end{proof}

\begin{corollary}\label{cor:ast}
  Let $\typingrel{\Gamma}{q}{e}{A}$ and $\valuerel{V}{\Gamma}$.
  If a program $e$ is instrumented with ticks that account for evaluation steps, then $e$ terminates with probability one, i.e., $\interp{e}^V_{\Rightarrow}(\circ,q_0) = 0$ for all $q_0 \in \bbQ_{\ge 0} \cup \{\infty\}$.
\end{corollary}
\begin{proof}
  For all $q_0 \in \bbQ_{\ge 0}$, the probability $\interp{e}^V_{\Rightarrow}(\circ,q_0)$ is zero because if an execution does not terminate, the cost will keep increasing.
  For the case where $q_0 = \infty$, by \cref{the:soundness:improved}, $\interp{e}^V_{\Rightarrow}(\circ,\infty) \cdot \infty$ is bounded by $\pot{V}{\Gamma}+q < \infty$, thus the probability $\interp{e}^V_{\Rightarrow}(\circ,\infty)$ must be zero.
\end{proof}

\section{Implementation and Examples}
\label{sec:goat}

In this section we present some non-trivial probabilistic models which our implementation pRaML can handle in the same manner as described in previous sections. We follow up with a collection of experimental benchmarks from typing variants of our examples, and other examples from literature.

For these complex examples, we use our implementation pRaML of the probabilistic AARA type system extended to \textit{multivariate polynomial} potential functions with user-defined data types. While the potential functions supported in linear AARA are already multivariate, as each addend can depend on a different input size, the term \textit{multivariate} in the setting of potential functions refers to each addend depending on \textit{products} of input sizes - and in this case, also products of symbolic probabilities. With user-defined data types, those sizes can also measure the number of particular constructor types. We also include additional support for symbolic probabilities by allowing complementation (i.e., subtraction from 1). Extending the probabilistic type system laid out here to these domains does not involve significant conceptual changes; the potential function extensions - described in \cite{POPL:HAH11} and \cite{POPL:HDW17} - are orthogonal to the new probabilistic operation. 

\cref{tab:benchmarks} shows some analysis data given by pRaML on models described below and some examples from literature. It displays the number of linear constraints generated by typing the program using resource polynomials at a fixed degree for all programs of the same class, as well as how fast pRaML can complete type inference on consumer hardware. The literature examples include some example probabilistic loop code and conditional sampling model \cite{gordon2014probabilistic}, the simulation of a fair die with a fair coin using a Markov chain \cite{knuth1976algorithms}, a probabilistic variant of example code demonstrating quadratic resource usage  \cite{CAV:CHR17}, and the program \code{miner} \cite{PLDI:NCH18}. The final example, \code{fill} and \code{consume}, fills a list with probability values of $\sfrac 1 2$ or $\sfrac 1 3$ randomly according to a symbolic probability $p$, then iterates over the list, flipping a coin biased by each probability, and paying cost 1 for each heads flip.

Random walks form the core of stochastic algorithms and simulations. The Internet is so large that the tractability of measuring its contents is real concern, and it can be solved by random walks \cite{bar2008random}. Modeling problems from various fields also use random walks, ranging from economics \cite{meese1983empirical}, to biology \cite{codling2008random}, to ecology \cite{visser1997using}, to astrophysics \cite{macleod2010modeling}, and beyond. However, many random walks are non-trivial to analyze, which obscures properties like code efficiency from a non-expert programmer, and obscures stochastic model properties from their users. Even knowing the bounds of complex random walk first, the bounds can be nontrivial to verify by hand. Nonetheless, AARA can find them quickly, giving non-experts automatic access to expert bounds.

\begin{figure}
  \centering
  \begin{subfigure}[b]{0.48\textwidth}
\begin{lstlisting}
let rec gr Alice Bob =
  match Alice with
  | [] -> ()
  | ha::ta ->
    match Bob with
    | [] -> ()
    | hb::tb ->
      let _ = tick 1 in
      match flip 0.5 with
      | H -> gr ta (ha::Bob)
      | T -> gr (hb::Alice) tb
\end{lstlisting}
\caption{Gambler's ruin\label{fig:gambler}}
  \end{subfigure}
  \begin{subfigure}[b]{0.48\textwidth}
\begin{lstlisting}
let rec goat below at above =
	let _ = tick 1 in
	match at with
	| Lichen -> match flip 0.75 with
		| H -> match below with
			| [] -> ()
			| hd::tl -> goat tl hd (at::above)
		| T -> match above with
			| [] -> ()
			| hd::tl -> goat (at::below) hd tl
	| Grass -> match flip 0.5 with
		| H -> match below with
			| [] -> ()
			| hd::tl -> goat tl hd (at::above)
		| T ->  match above with
			| [] -> ()
			| hd::tl -> goat (at::below) hd tl
\end{lstlisting}
\caption{The life expectancy of a goat\label{fig:goat}}
  \end{subfigure}
  \caption{Implementations probabilistic programs in pRaML.}
\end{figure}

\begin{example}[Gambler's Ruin] 
There is an old problem in probability called the \textit{Gambler's Ruin}. \cref{fig:gambler} shows an implementation. It is set up so that Alice and Bob continually bet one dollar against each other on the results of a coin-flip until one player runs out of money. This is essentially a 2-sided random walk. If the coin is fair, Alice starts with $A$ dollars and Bob starts with $B$ dollars, then this series of bets is expected to take $AB$ time. Our multivariate implementation finds this bound exactly.

\end{example}

\begin{figure}[b]
\vspace{-10pt}
  \centering
  \begin{subfigure}[b]{0.33\textwidth}
\begin{lstlisting}[xleftmargin=0pt]
let reprice price =
  match flip 0.6 with
  | H ->
    match price with
    | [] -> []
    | _::t -> t
  | T -> ()::price
\end{lstlisting}
\caption{\label{fig:exists}}
  \end{subfigure}
  \begin{subfigure}[b]{0.33\textwidth}
\begin{lstlisting}[xleftmargin=0pt]
let rec buy price =
  match price with
  | [] -> ()
  | _::t ->
    let _ = tick 1 in
    buy t
\end{lstlisting}
\caption{\label{fig:buy}}
  \end{subfigure}
  \begin{subfigure}[b]{0.32\textwidth}
\begin{lstlisting}[xleftmargin=0pt]
let rec trade price time =
  match time with
  | [] -> ()
  | _::t ->
    let () = match flip 1/3 with
    	| H -> buy price 
    	| T -> ()
    in
    trade (reprice price) t
\end{lstlisting}
\caption{\label{fig:stock}}
  \end{subfigure}
  \caption{Stock buying}
 \end{figure}

\begin{example}[The Life Expectancy of a Goat]\label{exa:mountain-goat}

Consider modeling the following scenario: A mountain goat lives high up in the Rocky Mountains, eating grasses and lichens from the rocks. Depending on the food it find abundant, it either moves up or down the mountain. When it finds only lichens, it moves down with probability 75\% in an attempt to find better food sources. When it finds grasses, it moves with equal probability in either direction. However, if the goat moves too far down the mountain, it passes the treeline and gets hunted by wolves. On the other hand, if the goat tries to go up the mountain when at the very top, it falls off a cliff. Given some distribution of grasses and lichens on the mountain, and where the goat starts, what is the expected lifetime of the goat?

This is nontrivial to analyze by hand, but easy to code with the function \code{goat} in \cref{fig:goat}.
Then pRaML can find a cost bound. Letting $B$ be the distance from the goat to the treeline below, $G_A$ be the number of grassy areas above the goat, and $G$ be the total number of grassy areas, the expected lifetime is bounded above by $(B+1)(2(G+1)-G_B)$. This bound is rather complex, but its generality reveals some interesting cost dependencies. For instance, the derived bound is independent of the actual distance to the top of the mountain. It also makes it easy to get a sense of cost behaviour for particular cases: If the whole mountain is covered in lichen, then the expected lifespan is $2(B+1)$, in line with the goat's expected movement of half-a-space down the mountain per iteration. On the other hand, if the mountain is all grassy, then the lifetime more like the stopping time of the Gambler's Ruin experiment. 

\cref{tab:benchmarks} lists the analysis data for many different movement probabilities for varying amounts of plants. There we also use $A$ to represent the distance to the top of the mountain. 
\end{example}

\begin{table}\footnotesize
\caption{Experimental data of typing with pRaML.\label{tab:benchmarks}}
\begin{center}
\begin{tabular}{ c c c c }
 \hline
 Program description & Bound & \#Constraints & Time (in sec.) \\ 
 \hline
\code{goat} with $\frac 1 2, \frac 3 4$ & $(B+1)(2(G+1)-G_B)$ & 2084  & 0.15 \\  
\code{goat} with $ \frac 2 3, \frac 3 4$ & $3B+3$ & 2084 & 0.14 \\
 \code{goat} with $\frac 1 2, \frac 2 3, \frac 3 4$ & $(B+1)(2(G+1.5)-G_B)$ & 5336  & 0.25 \\
 \code{goat} with $\frac 1 2, \frac 3 5, \frac 2 3, \frac 3 4$ &$(B+1)(2(G+2.5)-G_B)$  & 10996 & 1.95\\
 \code{trade} with $\frac 3 5, \frac 1 3$ & $\frac 1 {15} T^2 + \frac 1 3 TP + \frac 4 {15} T$ & 157 & 0.04\\
 \code{trade} with $\frac 3 5, 1$ & $\frac 1 {5} T^2 + TP + \frac 4 5 T$ & 157 & 0.03\\
  \code{trade} with $\frac 2 5, 1$ & $\frac 3 {10} T^2 + TP + \frac 7 {10} T$ & 157 & 0.03\\
  \code{trade} with $\frac 2 5, \frac 1 3$ & $\frac 1 {10} T^2 +  \frac 1 3 TP + \frac 7 {30} T$ & 157 & 0.04\\
 probabilistic loop Ex 3 \cite{gordon2014probabilistic} & $\sfrac 4 3$ probability & 61 & 0.01\\
 bayes sampling Ex 6 \cite{gordon2014probabilistic} & $\sfrac 3 5$ probability & 112 & 0.01\\
 die simulation from coin \cite{knuth1976algorithms} & $\sfrac 1 6$ per die face & 5731 & 0.33\\
 random no-op \code{nested} variant \cite{CAV:CHR17} & $M^2+M$ & 205 & 0.03\\
 \code{miner} from \cite{PLDI:NCH18} & $\sfrac{15}2 M$ & 31 & 0.01\\
 \code{fill} and \code{consume} & $(\frac 1 3 + \frac p 6)M$  & 633 & 0.11
 \\ \hline
\end{tabular}
\end{center}
\end{table}

\begin{example}[Stock Buying]\label{exa:stock-buying}
Stock prices may behave like a random walk. In \cref{fig:stock} we simulate a buyer occasionally buying some stock over time, similarly to \cite{PLDI:NCH18}. Analysis with pRaML finds that the expected expenditure is $\frac 1 {15} T^2 + \frac 1 3 TP + \frac 4{15} T$, where $T$ is the time span and $P$ is the starting stock price. Results for other parameters for the price's walk and buy rate, respectively, may be found under \code{trade} in \cref{tab:benchmarks}.
\end{example}

\section{Applications}
In this section, we discuss two application domains of pRaML: analysis of discrete distributions (\cref{sec:sample}) and estimation of average-case cost (\cref{sec:model}).
\label{sec:app}
\subsection{Analysis of Discrete Distributions}
\label{sec:sample}

Although the only probabilistic fragment introduced by our programming language is probabilistic branching, we are able to implement a broad suite of discrete probability distributions and analyze their properties in our system.
In this section, we demonstrate how our tool can be used to not only verify that a program implements the desired distribution, but also analyze \emph{sample complexity} of the program, i.e., the expected number of flips consumed by the program to obtain a sample.
Sampling from probability distributions is a fundamental activity in many fields, e.g., Bayesian inference on probabilistic programs~\cite{CoRR:abs/stat/1301.1299,misc:dippl}, and the efficiency of sampling algorithms becomes increasingly important because Monte Carlo methods for probabilistic inference have a trend of requiring billions of random samples per second~\cite{misc:Djuric19}.
Our work provides an approach for understanding sample complexity of discrete distributions.

\paragraph{Case study: Discrete distribution generating (DDG) trees}
Recent work provides a universal representation of sampling algorithms for finite supports as \emph{discrete distribution generating} (DDG) binary trees~\cite{POPL:SFR20}.
The idea is to implement discrete distributions by only \emph{fair} coin flips.
Given a DDG binary tree $T$, the sample algorithm starts at the root of $T$, then repeatedly flips a fair coin, takes the left (resp., right) branch if the coin shows heads (resp., tails) until it reaches a leaf node labeled with an outcome from the support of the distribution.
Note the tree $T$ may contain back edges, i.e., the algorithm goes back to an ancestor after taking a branch of the current non-leaf node.
Back-edges are crucial for implementing non-dyadic probabilities, and they make the running time of the sampling algorithm nontrivial because the algorithm can have non-terminating executions.

\cref{fig:ddg-trees} presents two sample algorithms modified from an example in prior work~\cite{POPL:SFR20}.
Both programs are supposed to implement a distribution over $\{\code{Red},\code{Black}\}$, and return $\code{Red}$ with probability $0.3$, otherwise return $\code{Black}$.
First, we verify that both programs correctly implement the target distribution.
We achieve this by inserting ticks such that the program has one unit of cost when returning $\code{Red}$.
Our tool then derives that the expected cost for both programs is bounded by $0.3$ from above.
Meanwhile, we insert ticks in original programs where the program returns \code{Black} instead of \code{Red}, and our tool infers that the expected cost for both programs is at most $0.7$.
Because the expectation of an indicator function for an event $E$ equals to the probability of $E$, i.e.,
\[
\bbE( \lambda \omega. [\text{$E(\omega)$ is true} ]) = \sum_{\omega} \bbP(\omega) \cdot [\text{$E(\omega)$ is true}] = \sum_{\omega \in E} \bbP(\omega) = \bbP(E),
\]
we conclude that $\bbP(\text{result is \code{Red}}) \le 0.3$ and $\bbP(\text{result is \code{Black}}) \le 0.7$, thus the programs implement the desired distribution, by the fact that probabilities sum up to one.

Then, we study the expected performance of the two sample algorithms in \cref{fig:ddg-tree-fast,fig:ddg-tree-slow}.
We instrument the two programs with ticks to count the number of probabilistic choices made during the execution.
Our expected cost analysis successfully derives an upper bound for both programs: $2.0$ for \cref{fig:ddg-tree-fast}, and $4.2$ for \cref{fig:ddg-tree-slow}.
By a manual analysis, we also verify that these bounds are tight.
The result suggests that \cref{fig:ddg-tree-fast} is better than \cref{fig:ddg-tree-slow}.
We leave automatic tightness checking (e.g., by integrating a lower-bound analysis~\cite{PLDI:WFG19,WangHR20}) for future work.

\begin{figure}
\centering
\begin{subfigure}[b]{0.3\textwidth}
\begin{lstlisting}
let sample_fast () = 
  let rec aux () =
    let _ = tick 1 in
    match flip 0.5 with
    | H ->let _ = tick 1 in
      match flip 0.5 with
      | H ->let _ = tick 1 in
        match flip 0.5 with
        | H ->let _ = tick 1 in
          match flip 0.5 with
          | H -> aux ()
          | T -> Red
        | T -> Black
      | T -> Black
    | T -> Red
  in
  let _ = tick 1 in
  match flip 0.5 with
  | H -> aux ()
  | T -> Black
\end{lstlisting}
\caption{\label{fig:ddg-tree-fast}}
\end{subfigure}
\begin{subfigure}[b]{0.3\textwidth}
\begin{lstlisting}
let rec sample_slow () =
	let _ = tick 1 in
  match flip 0.5 with
  | H ->let _ = tick 1 in
    match flip 0.5 with
    | H ->let _ = tick 1 in
      sample_slow ()
    | T ->let _ = tick 1 in
      match flip 0.5 with
      | H -> 
        sample_slow ()
      | T ->let _ = tick 1 in
        match flip 0.5 with
        | H -> Red
        | T -> Black
  | T ->let _ = tick 1 in
    match flip 0.5 with
    | H ->let _ = tick 1 in
      match flip 0.5 with
      | H -> Black
      | T -> Red
    | T -> Black
\end{lstlisting}
\caption{\label{fig:ddg-tree-slow}}
\end{subfigure}
\begin{subfigure}[b]{0.35\textwidth}
\begin{lstlisting}
let rec negative_binomial p l =
  match l with
  | [] -> []
  | _::l' ->
    let _ =
      consume p : prob{0}{1}
    in
    match flip p with
    | H ->
      ()::(negative_binomial p l)
    | T ->
      negative_binomial p l'
\end{lstlisting}
\caption{\label{fig:neg-binomial}}
\end{subfigure}
\caption{(a) and (b) are samplers that return \code{Red} with probability $0.3$ or \code{Black} with probability $0.7$.  (c) is a program for negative binomial distributions.\label{fig:ddg-trees}}
\end{figure}

\paragraph{Case study: Negative binomial distributions}
Beyond distributions with fixed, finite supports, our system is also capable of analyzing discrete distributions with infinite supports and symbolic probabilities.
\cref{fig:neg-binomial} gives an implementation of negative binomial distributions; it returns a unit list whose length is the number of heads in a series of independent coin flips with probability $p$ before $|\ell|$ number of tails occurs.
The $\code{consume}$ expression is used to specify value-\emph{dependent} costs, which we explain later.

In this example, we want to study the program's sample complexity with respect to $p$ and $\ell$.
At first glance, the task seems impossible for our system, because while our AARA-based approach is able to derive multivariate-polynomial bound, the expected number of flips for negative binomial distributions involves fractions like $\frac{1}{1-p}$, which is not expressible in our system.
Nevertheless, we come up with a workaround that scales all the costs in program by a factor of $(1-p)$ to get rid of the resource bound's denominator.
This is achieved by the \code{consume} expression. Intuitively, $\mathsf{consume}~x:\tau$ specifies a cost that equals to the potential of the value of $x$ with respect to $\tau$.
Recall that $\pot{p}{\probA{q_H}{q_T}} \defeq p \cdot q_H + (1-p) \cdot q_T$;
thus, the $\code{consume}$ expression in the program introduces a cost of $(1-p)$.
Our type system succeeds in finding a linear bound $|\ell|$ on the expected number of flips.
Taking the scale factor into account, we conclude that the expected sample complexity for negative binomial distributions is at most $\frac{|\ell|}{1-p}$.

\paragraph{More examples}
A summary of all the case studies in the analysis of discrete distributions carried out in our system can be found in \cref{tab:distribution-analysis} .
All the analyses were processed in around one second.
The fractional bounds are derived using the scaling technique mentioned above.
For distributions $\code{dist}$ with integer supports, we also create a variant $\code{dist}_\bbE$ that specifies the value of the output sample as the cost. For such a case, our tool essentially performs a \emph{first-moment} analysis that computes the \emph{mean} value of the distributions.

\begin{table}\footnotesize
  \caption{Examples for sample-complexity or first-moment analysis of discrete distributions. In the bounds, $p$ is the value of the first probability argument, $n$ is the length of the first list argument, and $p_i$'s are the probability-valued elements in the first list argument.\label{tab:distribution-analysis}}
  \begin{tabular}{lllH}
    \hline
    Function & Description & Inferred Bound & Time \\ \hline
    $\code{sample\_fast} : \arrT{\unitT}{\mathsf{red\_or\_black}}$ & \cref{fig:ddg-tree-fast} & $2.00$ & 0.01s \\
    $\code{sample\_slow} : \arrT{\unitT}{\mathsf{red\_or\_black}}$ & \cref{fig:ddg-tree-slow} & $4.20$ & 0.01s \\
    $\code{dice} : \arrT{\unitT}{ \mathsf{dice} } $ & A fair dice & $3.67$ & 0.04s \\
    $\code{von\_neumann} : \arrT{\probT{}}{\boolT}$ & Make a fair coin from a biased one & $\frac{1}{p(1-p)}$ & 0.04s \\
    $\code{binomial} : \arrT{\probT{}}{\arrT{\li{\unitT}}{\li{\unitT}}}$ & Binomial distribution & $n$ & 0.18s \\
    $\code{binomial}_{\bbE} : \arrT{\probT{}}{\arrT{\li{\unitT}}{\li{\unitT}}}$ & Binomial distribution; output as cost & $p \cdot n$ & 0.18s \\
    $\code{geometric} : \arrT{\probT{}}{\li{\unitT}}$ & Geometric distribution & $\frac{1}{p}$ & 0.08s \\
    $\code{geometric}_\bbE : \arrT{\probT{}}{\li{\unitT}}$ & Geometric distribution; output as cost & $\frac{1-p}{p}$ & 0.08s \\
    $\code{poisson\_binomial} : \arrT{\li{\probT{}}}{\li{\unitT}}$ & Poisson binomial distribution & $n$ & 0.05s \\
    $\code{poisson\_binomial}_\bbE : \arrT{\li{\probT{}}}{\li{\unitT}}$ & Poisson binomial distribution; output as cost & $\sum_{1 \le i \le n} p_i$ & 0.05s \\
    $\code{negative\_binomial} : \arrT{\probT{}}{\arrT{\li{\unitT}}{\li{\unitT}}}$ & Negative binomial distribution & $\frac{n}{1-p}$ & 0.06s \\
    $\code{negative\_binomial}_\bbE : \arrT{\probT{}}{\arrT{\li{\unitT}}{\li{\unitT}}}$ & Negative binomial distribution; output as cost & $\frac{p \cdot n}{1 - p}$ & 0.06s \\ \hline
  \end{tabular}
\end{table}

\subsection{Estimation of Average Case Cost}
\label{sec:model}

Understanding resource requirements of computer programs is important for software engineering.
Much of the research has been focused on analyzing \emph{worst-case} resource usage and generating an input that exhibits the \emph{worst-case} performance, e.g.~\cite{ISSTA:KNP18,POPL:WH19}.
However, in practice, software performance can be sensitive to the \emph{distribution} of the actual inputs.
For example, although quicksort has a worst-case quadratic time complexity, it usually outperforms many other sorting algorithms (e.g., insertion sort) on randomly generated inputs.
Understanding the \emph{performance distribution} induced by the real-world \emph{input distribution} can then help carry out important tasks in software development such as performance evaluation and algorithm selection.
In this section, we illustrate how our tool can be used to characterize performance distributions of deterministic programs by their \emph{average-case} resource usage, through a combination with \emph{profiling} techniques.

\paragraph{Program tranformation} Profiling techniques, such as \emph{edge} profiling and \emph{path} profiling, have been used for speculative optimization (especially of branch conditions)~\cite{PLDI:Ramalingam96,ASPLOS:DS06}, symbolic execution~\cite{ICSE:FPV13,FSE:FPV14}, and performance analysis~\cite{ICSE:CLL16}.
The idea is to approximate a deterministic branch condition as a probabilistic choice, whose probability is determined by counting \emph{frequencies} of the two branches executed by a program on a collection of real-world inputs.
For example, if the then-branch $e_1$ of the expression $\condC{x}{e_1}{e_2}$ is executed 90\% of the time, then we transform the conditional with a probabilistic choice $\flipC{0.9}{e_1}{e_2}$.
Benefits of such profiling-based program transformation are:
(i) it does not require complicated analyses to account for the conditional probability of branches, 
(ii) it provides insights how the input distribution influences the control-flow of a program via an empirical probabilistic model, and
(iii) it can accrue profiling information from samples with small sizes but still generalize its average-case cost bounds to inputs with large sizes.

We have implemented an interpreter for the deterministic fragment of our programming language, which executes programs with concrete inputs and collects profiling information including frequencies of control-flow transitions.
We then use the profiling information to transform branch conditions to proper probabilistic choices.
Note that we have also implemented a statistical independence test to ensure that branch probabilities are constants, rather than dependent on structural features (e.g., lengths of lists) of the input.
Then we pass the transformed program to our type-based expected cost analysis to obtain a symbolic bound as the average-case estimation for the cost of the original program.

\paragraph{Case study: Sorting nearly-sorted lists}
It is known that comparison-based sorting algorithms cannot beat the $\Theta(n \log n)$ time complexity for input lists of length $n$.
However, if the sorting function is intended to process \emph{nearly-sorted} data---where every element may \emph{on average} be misplaced by at most some constant number $k$ of positions from the correct sorted order---then some sorting algorithms, e.g., \emph{insertion sort}, can achieve linear time complexity.
\cref{fig:isort-original} presents an implementation for insertion sort that uses ticks to count the number of comparisons.
Our tool derives that the worst-case cost bound for $\code{isort}(\ell)$ is $\binom{|\ell|}{2}$, which is quadratic in the length $|\ell|$ of the list $\ell$.
The only conditional expression occurs when the \code{insert} function compares the inserted element $x$ and the head $y$ of the sorted list $\ell$, and it recurses on the tail of $\ell$ if $x>y$.
Since an element may be misplaced by $k$ positions on average, intuitively, there should on average be $k$ recursions when inserting an element, which means that the condition $x>y$ evaluates to true with a constant probability $\frac{k-1}{k}$.

In our experiments, our tool managed to detect from a set of nearly-sorted lists that the conditional expression in \code{insert} can be approximated by a probabilistic choice with a constant probability.
\cref{fig:isort-transformed} illustrates one case where the branching probability is about $0.9$.
Our tool derives that the expected cost bound for $\code{isort'}(\ell)$ is $10 \cdot |\ell|$, which is linear in the length $|\ell|$ of the list $\ell$.
The linear bound also reflects that the list $\ell$ should be nearly sorted, in the sense that every element in $\ell$, on average, is misplaced by 10 positions from the correct sorted order.

\begin{figure}
\centering
\begin{subfigure}{0.48\textwidth}
\begin{lstlisting}
let rec insert x l =
  match l with
  | [] -> [x]
  | y::ys ->
    let _ = tick 1.0 in
    match (x > y) with
    | true -> y::(insert x ys)
    | false -> x::y::ys

let rec isort l =
  match l with
  | [] -> []
  | x::xs ->
    insert x (isort xs)
\end{lstlisting}    
\caption{Original\label{fig:isort-original}}
\end{subfigure}
\begin{subfigure}{0.48\textwidth}
\begin{lstlisting}
let rec insert' x l =
  match l with
  | [] -> [x]
  | y::ys ->
    let _ = tick 1.0 in
    match flip 0.9 with
    | H -> y::(insert' x ys)
    | T -> x::y::ys

let rec isort' l =
  match l with
  | [] -> []
  | x::xs ->
    insert' x (isort' xs)
\end{lstlisting}
\caption{Transformed\label{fig:isort-transformed}}
\end{subfigure}
\caption{Average-case cost estimation for insertion sort on nearly-sorted inputs\label{fig:insertion-sort}}
\end{figure}

\paragraph{Case study: Short-circuit Boolean interpretation} 

When implementing a compiler, one usually must decide how to interpret Boolean expressions. Most commonly, the decision is made to \textit{short-circuit} the \code{and} and \code{or} connectives. That is, if the first term determines the whole expression - \code{false} for \code{and} or \code{true} for \code{or} - then one skips evaluation of the second. Programmatically, this can be implemented with conditionals as in the following code for \code{interpret} in \cref{fig:shortc-original}. 

In the worst case, the code in \cref{fig:shortc-original} must iterate over every node of its input Boolean expression tree, which is exactly the non-probabilistic bound given by our tool. Specifically, letting $C$ be the number of constants, $B$ the number of binary connectives, and $N$ the number of negations, the bound is $C+B+N$. This is the same cost bound as naively evaluating every sub-expression, so it is unclear what value short-circuiting provides. However, if the Boolean constants used are uniformly random, one finds that the branching probability at each conditional can be approximated by a constant: about $0.5$ probability for each branch. Converting the code into \code{interpret'}, our tool now finds a better expected cost bound of $1.5B+N$. Because $C$ is always equal to $B+1$, this is a strictly better cost bound. 

\begin{figure}
\centering
\begin{subfigure}[b]{0.48\textwidth}
\begin{lstlisting}
let rec interpret exp =
  let _ = tick 1 in
  match exp with
  | True -> true
  | False -> false
  | Or (a,b) ->
      match interpret a with
      | true -> true
      | false -> interpret b
  | And (a,b) ->
      match interpret a with
      | true -> interpret b
      | false -> false
  | Neg a -> not (interpret a)
\end{lstlisting}    
\caption{Original\label{fig:shortc-original}}
\end{subfigure}
\begin{subfigure}[b]{0.48\textwidth}
\begin{lstlisting}
let rec interpret' exp =
  let _ = tick 1 in
  match exp with
  | True -> true
  | False -> false
  | Or (a,b) ->
      let _ = interpret' a in
      match flip 0.5 with
      | H -> true
      | T -> interpret' b
  | And (a,b) ->
      let _ = interpet' a in
      match flip 0.5 with
      | H -> interpret' b
      | T -> false
  | Neg a -> not (interpret' a)
\end{lstlisting}
\caption{Transformed\label{fig:shortc-transformed}}
\end{subfigure}
\caption{Average-case cost estimation for short-circuiting Boolean interpretation across uniform inputs\label{fig:shortc}}
\end{figure}

\paragraph{Case study: Sequential insertions in a hash table}
We implement a program in our language to model the hash table function from prior work on worst-case analyisis~\cite{ISSTA:KNP18,POPL:WH19}.
This is a complicated program where each key in the hash table is a string of length 8 and the hash function is DJBX33A from a PHP implementation.
The resource model is defined as the number of hash collisions.
In the worst case, our system derives that the number of collisions is bounded by $\binom{n}{2}$ where $n$ is the number of insertions.
The worst-case quadratic bound makes sense because one may construct a list of different strings with the same hash key.
However, if the hash table is used in a setting where security vulnerabilities like Denial-of-Service are not crucial and the inputs are sufficiently random, then the quadratic bound is not meaningful because it is usually assumed that an insertion into a hash table takes constant time.

In our experiments, from a set of randomly generated strings, our tool found out that both the probability that two input strings have the same hash key---and the probability that two input strings with the same hash key are different---are small constants.
Our tool then derives $0.11 \cdot n$ as an expected cost bound for the transformed hash-table program with $n$ insertions, which indicates that the number of hash collisions should be linear in the number of insertions in practice.


\section{Related Work}
\label{sec:related-work}

We discuss the most-closely related work on expected cost analysis of
probabilistic programs in \cref{sec:intro}. Other related work
includes cost analysis of deterministic programs and other (type-based)
analyses of probabilistic programs.

\vspace{-6pt}
\paragraph{Cost analysis for deterministic programs}
Automatic and semiautomatic resource bound analysis for deterministic programs has been extensively studied.
Our work is based on AARA, which was initially introduced~\cite{POPL:HJ03} to automatically derive linear heap-space bounds for first-order functional programs.  
AARA has been extended to polynomial bounds~\cite{POPL:HAH11,ESOP:HH10,TLCA:HM15}, exponential bounds~\cite{FoSSaCS:KH20}, logarithmic bounds~\cite{CoRR:abs/cs/1807.08242}, higher-order functions~\cite{POPL:JHL10,POPL:HDW17}, user-defined datatypes~\cite{POPL:HDW17,FM:JLH09}, and separation logic~\cite{ESOP:Atkey10}.
The technique has also been generalized to imperative arithmetic programs~\cite{PLDI:CHS15,CAV:CHR17}, as well as integrated into formal proof assistants~\cite{ITP:CP15,ITP:Nipkow15}.

Beyond AARA, there have been many other approaches to formal resource analysis of deterministic programs. Some approaches, similarly to AARA, do so via a type system, including sized types~\cite{phd:Vasconcelos08} 
refinement types~\cite{OOPSLA:WWC17,POPL:CBG17,ESOP:CGA15,PLDI:KWP19,POPL:RBG18,kn:Xi02},
linear dependent types~\cite{LICS:LG11,POPL:LP13}, and
annotated type systems~\cite{POPL:CW00,POPL:Danielsson08}. Such a type-based approach usually involves constraint solving, some notion of linearity, and high composability, like AARA. However, there is often a tradeoff in programmer burden, like requiring more user annotation for better results. There are also non-typed-based approaches,
recurrence solving~\cite{TACAS:AFR15,PLDI:KBB17,ICFP:DLR15,APLAS:FH14,POPL:KML20,ENTCS:AAG09},
abstract interpretation~\cite{LPAR:BHH10,CAV:Gulwani09,CAV:SZV14,SAS:ZSG11},
term-rewriting techniques~\cite{RTA:AM13,TACAS:BEG14,IJCAR:FNH16,JAR:NEG13},
defunctionalization~\cite{ICFP:ALM15},
and symbolic execution~\cite{ICSE:BJS09,ISSTA:KNP18}. These approaches vary more wildly from the system used in this work. 

Despite the number of such approaches to resource analysis, exceedingly few have been adapted to the probabilistic domain, and even less automated. The work in this article represents the first such automated system for probabilistic functional programs. Imperative probabilistic programs have already enjoyed such automated resource analysis in prior work, first established through imperative AARA techniques \cite{PLDI:NCH18}. 




\vspace{-6pt}
\paragraph{Type-base analysis for probabilistic programs}
Other properties of probabilistic programs, aside from expected cost, can be analyzed by type-based approaches. Almost-sure termination of functional probabilistic programs can be reasoned about through the dependent type systems of of Dal Lago et al.~\cite{kn:LG18,TOPLAS:LG19}.
Bhat et al.~\cite{POPL:BAVG12,TACAS:BBGR13} develop a type system to check absolute continuity of probabilistic first-order let-programs and derive corresponding density functions for the distributions specified by the programs.
\textsc{Fuzz}~\cite{ICFP:RP10} uses linear types augmented with a probability monad to reason about differential privacy of randomized computation,
and \textsc{DFuzz}~\cite{POPL:GHH13} later generalizes it with indexed types and lightweight dependent types to certify differential privacy for a broader class of benchmarks.
Recently, Lew et al.~\cite{POPL:LCS20} have developed a type system for programmable probabilistic inference with trace types, where well-typed inference programs soundly derive posterior distributions by construction.
In this paper, we focus on expected cost bound analysis for probabilistic programs.

\vspace{-4pt}

\section{Conclusion}

By combining a carefully developed probabilistic semantics with the AARA type system, we have shown that probabilistic programs in a functional language can be effectively analyzed in an automated manner. Our implementation pRaML infers worst-case expected bounds on resource usage for a variety of probabilistic models and algorithms, and parameterizes the bounds by both input sizes and symbolic probabilities. We make use of these parameterized bounds to analyze new and interesting application domains, like sample-complexity and a generalized average-case analysis. In the future, we hope to overcome the semantic soundness obstacles that bar non-monotone resource usage, and in doing so provide a fully-conservative extension of non-probabilistic AARA.

\begin{acks}                            
  This article is based on research supported by DARPA under AA
  Contract FA8750-18-C-0092 and by the National Science Foundation
  under SaTC Award 1801369, CAREER Award 1845514, and SHF Awards
  1812876 and 2007784.
  Any opinions, findings, and conclusions contained in this document
  are those of the authors and do not necessarily reflect the views of
  the sponsoring organizations.
\end{acks}

\bibliography{db,more}


\begin{thebibliography}{100}


\ifx \showCODEN    \undefined \def \showCODEN     #1{\unskip}     \fi
\ifx \showDOI      \undefined \def \showDOI       #1{#1}\fi
\ifx \showISBNx    \undefined \def \showISBNx     #1{\unskip}     \fi
\ifx \showISBNxiii \undefined \def \showISBNxiii  #1{\unskip}     \fi
\ifx \showISSN     \undefined \def \showISSN      #1{\unskip}     \fi
\ifx \showLCCN     \undefined \def \showLCCN      #1{\unskip}     \fi
\ifx \shownote     \undefined \def \shownote      #1{#1}          \fi
\ifx \showarticletitle \undefined \def \showarticletitle #1{#1}   \fi
\ifx \showURL      \undefined \def \showURL       {\relax}        \fi
\providecommand\bibfield[2]{#2}
\providecommand\bibinfo[2]{#2}
\providecommand\natexlab[1]{#1}
\providecommand\showeprint[2][]{arXiv:#2}

\bibitem[\protect\citeauthoryear{Albert, Arenas, Genaim, G{\'o}mez-Zamalloa,
  Puebla, Ram{\'\i}rez, Rom{\'a}n, and Zanardini}{Albert et~al\mbox{.}}{2009}]%
        {ENTCS:AAG09}
\bibfield{author}{\bibinfo{person}{E. Albert}, \bibinfo{person}{P. Arenas},
  \bibinfo{person}{S. Genaim}, \bibinfo{person}{M. G{\'o}mez-Zamalloa},
  \bibinfo{person}{G. Puebla}, \bibinfo{person}{D. Ram{\'\i}rez},
  \bibinfo{person}{G. Rom{\'a}n}, {and} \bibinfo{person}{D. Zanardini}.}
  \bibinfo{year}{2009}\natexlab{}.
\newblock \showarticletitle{{Termination and Cost Analysis with COSTA and its
  User Interfaces}}.
\newblock \bibinfo{journal}{\emph{Electr.\ Notes Theor.\ Comp.\ Sci.}}
  \bibinfo{volume}{258} (\bibinfo{date}{December} \bibinfo{year}{2009}).
\newblock
Issue 1.


\bibitem[\protect\citeauthoryear{Albert, Fern{\'{a}}ndez, and
  Rom{\'{a}}n{-}D{\'{\i}}ez}{Albert et~al\mbox{.}}{2015}]%
        {TACAS:AFR15}
\bibfield{author}{\bibinfo{person}{E. Albert}, \bibinfo{person}{J.~C.
  Fern{\'{a}}ndez}, {and} \bibinfo{person}{G. Rom{\'{a}}n{-}D{\'{\i}}ez}.}
  \bibinfo{year}{2015}\natexlab{}.
\newblock \showarticletitle{{Non-cumulative Resource Analysis}}. In
  \bibinfo{booktitle}{\emph{Tools and Algs.\ for the Construct.\ and Anal.\ of
  Syst. (TACAS'15)}}.
\newblock


\bibitem[\protect\citeauthoryear{Atkey}{Atkey}{2010}]%
        {ESOP:Atkey10}
\bibfield{author}{\bibinfo{person}{R. Atkey}.} \bibinfo{year}{2010}\natexlab{}.
\newblock \showarticletitle{{Amortised Resource Analysis with Separation
  Logic}}. In \bibinfo{booktitle}{\emph{European Symp.\ on Programming
  (ESOP'10)}}.
\newblock


\bibitem[\protect\citeauthoryear{Avanzini, Dal~Lago, and Ghyselen}{Avanzini
  et~al\mbox{.}}{2019}]%
        {LICS:ALG19}
\bibfield{author}{\bibinfo{person}{M. Avanzini}, \bibinfo{person}{U. Dal~Lago},
  {and} \bibinfo{person}{A. Ghyselen}.} \bibinfo{year}{2019}\natexlab{}.
\newblock \showarticletitle{{Type-Based Complexity Analysis of Probabilistic
  Functional Programs}}. In \bibinfo{booktitle}{\emph{Logic in Computer Science
  (LICS'19)}}.
\newblock


\bibitem[\protect\citeauthoryear{Avanzini, Dal~Lago, and Moser}{Avanzini
  et~al\mbox{.}}{2015}]%
        {ICFP:ALM15}
\bibfield{author}{\bibinfo{person}{M. Avanzini}, \bibinfo{person}{U. Dal~Lago},
  {and} \bibinfo{person}{G. Moser}.} \bibinfo{year}{2015}\natexlab{}.
\newblock \showarticletitle{{Analysing the Complexity of Functional Programs:
  Higher-Order Meets First-Order}}. In \bibinfo{booktitle}{\emph{Int.\ Conf.\
  on Functional Programming (ICFP'15)}}.
\newblock


\bibitem[\protect\citeauthoryear{Avanzini and Moser}{Avanzini and
  Moser}{2013}]%
        {RTA:AM13}
\bibfield{author}{\bibinfo{person}{M. Avanzini} {and} \bibinfo{person}{G.
  Moser}.} \bibinfo{year}{2013}\natexlab{}.
\newblock \showarticletitle{{A Combination Framework for Complexity}}. In
  \bibinfo{booktitle}{\emph{Int.\ Conf.\ on Rewriting Techniques and
  Applications (RTA'13)}}.
\newblock


\bibitem[\protect\citeauthoryear{Bar-Yossef and Gurevich}{Bar-Yossef and
  Gurevich}{2008}]%
        {bar2008random}
\bibfield{author}{\bibinfo{person}{Ziv Bar-Yossef} {and} \bibinfo{person}{Maxim
  Gurevich}.} \bibinfo{year}{2008}\natexlab{}.
\newblock \showarticletitle{Random sampling from a search engine's index}.
\newblock \bibinfo{journal}{\emph{Journal of the ACM (JACM)}}
  \bibinfo{volume}{55}, \bibinfo{number}{5} (\bibinfo{year}{2008}),
  \bibinfo{pages}{1--74}.
\newblock


\bibitem[\protect\citeauthoryear{Barthe, Gr{\'e}goire, and
  Zanella~B{\'e}guelin}{Barthe et~al\mbox{.}}{2009}]%
        {POPL:BGB09}
\bibfield{author}{\bibinfo{person}{G. Barthe}, \bibinfo{person}{B.
  Gr{\'e}goire}, {and} \bibinfo{person}{S. Zanella~B{\'e}guelin}.}
  \bibinfo{year}{2009}\natexlab{}.
\newblock \showarticletitle{{Formal Certification of Code-based Cryptographic
  Proofs}}. In \bibinfo{booktitle}{\emph{Princ.\ of Prog.\ Lang. (POPL'09)}}.
\newblock


\bibitem[\protect\citeauthoryear{Barthe, K{\"o}pf, Olmedo, and
  Zanella~B{\'e}guelin}{Barthe et~al\mbox{.}}{2012}]%
        {POPL:BKO12}
\bibfield{author}{\bibinfo{person}{G. Barthe}, \bibinfo{person}{B. K{\"o}pf},
  \bibinfo{person}{F. Olmedo}, {and} \bibinfo{person}{S.
  Zanella~B{\'e}guelin}.} \bibinfo{year}{2012}\natexlab{}.
\newblock \showarticletitle{{Probabilistic Relational Reasoning for
  Differential Privacy}}. In \bibinfo{booktitle}{\emph{Princ.\ of Prog.\ Lang.
  (POPL'12)}}.
\newblock


\bibitem[\protect\citeauthoryear{Batz, Kaminski, Katoen, and Matheja}{Batz
  et~al\mbox{.}}{2018}]%
        {ESOP:BKKM18}
\bibfield{author}{\bibinfo{person}{Kevin Batz}, \bibinfo{person}{B.~L.
  Kaminski}, \bibinfo{person}{J.-P. Katoen}, {and} \bibinfo{person}{C.
  Matheja}.} \bibinfo{year}{2018}\natexlab{}.
\newblock \showarticletitle{{How long, O Bayesian network, will I sample
  thee?}}. In \bibinfo{booktitle}{\emph{European Symp.\ on Programming
  (ESOP'18)}}.
\newblock


\bibitem[\protect\citeauthoryear{Bhat, Agarwal, Vuduc, and Gray}{Bhat
  et~al\mbox{.}}{2012}]%
        {POPL:BAVG12}
\bibfield{author}{\bibinfo{person}{S. Bhat}, \bibinfo{person}{A. Agarwal},
  \bibinfo{person}{R. Vuduc}, {and} \bibinfo{person}{A. Gray}.}
  \bibinfo{year}{2012}\natexlab{}.
\newblock \showarticletitle{{A Type Theory for Probability Density Functions}}.
  In \bibinfo{booktitle}{\emph{Princ.\ of Prog.\ Lang. (POPL'12)}}.
\newblock


\bibitem[\protect\citeauthoryear{Bhat, Borgstr{\"o}m, Gordon, and Russo}{Bhat
  et~al\mbox{.}}{2013}]%
        {TACAS:BBGR13}
\bibfield{author}{\bibinfo{person}{S. Bhat}, \bibinfo{person}{J.
  Borgstr{\"o}m}, \bibinfo{person}{A.~D. Gordon}, {and} \bibinfo{person}{C.
  Russo}.} \bibinfo{year}{2013}\natexlab{}.
\newblock \showarticletitle{{Deriving probability density functions from
  probabilistic functional programs}}. In \bibinfo{booktitle}{\emph{Tools and
  Algs.\ for the Construct.\ and Anal.\ of Syst. (TACAS'13)}}.
\newblock


\bibitem[\protect\citeauthoryear{Billingsley}{Billingsley}{2012}]%
        {book:Billingsley12}
\bibfield{author}{\bibinfo{person}{P. Billingsley}.}
  \bibinfo{year}{2012}\natexlab{}.
\newblock \bibinfo{booktitle}{\emph{{Probability and Measure}}}.
\newblock \bibinfo{publisher}{John Wiley \& Sons, Inc.}
\newblock


\bibitem[\protect\citeauthoryear{Blanc, Henzinger, Hottelier, and
  Kov{\'a}cs}{Blanc et~al\mbox{.}}{2010}]%
        {LPAR:BHH10}
\bibfield{author}{\bibinfo{person}{R. Blanc}, \bibinfo{person}{T.~A.
  Henzinger}, \bibinfo{person}{T. Hottelier}, {and} \bibinfo{person}{L.
  Kov{\'a}cs}.} \bibinfo{year}{2010}\natexlab{}.
\newblock \showarticletitle{{ABC: Algebraic Bound Computation for Loops}}. In
  \bibinfo{booktitle}{\emph{Logic for Prog., AI., and Reasoning (LPAR'10)}}.
\newblock


\bibitem[\protect\citeauthoryear{Borgstr{\"o}m, Dal~Lago, Gordon, and
  Szymczak}{Borgstr{\"o}m et~al\mbox{.}}{2016}]%
        {ICFP:BLG16}
\bibfield{author}{\bibinfo{person}{J. Borgstr{\"o}m}, \bibinfo{person}{U.
  Dal~Lago}, \bibinfo{person}{A.~D. Gordon}, {and} \bibinfo{person}{M.
  Szymczak}.} \bibinfo{year}{2016}\natexlab{}.
\newblock \showarticletitle{{A Lambda-Calculus Foundation for Universal
  Probabilistic Programming}}. In \bibinfo{booktitle}{\emph{Int.\ Conf.\ on
  Functional Programming (ICFP'16)}}.
\newblock


\bibitem[\protect\citeauthoryear{Br{\'a}zdil, Esparza, and Ku{\v
  c}era}{Br{\'a}zdil et~al\mbox{.}}{2005}]%
        {FOCS:BEK05}
\bibfield{author}{\bibinfo{person}{T. Br{\'a}zdil}, \bibinfo{person}{J.
  Esparza}, {and} \bibinfo{person}{A. Ku{\v c}era}.}
  \bibinfo{year}{2005}\natexlab{}.
\newblock \showarticletitle{{Analysis and Prediction of the Long-Run Behavior
  of Probabilistic Sequential Programs with Recursion}}. In
  \bibinfo{booktitle}{\emph{Found.\ of Comp.\ Sci. (FOCS'05)}}.
\newblock


\bibitem[\protect\citeauthoryear{Br{\'a}zdil, Kiefer, and Ku{\v
  c}era}{Br{\'a}zdil et~al\mbox{.}}{2014}]%
        {JACM:BKK14}
\bibfield{author}{\bibinfo{person}{T. Br{\'a}zdil}, \bibinfo{person}{S.
  Kiefer}, {and} \bibinfo{person}{A. Ku{\v c}era}.}
  \bibinfo{year}{2014}\natexlab{}.
\newblock \showarticletitle{{Efficient Analysis of Probabilistic Programs with
  an Unbounded Counter}}.
\newblock \bibinfo{journal}{\emph{J.\ ACM}}  \bibinfo{volume}{61}
  (\bibinfo{date}{November} \bibinfo{year}{2014}).
\newblock
Issue 6.


\bibitem[\protect\citeauthoryear{Brockschmidt, Emmes, Falke, Fuhs, and
  Giesl}{Brockschmidt et~al\mbox{.}}{2014}]%
        {TACAS:BEG14}
\bibfield{author}{\bibinfo{person}{M. Brockschmidt}, \bibinfo{person}{F.
  Emmes}, \bibinfo{person}{S. Falke}, \bibinfo{person}{C. Fuhs}, {and}
  \bibinfo{person}{J. Giesl}.} \bibinfo{year}{2014}\natexlab{}.
\newblock \showarticletitle{{Alternating Runtime and Size Complexity Analysis
  of Integer Programs}}. In \bibinfo{booktitle}{\emph{Tools and Algs.\ for the
  Construct.\ and Anal.\ of Syst. (TACAS'14)}}.
\newblock


\bibitem[\protect\citeauthoryear{Burnim, Juvekar, and Sen}{Burnim
  et~al\mbox{.}}{2009}]%
        {ICSE:BJS09}
\bibfield{author}{\bibinfo{person}{J. Burnim}, \bibinfo{person}{S. Juvekar},
  {and} \bibinfo{person}{K. Sen}.} \bibinfo{year}{2009}\natexlab{}.
\newblock \showarticletitle{{WISE: Automated Test Generation for Worst-case
  Complexity}}. In \bibinfo{booktitle}{\emph{Int.\ Conf.\ on Softw.\ Eng.
  (ICSE'09)}}.
\newblock


\bibitem[\protect\citeauthoryear{Carbonneaux, Hoffmann, Reps, and
  Shao}{Carbonneaux et~al\mbox{.}}{2017}]%
        {CAV:CHR17}
\bibfield{author}{\bibinfo{person}{Q. Carbonneaux}, \bibinfo{person}{J.
  Hoffmann}, \bibinfo{person}{T. Reps}, {and} \bibinfo{person}{Z. Shao}.}
  \bibinfo{year}{2017}\natexlab{}.
\newblock \showarticletitle{{Automated Resource Analysis with Coq Proof
  Objects}}. In \bibinfo{booktitle}{\emph{Computer Aided Verif. (CAV'17)}}.
\newblock


\bibitem[\protect\citeauthoryear{Carbonneaux, Hoffmann, and Shao}{Carbonneaux
  et~al\mbox{.}}{2015}]%
        {PLDI:CHS15}
\bibfield{author}{\bibinfo{person}{Q. Carbonneaux}, \bibinfo{person}{J.
  Hoffmann}, {and} \bibinfo{person}{Z. Shao}.} \bibinfo{year}{2015}\natexlab{}.
\newblock \showarticletitle{{Compositional Certified Resource Bounds}}. In
  \bibinfo{booktitle}{\emph{Prog.\ Lang.\ Design and Impl. (PLDI'15)}}.
\newblock


\bibitem[\protect\citeauthoryear{Carpenter, Gelman, Hoffman, Lee, Goodrich,
  Betancourt, Brubaker, Guo, Li, and Riddell}{Carpenter et~al\mbox{.}}{2017}]%
        {kn:CGH17}
\bibfield{author}{\bibinfo{person}{B. Carpenter}, \bibinfo{person}{A. Gelman},
  \bibinfo{person}{M.~D. Hoffman}, \bibinfo{person}{D. Lee},
  \bibinfo{person}{B. Goodrich}, \bibinfo{person}{M. Betancourt},
  \bibinfo{person}{M. Brubaker}, \bibinfo{person}{J. Guo}, \bibinfo{person}{P.
  Li}, {and} \bibinfo{person}{A. Riddell}.} \bibinfo{year}{2017}\natexlab{}.
\newblock \showarticletitle{{Stan: A Probabilistic Programming Language}}.
\newblock \bibinfo{journal}{\emph{J.\ Statistical Softw.}}
  \bibinfo{volume}{76} (\bibinfo{year}{2017}).
\newblock
Issue 1.


\bibitem[\protect\citeauthoryear{Chargu{\'{e}}raud and
  Pottier}{Chargu{\'{e}}raud and Pottier}{2015}]%
        {ITP:CP15}
\bibfield{author}{\bibinfo{person}{A. Chargu{\'{e}}raud} {and}
  \bibinfo{person}{F. Pottier}.} \bibinfo{year}{2015}\natexlab{}.
\newblock \showarticletitle{{Machine-Checked Verification of the Correctness
  and Amortized Complexity of an Efficient Union-Find Implementation}}. In
  \bibinfo{booktitle}{\emph{Interactive Theorem Proving (ITP'15)}}.
\newblock


\bibitem[\protect\citeauthoryear{Chatterjee and Fu}{Chatterjee and Fu}{2017}]%
        {CoRR:abs/cs/1701.02944}
\bibfield{author}{\bibinfo{person}{K. Chatterjee} {and} \bibinfo{person}{H.
  Fu}.} \bibinfo{year}{2017}\natexlab{}.
\newblock \bibinfo{booktitle}{\emph{{Termination of Nondeterministic Recursive
  Probabilistic Programs}}}.
\newblock \bibinfo{type}{{T}echnical {R}eport}. \bibinfo{institution}{Computing
  Research Repository}.
\newblock


\bibitem[\protect\citeauthoryear{Chatterjee, Fu, and Goharshady}{Chatterjee
  et~al\mbox{.}}{2016a}]%
        {CAV:CFG16}
\bibfield{author}{\bibinfo{person}{K. Chatterjee}, \bibinfo{person}{H. Fu},
  {and} \bibinfo{person}{A.~K. Goharshady}.} \bibinfo{year}{2016}\natexlab{a}.
\newblock \showarticletitle{{Termination Analysis of Probabilistic Programs
  Through Positivstellensatz's}}. In \bibinfo{booktitle}{\emph{Computer Aided
  Verif. (CAV'16)}}.
\newblock


\bibitem[\protect\citeauthoryear{Chatterjee, Fu, and Murhekar}{Chatterjee
  et~al\mbox{.}}{2017a}]%
        {CAV:CFM17}
\bibfield{author}{\bibinfo{person}{K. Chatterjee}, \bibinfo{person}{H. Fu},
  {and} \bibinfo{person}{A. Murhekar}.} \bibinfo{year}{2017}\natexlab{a}.
\newblock \showarticletitle{{Automated Recurrence Analysis for Almost-Linear
  Expected-Runtime Bounds}}. In \bibinfo{booktitle}{\emph{Computer Aided Verif.
  (CAV'17)}}.
\newblock


\bibitem[\protect\citeauthoryear{Chatterjee, Fu, Novotn{\'y}, and
  Hasheminezhad}{Chatterjee et~al\mbox{.}}{2016b}]%
        {POPL:CFN16}
\bibfield{author}{\bibinfo{person}{K. Chatterjee}, \bibinfo{person}{H. Fu},
  \bibinfo{person}{P. Novotn{\'y}}, {and} \bibinfo{person}{R. Hasheminezhad}.}
  \bibinfo{year}{2016}\natexlab{b}.
\newblock \showarticletitle{{Algorithmic Analysis of Qualitative and
  Quantitative Termination Problems for Affine Probabilistic Programs}}. In
  \bibinfo{booktitle}{\emph{Princ.\ of Prog.\ Lang. (POPL'16)}}.
\newblock


\bibitem[\protect\citeauthoryear{Chatterjee, Novotn{\'y}, and {\v
  Z}ikeli{\'c}}{Chatterjee et~al\mbox{.}}{2017b}]%
        {POPL:CNZ17}
\bibfield{author}{\bibinfo{person}{K. Chatterjee}, \bibinfo{person}{P.
  Novotn{\'y}}, {and} \bibinfo{person}{{\DJ}. {\v Z}ikeli{\'c}}.}
  \bibinfo{year}{2017}\natexlab{b}.
\newblock \showarticletitle{{Stochastic Invariants for Probabilistic
  Termination}}. In \bibinfo{booktitle}{\emph{Princ.\ of Prog.\ Lang.
  (POPL'17)}}.
\newblock


\bibitem[\protect\citeauthoryear{Chen, Liu, and Le}{Chen et~al\mbox{.}}{2016}]%
        {ICSE:CLL16}
\bibfield{author}{\bibinfo{person}{B. Chen}, \bibinfo{person}{Y. Liu}, {and}
  \bibinfo{person}{W. Le}.} \bibinfo{year}{2016}\natexlab{}.
\newblock \showarticletitle{{Generating Performance Distributions via
  Probabilistic Symbolic Execution}}. In \bibinfo{booktitle}{\emph{Int.\ Conf.\
  on Softw.\ Eng. (ICSE'16)}}.
\newblock


\bibitem[\protect\citeauthoryear{{\c{C}}i{\c{c}}ek, Barthe, Gaboardi, Garg, and
  Hoffmann}{{\c{C}}i{\c{c}}ek et~al\mbox{.}}{2017}]%
        {POPL:CBG17}
\bibfield{author}{\bibinfo{person}{E. {\c{C}}i{\c{c}}ek}, \bibinfo{person}{G.
  Barthe}, \bibinfo{person}{M. Gaboardi}, \bibinfo{person}{D. Garg}, {and}
  \bibinfo{person}{J. Hoffmann}.} \bibinfo{year}{2017}\natexlab{}.
\newblock \showarticletitle{{Relational Cost Analysis}}. In
  \bibinfo{booktitle}{\emph{Princ.\ of Prog.\ Lang. (POPL'17)}}.
\newblock


\bibitem[\protect\citeauthoryear{{\c{C}}i{\c{c}}ek, Garg, and
  Acar}{{\c{C}}i{\c{c}}ek et~al\mbox{.}}{2015}]%
        {ESOP:CGA15}
\bibfield{author}{\bibinfo{person}{E. {\c{C}}i{\c{c}}ek}, \bibinfo{person}{D.
  Garg}, {and} \bibinfo{person}{U.~A. Acar}.} \bibinfo{year}{2015}\natexlab{}.
\newblock \showarticletitle{{Refinement Types for Incremental Computational
  Complexity}}. In \bibinfo{booktitle}{\emph{European Symp.\ on Programming
  (ESOP'15)}}.
\newblock


\bibitem[\protect\citeauthoryear{Codling, Plank, and Benhamou}{Codling
  et~al\mbox{.}}{2008}]%
        {codling2008random}
\bibfield{author}{\bibinfo{person}{Edward~A Codling},
  \bibinfo{person}{Michael~J Plank}, {and} \bibinfo{person}{Simon Benhamou}.}
  \bibinfo{year}{2008}\natexlab{}.
\newblock \showarticletitle{Random walk models in biology}.
\newblock \bibinfo{journal}{\emph{Journal of the Royal Society Interface}}
  \bibinfo{volume}{5}, \bibinfo{number}{25} (\bibinfo{year}{2008}),
  \bibinfo{pages}{813--834}.
\newblock


\bibitem[\protect\citeauthoryear{Crary and Weirich}{Crary and Weirich}{2000}]%
        {POPL:CW00}
\bibfield{author}{\bibinfo{person}{K. Crary} {and} \bibinfo{person}{S.
  Weirich}.} \bibinfo{year}{2000}\natexlab{}.
\newblock \showarticletitle{{Resource Bound Certification}}. In
  \bibinfo{booktitle}{\emph{Princ.\ of Prog.\ Lang. (POPL'00)}}.
\newblock


\bibitem[\protect\citeauthoryear{Da~Silva and Steffan}{Da~Silva and
  Steffan}{2006}]%
        {ASPLOS:DS06}
\bibfield{author}{\bibinfo{person}{J. Da~Silva} {and} \bibinfo{person}{J.~G.
  Steffan}.} \bibinfo{year}{2006}\natexlab{}.
\newblock \showarticletitle{{A Probabilistic Pointer Analysis for Speculative
  Optimizations}}. In \bibinfo{booktitle}{\emph{Architectural Support for
  Prog.\ Lang.\ and Op.\ Syst. (ASPLOS'06)}}.
\newblock


\bibitem[\protect\citeauthoryear{Dal~Lago and Gaboardi}{Dal~Lago and
  Gaboardi}{2011}]%
        {LICS:LG11}
\bibfield{author}{\bibinfo{person}{U. Dal~Lago} {and} \bibinfo{person}{M.
  Gaboardi}.} \bibinfo{year}{2011}\natexlab{}.
\newblock \showarticletitle{{Linear Dependent Types and Relative
  Completeness}}. In \bibinfo{booktitle}{\emph{Logic in Computer Science
  (LICS'11)}}.
\newblock


\bibitem[\protect\citeauthoryear{Dal~Lago and Ghyselen}{Dal~Lago and
  Ghyselen}{2018}]%
        {kn:LG18}
\bibfield{author}{\bibinfo{person}{U. Dal~Lago} {and} \bibinfo{person}{A.
  Ghyselen}.} \bibinfo{year}{2018}\natexlab{}.
\newblock \showarticletitle{{On Linear Dependent Types and Probabilistic
  Termination}}. In \bibinfo{booktitle}{\emph{International Workshop on
  Developments in Implicit Computational Complexity}}.
\newblock


\bibitem[\protect\citeauthoryear{Dal~Lago and Grellois}{Dal~Lago and
  Grellois}{2019}]%
        {TOPLAS:LG19}
\bibfield{author}{\bibinfo{person}{U. Dal~Lago} {and} \bibinfo{person}{C.
  Grellois}.} \bibinfo{year}{2019}\natexlab{}.
\newblock \showarticletitle{{Probabilistic Termination by Monadic Affine Sized
  Typing}}.
\newblock \bibinfo{journal}{\emph{Trans.\ on Prog.\ Lang.\ and Syst.}}
  \bibinfo{volume}{41} (\bibinfo{date}{June} \bibinfo{year}{2019}).
\newblock
Issue 2.


\bibitem[\protect\citeauthoryear{Dal~Lago and Petit}{Dal~Lago and
  Petit}{2013}]%
        {POPL:LP13}
\bibfield{author}{\bibinfo{person}{U. Dal~Lago} {and} \bibinfo{person}{B.
  Petit}.} \bibinfo{year}{2013}\natexlab{}.
\newblock \showarticletitle{{The Geometry of Types}}. In
  \bibinfo{booktitle}{\emph{Princ.\ of Prog.\ Lang. (POPL'13)}}.
\newblock


\bibitem[\protect\citeauthoryear{Danielsson}{Danielsson}{2008}]%
        {POPL:Danielsson08}
\bibfield{author}{\bibinfo{person}{N.~A. Danielsson}.}
  \bibinfo{year}{2008}\natexlab{}.
\newblock \showarticletitle{{Lightweight Semiformal Time Complexity Analysis
  for Purely Functional Data Structures}}. In \bibinfo{booktitle}{\emph{Princ.\
  of Prog.\ Lang. (POPL'08)}}.
\newblock


\bibitem[\protect\citeauthoryear{Danner, Licata, and Ramyaa}{Danner
  et~al\mbox{.}}{2015}]%
        {ICFP:DLR15}
\bibfield{author}{\bibinfo{person}{N. Danner}, \bibinfo{person}{D.~R. Licata},
  {and} \bibinfo{person}{R. Ramyaa}.} \bibinfo{year}{2015}\natexlab{}.
\newblock \showarticletitle{{Denotational Cost Semantics for Functional
  Languages with Inductive Types}}. In \bibinfo{booktitle}{\emph{Int.\ Conf.\
  on Functional Programming (ICFP'15)}}.
\newblock


\bibitem[\protect\citeauthoryear{Djuric}{Djuric}{2019}]%
        {misc:Djuric19}
\bibfield{author}{\bibinfo{person}{D. Djuric}.}
  \bibinfo{year}{2019}\natexlab{}.
\newblock \bibinfo{title}{{Billions of Random Numbers in a Blink of an Eye}}.
\newblock \bibinfo{howpublished}{Available on
  \url{https://dragan.rocks/articles/19/Billion-random-numbers-blink-eye-Clojure}}.
\newblock


\bibitem[\protect\citeauthoryear{Ferrer~Fioriti and Hermanns}{Ferrer~Fioriti
  and Hermanns}{2015}]%
        {POPL:FH15}
\bibfield{author}{\bibinfo{person}{L.~M. Ferrer~Fioriti} {and}
  \bibinfo{person}{H. Hermanns}.} \bibinfo{year}{2015}\natexlab{}.
\newblock \showarticletitle{{Probabilistic Termination: Soundness,
  Completeness, and Compositionality}}. In \bibinfo{booktitle}{\emph{Princ.\ of
  Prog.\ Lang. (POPL'15)}}.
\newblock


\bibitem[\protect\citeauthoryear{Filieri, P{\u a}s{\u a}reanu, and
  Visser}{Filieri et~al\mbox{.}}{2013}]%
        {ICSE:FPV13}
\bibfield{author}{\bibinfo{person}{A. Filieri}, \bibinfo{person}{C.~S. P{\u
  a}s{\u a}reanu}, {and} \bibinfo{person}{W. Visser}.}
  \bibinfo{year}{2013}\natexlab{}.
\newblock \showarticletitle{{Reliability Analysis in Symbolic Pathfinder}}. In
  \bibinfo{booktitle}{\emph{Int.\ Conf.\ on Softw.\ Eng. (ICSE'13)}}.
\newblock


\bibitem[\protect\citeauthoryear{Filieri, P{\u a}s{\u a}reanu, Visser, and
  Geldenhuys}{Filieri et~al\mbox{.}}{2014}]%
        {FSE:FPV14}
\bibfield{author}{\bibinfo{person}{A. Filieri}, \bibinfo{person}{C.~S. P{\u
  a}s{\u a}reanu}, \bibinfo{person}{W. Visser}, {and} \bibinfo{person}{J.
  Geldenhuys}.} \bibinfo{year}{2014}\natexlab{}.
\newblock \showarticletitle{{Statistical Symbolic Execution with Informed
  Sampling}}. In \bibinfo{booktitle}{\emph{Found.\ of Softw.\ Eng. (FSE'14)}}.
\newblock


\bibitem[\protect\citeauthoryear{Flores{-}Montoya and
  H{\"{a}}hnle}{Flores{-}Montoya and H{\"{a}}hnle}{2014}]%
        {APLAS:FH14}
\bibfield{author}{\bibinfo{person}{A. Flores{-}Montoya} {and}
  \bibinfo{person}{R. H{\"{a}}hnle}.} \bibinfo{year}{2014}\natexlab{}.
\newblock \showarticletitle{{Resource Analysis of Complex Programs with Cost
  Equations}}. In \bibinfo{booktitle}{\emph{Asian Symp.\ on Prog.\ Lang.\ and
  Systems (APLAS'14)}}.
\newblock


\bibitem[\protect\citeauthoryear{Frohn, Naaf, Hensel, Brockschmidt, and
  Giesl}{Frohn et~al\mbox{.}}{2016}]%
        {IJCAR:FNH16}
\bibfield{author}{\bibinfo{person}{F. Frohn}, \bibinfo{person}{M. Naaf},
  \bibinfo{person}{J. Hensel}, \bibinfo{person}{M. Brockschmidt}, {and}
  \bibinfo{person}{J. Giesl}.} \bibinfo{year}{2016}\natexlab{}.
\newblock \showarticletitle{{Lower Runtime Bounds for Integer Programs}}. In
  \bibinfo{booktitle}{\emph{Int.\ Joint Conf.\ on Automated Reasoning
  (IJCAR'16)}}.
\newblock


\bibitem[\protect\citeauthoryear{Gaboardi, Haeberlen, Hsu, Narayan, and
  Pierce}{Gaboardi et~al\mbox{.}}{2013}]%
        {POPL:GHH13}
\bibfield{author}{\bibinfo{person}{M. Gaboardi}, \bibinfo{person}{A.
  Haeberlen}, \bibinfo{person}{J. Hsu}, \bibinfo{person}{A. Narayan}, {and}
  \bibinfo{person}{B.~C. Pierce}.} \bibinfo{year}{2013}\natexlab{}.
\newblock \showarticletitle{{Linear Dependent Types for Differential Privacy}}.
  In \bibinfo{booktitle}{\emph{Princ.\ of Prog.\ Lang. (POPL'13)}}.
\newblock


\bibitem[\protect\citeauthoryear{Goodman and Stuhlm{\"u}ller}{Goodman and
  Stuhlm{\"u}ller}{2014}]%
        {misc:dippl}
\bibfield{author}{\bibinfo{person}{N.~D. Goodman} {and} \bibinfo{person}{A.
  Stuhlm{\"u}ller}.} \bibinfo{year}{2014}\natexlab{}.
\newblock \bibinfo{title}{{The Design and Implementation of Probabilistic
  Programming Languages}}.
\newblock \bibinfo{howpublished}{Available on \url{http://dippl.org}}.
\newblock


\bibitem[\protect\citeauthoryear{Gordon, Henzinger, Nori, and Rajamani}{Gordon
  et~al\mbox{.}}{2014}]%
        {gordon2014probabilistic}
\bibfield{author}{\bibinfo{person}{Andrew~D Gordon}, \bibinfo{person}{Thomas~A
  Henzinger}, \bibinfo{person}{Aditya~V Nori}, {and} \bibinfo{person}{Sriram~K
  Rajamani}.} \bibinfo{year}{2014}\natexlab{}.
\newblock \showarticletitle{Probabilistic programming}.
\newblock In \bibinfo{booktitle}{\emph{Proceedings of the on Future of Software
  Engineering}}. \bibinfo{pages}{167--181}.
\newblock


\bibitem[\protect\citeauthoryear{Gulwani}{Gulwani}{2009}]%
        {CAV:Gulwani09}
\bibfield{author}{\bibinfo{person}{S. Gulwani}.}
  \bibinfo{year}{2009}\natexlab{}.
\newblock \showarticletitle{{SPEED: Symbolic Complexity Bound Analysis}}. In
  \bibinfo{booktitle}{\emph{Computer Aided Verif. (CAV'09)}}.
\newblock


\bibitem[\protect\citeauthoryear{Hark, Kaminski, Giesl, and Katoen}{Hark
  et~al\mbox{.}}{2020}]%
        {POPL:HKG20}
\bibfield{author}{\bibinfo{person}{M. Hark}, \bibinfo{person}{B.~L. Kaminski},
  \bibinfo{person}{J. Giesl}, {and} \bibinfo{person}{J.-P. Katoen}.}
  \bibinfo{year}{2020}\natexlab{}.
\newblock \showarticletitle{{Aiming Low Is Harder: Induction for Lower Bounds
  in Probabilistic Program Verification}}. In \bibinfo{booktitle}{\emph{Princ.\
  of Prog.\ Lang. (POPL'20)}}.
\newblock


\bibitem[\protect\citeauthoryear{Harper}{Harper}{2016}]%
        {book:PFPL16}
\bibfield{author}{\bibinfo{person}{R. Harper}.}
  \bibinfo{year}{2016}\natexlab{}.
\newblock \bibinfo{booktitle}{\emph{{Practical Foundations for Programming
  Languages}}}.
\newblock \bibinfo{publisher}{Cambridge University Press}.
\newblock


\bibitem[\protect\citeauthoryear{Hoffmann, Aehlig, and Hofmann}{Hoffmann
  et~al\mbox{.}}{2011}]%
        {POPL:HAH11}
\bibfield{author}{\bibinfo{person}{J. Hoffmann}, \bibinfo{person}{K. Aehlig},
  {and} \bibinfo{person}{M. Hofmann}.} \bibinfo{year}{2011}\natexlab{}.
\newblock \showarticletitle{{Multivariate Amortized Resource Analysis}}. In
  \bibinfo{booktitle}{\emph{Princ.\ of Prog.\ Lang. (POPL'11)}}.
\newblock


\bibitem[\protect\citeauthoryear{Hoffmann, Das, and Weng}{Hoffmann
  et~al\mbox{.}}{2017}]%
        {POPL:HDW17}
\bibfield{author}{\bibinfo{person}{J. Hoffmann}, \bibinfo{person}{A. Das},
  {and} \bibinfo{person}{S.-C. Weng}.} \bibinfo{year}{2017}\natexlab{}.
\newblock \showarticletitle{{Towards Automatic Resource Bound Analysis for
  OCaml}}. In \bibinfo{booktitle}{\emph{Princ.\ of Prog.\ Lang. (POPL'17)}}.
\newblock


\bibitem[\protect\citeauthoryear{Hoffmann and Hofmann}{Hoffmann and
  Hofmann}{2010a}]%
        {APLAS:HH10}
\bibfield{author}{\bibinfo{person}{J. Hoffmann} {and} \bibinfo{person}{M.
  Hofmann}.} \bibinfo{year}{2010}\natexlab{a}.
\newblock \showarticletitle{{Amortized Resource Analysis with Polymorphic
  Recursion and Partial Big-Step Operational Semantics}}. In
  \bibinfo{booktitle}{\emph{Asian Symp.\ on Prog.\ Lang.\ and Systems
  (APLAS'10)}}.
\newblock


\bibitem[\protect\citeauthoryear{Hoffmann and Hofmann}{Hoffmann and
  Hofmann}{2010b}]%
        {ESOP:HH10}
\bibfield{author}{\bibinfo{person}{J. Hoffmann} {and} \bibinfo{person}{M.
  Hofmann}.} \bibinfo{year}{2010}\natexlab{b}.
\newblock \showarticletitle{{Amortized Resource Analysis with Polynomial
  Potential}}. In \bibinfo{booktitle}{\emph{European Symp.\ on Programming
  (ESOP'10)}}.
\newblock


\bibitem[\protect\citeauthoryear{Hofmann and Jost}{Hofmann and Jost}{2003}]%
        {POPL:HJ03}
\bibfield{author}{\bibinfo{person}{M. Hofmann} {and} \bibinfo{person}{S.
  Jost}.} \bibinfo{year}{2003}\natexlab{}.
\newblock \showarticletitle{{Static Prediction of Heap Space Usage for
  First-Order Functional Programs}}. In \bibinfo{booktitle}{\emph{Princ.\ of
  Prog.\ Lang. (POPL'03)}}.
\newblock


\bibitem[\protect\citeauthoryear{Hofmann and Moser}{Hofmann and Moser}{2015}]%
        {TLCA:HM15}
\bibfield{author}{\bibinfo{person}{M. Hofmann} {and} \bibinfo{person}{G.
  Moser}.} \bibinfo{year}{2015}\natexlab{}.
\newblock \showarticletitle{{Multivariate Amortised Resource Analysis for Term
  Rewrite Systems}}. In \bibinfo{booktitle}{\emph{Int.\ Conf.\ on Typed Lambda
  Calculi and Applications (TLCA'15)}}.
\newblock


\bibitem[\protect\citeauthoryear{Hofmann and Moser}{Hofmann and Moser}{2018}]%
        {CoRR:abs/cs/1807.08242}
\bibfield{author}{\bibinfo{person}{M. Hofmann} {and} \bibinfo{person}{G.
  Moser}.} \bibinfo{year}{2018}\natexlab{}.
\newblock \bibinfo{booktitle}{\emph{{Analysis of Logarithmic Amortised
  Complexity}}}.
\newblock \bibinfo{type}{{T}echnical {R}eport}. \bibinfo{institution}{Computing
  Research Repository}.
\newblock


\bibitem[\protect\citeauthoryear{Jost, Hammond, Loidl, and Hofmann}{Jost
  et~al\mbox{.}}{2010}]%
        {POPL:JHL10}
\bibfield{author}{\bibinfo{person}{S. Jost}, \bibinfo{person}{K. Hammond},
  \bibinfo{person}{H.-W. Loidl}, {and} \bibinfo{person}{M. Hofmann}.}
  \bibinfo{year}{2010}\natexlab{}.
\newblock \showarticletitle{{Static Determination of Quantitative Resource
  Usage for Higher-Order Programs}}. In \bibinfo{booktitle}{\emph{Princ.\ of
  Prog.\ Lang. (POPL'10)}}.
\newblock


\bibitem[\protect\citeauthoryear{Jost, Loidl, Hammond, Scaife, and
  Hofmann}{Jost et~al\mbox{.}}{2009}]%
        {FM:JLH09}
\bibfield{author}{\bibinfo{person}{S. Jost}, \bibinfo{person}{H.-W. Loidl},
  \bibinfo{person}{K. Hammond}, \bibinfo{person}{N. Scaife}, {and}
  \bibinfo{person}{M. Hofmann}.} \bibinfo{year}{2009}\natexlab{}.
\newblock \showarticletitle{{Carbon Credits for Resource-Bounded Computations
  using Amortised Analysis}}. In \bibinfo{booktitle}{\emph{Symp.\ on Form.\
  Meth. (FM'09)}}.
\newblock


\bibitem[\protect\citeauthoryear{Kahn and Hoffmann}{Kahn and Hoffmann}{2020}]%
        {FoSSaCS:KH20}
\bibfield{author}{\bibinfo{person}{D.~M. Kahn} {and} \bibinfo{person}{J.
  Hoffmann}.} \bibinfo{year}{2020}\natexlab{}.
\newblock \showarticletitle{{Exponential Automatic Amortized Resource
  Analysis}}. In \bibinfo{booktitle}{\emph{International Conference on
  Foundations of Software Science and Computation Structures (FoSSaCS'20)}}.
\newblock


\bibitem[\protect\citeauthoryear{Kaminski, Katoen, Matheja, and
  Olmedo}{Kaminski et~al\mbox{.}}{2016}]%
        {ESOP:KKM16}
\bibfield{author}{\bibinfo{person}{B.~L. Kaminski}, \bibinfo{person}{J.-P.
  Katoen}, \bibinfo{person}{C. Matheja}, {and} \bibinfo{person}{F. Olmedo}.}
  \bibinfo{year}{2016}\natexlab{}.
\newblock \showarticletitle{{Weakest Precondition Reasoning for Expected
  Run---Times of Probabilistic Programs}}. In
  \bibinfo{booktitle}{\emph{European Symp.\ on Programming (ESOP'16)}}.
\newblock


\bibitem[\protect\citeauthoryear{Kavvos, Morehouse, Licata, and Danner}{Kavvos
  et~al\mbox{.}}{2020}]%
        {POPL:KML20}
\bibfield{author}{\bibinfo{person}{G.~A. Kavvos}, \bibinfo{person}{E.
  Morehouse}, \bibinfo{person}{D.~R. Licata}, {and} \bibinfo{person}{N.
  Danner}.} \bibinfo{year}{2020}\natexlab{}.
\newblock \showarticletitle{{Recurrence Extraction for Functional Programs
  through Call-by-Push-Value}}. In \bibinfo{booktitle}{\emph{Princ.\ of Prog.\
  Lang. (POPL'20)}}.
\newblock


\bibitem[\protect\citeauthoryear{Kincaid, Breck, Boroujeni, and Reps}{Kincaid
  et~al\mbox{.}}{2017}]%
        {PLDI:KBB17}
\bibfield{author}{\bibinfo{person}{Z. Kincaid}, \bibinfo{person}{J. Breck},
  \bibinfo{person}{A.~F. Boroujeni}, {and} \bibinfo{person}{T. Reps}.}
  \bibinfo{year}{2017}\natexlab{}.
\newblock \showarticletitle{{Compositional Recurrence Analysis Revisited}}. In
  \bibinfo{booktitle}{\emph{Prog.\ Lang.\ Design and Impl. (PLDI'17)}}.
\newblock


\bibitem[\protect\citeauthoryear{Knoth, Wang, Polikarpova, and Hoffmann}{Knoth
  et~al\mbox{.}}{2019}]%
        {PLDI:KWP19}
\bibfield{author}{\bibinfo{person}{T. Knoth}, \bibinfo{person}{D. Wang},
  \bibinfo{person}{N. Polikarpova}, {and} \bibinfo{person}{J. Hoffmann}.}
  \bibinfo{year}{2019}\natexlab{}.
\newblock \showarticletitle{{Resource-Guided Program Synthesis}}. In
  \bibinfo{booktitle}{\emph{Prog.\ Lang.\ Design and Impl. (PLDI'19)}}.
\newblock


\bibitem[\protect\citeauthoryear{Knuth and Yao}{Knuth and Yao}{1976}]%
        {knuth1976algorithms}
\bibfield{author}{\bibinfo{person}{Donald Knuth} {and} \bibinfo{person}{Andrew
  Yao}.} \bibinfo{year}{1976}\natexlab{}.
\newblock \bibinfo{title}{Algorithms and Complexity: New Directions and Recent
  Results, chapter The complexity of nonuniform random number generation}.
\newblock
\newblock


\bibitem[\protect\citeauthoryear{Kozen}{Kozen}{1981}]%
        {JCSS:Kozen81}
\bibfield{author}{\bibinfo{person}{D. Kozen}.} \bibinfo{year}{1981}\natexlab{}.
\newblock \showarticletitle{{Semantics of Probabilistic Programs}}.
\newblock \bibinfo{journal}{\emph{J.\ Comput.\ Syst.\ Sci.}}
  \bibinfo{volume}{22} (\bibinfo{date}{June} \bibinfo{year}{1981}).
\newblock
Issue 3.


\bibitem[\protect\citeauthoryear{Kura, Urabe, and Hasuo}{Kura
  et~al\mbox{.}}{2019}]%
        {TACAS:KUH19}
\bibfield{author}{\bibinfo{person}{S. Kura}, \bibinfo{person}{N. Urabe}, {and}
  \bibinfo{person}{I. Hasuo}.} \bibinfo{year}{2019}\natexlab{}.
\newblock \showarticletitle{{Tail Probability for Randomized Program Runtimes
  via Martingales for Higher Moments}}. In \bibinfo{booktitle}{\emph{Tools and
  Algs.\ for the Construct.\ and Anal.\ of Syst. (TACAS'19)}}.
\newblock


\bibitem[\protect\citeauthoryear{Lew, Cusumano-Towner, Sherman, Carbin, and
  Mansinghka}{Lew et~al\mbox{.}}{2020}]%
        {POPL:LCS20}
\bibfield{author}{\bibinfo{person}{A.~K. Lew}, \bibinfo{person}{M.~F.
  Cusumano-Towner}, \bibinfo{person}{B. Sherman}, \bibinfo{person}{M. Carbin},
  {and} \bibinfo{person}{V.~K. Mansinghka}.} \bibinfo{year}{2020}\natexlab{}.
\newblock \showarticletitle{{Trace Types and Denotational Semantics for Sound
  Programmable Inference in Probabilistic Languages}}. In
  \bibinfo{booktitle}{\emph{Princ.\ of Prog.\ Lang. (POPL'20)}}.
\newblock


\bibitem[\protect\citeauthoryear{MacLeod, Ivezi{\'c}, Kochanek, Koz{\l}owski,
  Kelly, Bullock, Kimball, Sesar, Westman, Brooks, et~al\mbox{.}}{MacLeod
  et~al\mbox{.}}{2010}]%
        {macleod2010modeling}
\bibfield{author}{\bibinfo{person}{Ch~L MacLeod}, \bibinfo{person}{{\v{Z}}
  Ivezi{\'c}}, \bibinfo{person}{CS Kochanek}, \bibinfo{person}{S Koz{\l}owski},
  \bibinfo{person}{B Kelly}, \bibinfo{person}{E Bullock}, \bibinfo{person}{A
  Kimball}, \bibinfo{person}{B Sesar}, \bibinfo{person}{D Westman},
  \bibinfo{person}{K Brooks}, {et~al\mbox{.}}} \bibinfo{year}{2010}\natexlab{}.
\newblock \showarticletitle{Modeling the time variability of SDSS stripe 82
  quasars as a damped random walk}.
\newblock \bibinfo{journal}{\emph{The Astrophysical Journal}}
  \bibinfo{volume}{721}, \bibinfo{number}{2} (\bibinfo{year}{2010}),
  \bibinfo{pages}{1014}.
\newblock


\bibitem[\protect\citeauthoryear{Mansinghka, Schaechtle, Handa, Radul, Chen,
  and Rinard}{Mansinghka et~al\mbox{.}}{2018}]%
        {PLDI:MSH18}
\bibfield{author}{\bibinfo{person}{V.~K. Mansinghka}, \bibinfo{person}{U.
  Schaechtle}, \bibinfo{person}{S. Handa}, \bibinfo{person}{A. Radul},
  \bibinfo{person}{Y. Chen}, {and} \bibinfo{person}{M.~C. Rinard}.}
  \bibinfo{year}{2018}\natexlab{}.
\newblock \showarticletitle{{Probabilistic Programming with Programmable
  Inference}}. In \bibinfo{booktitle}{\emph{Prog.\ Lang.\ Design and Impl.
  (PLDI'18)}}.
\newblock


\bibitem[\protect\citeauthoryear{McIver and Morgan}{McIver and Morgan}{2005}]%
        {book:MM05}
\bibfield{author}{\bibinfo{person}{A.~K. McIver} {and} \bibinfo{person}{C.~C.
  Morgan}.} \bibinfo{year}{2005}\natexlab{}.
\newblock \bibinfo{booktitle}{\emph{{Abstraction, Refinement and Proof for
  Probabilistic Systems}}}.
\newblock \bibinfo{publisher}{Springer Science+Business Media, Inc.}
\newblock


\bibitem[\protect\citeauthoryear{Meese and Rogoff}{Meese and Rogoff}{1983}]%
        {meese1983empirical}
\bibfield{author}{\bibinfo{person}{Richard~A Meese} {and}
  \bibinfo{person}{Kenneth Rogoff}.} \bibinfo{year}{1983}\natexlab{}.
\newblock \showarticletitle{Empirical exchange rate models of the seventies: Do
  they fit out of sample?}
\newblock \bibinfo{journal}{\emph{Journal of international economics}}
  \bibinfo{volume}{14}, \bibinfo{number}{1-2} (\bibinfo{year}{1983}),
  \bibinfo{pages}{3--24}.
\newblock


\bibitem[\protect\citeauthoryear{Ngo, Carbonneaux, and Hoffmann}{Ngo
  et~al\mbox{.}}{2018}]%
        {PLDI:NCH18}
\bibfield{author}{\bibinfo{person}{V.~C. Ngo}, \bibinfo{person}{Q.
  Carbonneaux}, {and} \bibinfo{person}{J. Hoffmann}.}
  \bibinfo{year}{2018}\natexlab{}.
\newblock \showarticletitle{{Bounded Expectations: Resource Analysis for
  Probabilistic Programs}}. In \bibinfo{booktitle}{\emph{Prog.\ Lang.\ Design
  and Impl. (PLDI'18)}}.
\newblock


\bibitem[\protect\citeauthoryear{Ngo, Dehesa-Azuara, Fredrikson, and
  Hoffmann}{Ngo et~al\mbox{.}}{2017}]%
        {SP:NDF17}
\bibfield{author}{\bibinfo{person}{V.~C. Ngo}, \bibinfo{person}{Mario
  Dehesa-Azuara}, \bibinfo{person}{M. Fredrikson}, {and} \bibinfo{person}{J.
  Hoffmann}.} \bibinfo{year}{2017}\natexlab{}.
\newblock \showarticletitle{{Verifying and Synthesizing Constant-Resource
  Implementations with Types}}. In \bibinfo{booktitle}{\emph{Symp.\ on Sec.\
  and Privacy (SP'17)}}.
\newblock


\bibitem[\protect\citeauthoryear{Nipkow}{Nipkow}{2015}]%
        {ITP:Nipkow15}
\bibfield{author}{\bibinfo{person}{T. Nipkow}.}
  \bibinfo{year}{2015}\natexlab{}.
\newblock \showarticletitle{{Amortized Complexity Verified}}. In
  \bibinfo{booktitle}{\emph{Interactive Theorem Proving (ITP'15)}}.
\newblock


\bibitem[\protect\citeauthoryear{Noller, Kersten, and P{\u a}s{\u
  a}reanu}{Noller et~al\mbox{.}}{2018}]%
        {ISSTA:KNP18}
\bibfield{author}{\bibinfo{person}{Y. Noller}, \bibinfo{person}{R. Kersten},
  {and} \bibinfo{person}{C.~S. P{\u a}s{\u a}reanu}.}
  \bibinfo{year}{2018}\natexlab{}.
\newblock \showarticletitle{{Badger: Complexity Analysis with Fuzzing and
  Symbolic Execution}}. In \bibinfo{booktitle}{\emph{Int.\ Symp.\ on Softw.\
  Testing and Analysis (ISSTA'18)}}.
\newblock


\bibitem[\protect\citeauthoryear{Noschinski, Emmes, and Giesl}{Noschinski
  et~al\mbox{.}}{2013}]%
        {JAR:NEG13}
\bibfield{author}{\bibinfo{person}{L. Noschinski}, \bibinfo{person}{F. Emmes},
  {and} \bibinfo{person}{J. Giesl}.} \bibinfo{year}{2013}\natexlab{}.
\newblock \showarticletitle{{Analyzing Innermost Runtime Complexity of Term
  Rewriting by Dependency Pairs}}.
\newblock \bibinfo{journal}{\emph{J.\ Automated Reasoning}}
  \bibinfo{volume}{51} (\bibinfo{date}{June} \bibinfo{year}{2013}).
\newblock
Issue 1.


\bibitem[\protect\citeauthoryear{Olmedo, Kaminski, Katoen, and Matheja}{Olmedo
  et~al\mbox{.}}{2016}]%
        {LICS:OKK16}
\bibfield{author}{\bibinfo{person}{F. Olmedo}, \bibinfo{person}{B.~L.
  Kaminski}, \bibinfo{person}{J.-P. Katoen}, {and} \bibinfo{person}{C.
  Matheja}.} \bibinfo{year}{2016}\natexlab{}.
\newblock \showarticletitle{{Reasoning about Recursive Probabilistic
  Programs}}. In \bibinfo{booktitle}{\emph{Logic in Computer Science
  (LICS'16)}}.
\newblock


\bibitem[\protect\citeauthoryear{Petsios, Zhao, Keromytis, and Jana}{Petsios
  et~al\mbox{.}}{2017}]%
        {CCS:PZK17}
\bibfield{author}{\bibinfo{person}{T. Petsios}, \bibinfo{person}{J. Zhao},
  \bibinfo{person}{A.~D. Keromytis}, {and} \bibinfo{person}{S. Jana}.}
  \bibinfo{year}{2017}\natexlab{}.
\newblock \showarticletitle{{SlowFuzz: Automated Domain-Independent Detection
  of Algorithmic Complexity Vulnerabilities}}. In
  \bibinfo{booktitle}{\emph{Conf.\ on Comp.\ and Comm.\ Sec. (CCS'17)}}.
\newblock


\bibitem[\protect\citeauthoryear{Plotkin}{Plotkin}{1977}]%
        {Plotkin77}
\bibfield{author}{\bibinfo{person}{G.~D. Plotkin}.}
  \bibinfo{year}{1977}\natexlab{}.
\newblock \showarticletitle{{LCF Considered as a Programming Language}}.
\newblock \bibinfo{journal}{\emph{Theor. Comput. Sci.}}  \bibinfo{volume}{5}
  (\bibinfo{year}{1977}), \bibinfo{pages}{223--255}.
\newblock


\bibitem[\protect\citeauthoryear{Radicek, Barthe, Gaboardi, Garg, and
  Zuleger}{Radicek et~al\mbox{.}}{2018}]%
        {POPL:RBG18}
\bibfield{author}{\bibinfo{person}{I. Radicek}, \bibinfo{person}{G. Barthe},
  \bibinfo{person}{M. Gaboardi}, \bibinfo{person}{D. Garg}, {and}
  \bibinfo{person}{F. Zuleger}.} \bibinfo{year}{2018}\natexlab{}.
\newblock \showarticletitle{{Monadic Refinements for Relational Cost
  Analysis}}. In \bibinfo{booktitle}{\emph{Princ.\ of Prog.\ Lang. (POPL'18)}}.
\newblock


\bibitem[\protect\citeauthoryear{Ramalingam}{Ramalingam}{1996}]%
        {PLDI:Ramalingam96}
\bibfield{author}{\bibinfo{person}{G. Ramalingam}.}
  \bibinfo{year}{1996}\natexlab{}.
\newblock \showarticletitle{{Data Flow Frequency Analysis}}. In
  \bibinfo{booktitle}{\emph{Prog.\ Lang.\ Design and Impl. (PLDI'96)}}.
\newblock


\bibitem[\protect\citeauthoryear{Reed and Pierce}{Reed and Pierce}{2010}]%
        {ICFP:RP10}
\bibfield{author}{\bibinfo{person}{J. Reed} {and} \bibinfo{person}{B.~C.
  Pierce}.} \bibinfo{year}{2010}\natexlab{}.
\newblock \showarticletitle{{Distance Makes the Types Grow Stronger: A Calculus
  for Differential Privacy}}. In \bibinfo{booktitle}{\emph{Int.\ Conf.\ on
  Functional Programming (ICFP'10)}}.
\newblock


\bibitem[\protect\citeauthoryear{Saad, Freer, Rinard, and Mansinghka}{Saad
  et~al\mbox{.}}{2020}]%
        {POPL:SFR20}
\bibfield{author}{\bibinfo{person}{F.~A. Saad}, \bibinfo{person}{C.~E. Freer},
  \bibinfo{person}{M.~C. Rinard}, {and} \bibinfo{person}{V.~K. Mansinghka}.}
  \bibinfo{year}{2020}\natexlab{}.
\newblock \showarticletitle{{Optimal Approximate Sampling from Discrete
  Probability Distributions}}. In \bibinfo{booktitle}{\emph{Princ.\ of Prog.\
  Lang. (POPL'20)}}.
\newblock


\bibitem[\protect\citeauthoryear{Sinn, Zuleger, and Veith}{Sinn
  et~al\mbox{.}}{2014}]%
        {CAV:SZV14}
\bibfield{author}{\bibinfo{person}{M. Sinn}, \bibinfo{person}{F. Zuleger},
  {and} \bibinfo{person}{H. Veith}.} \bibinfo{year}{2014}\natexlab{}.
\newblock \showarticletitle{{A Simple and Scalable Approach to Bound Analysis
  and Amortized Complexity Analysis}}. In \bibinfo{booktitle}{\emph{Computer
  Aided Verif. (CAV'14)}}.
\newblock


\bibitem[\protect\citeauthoryear{Tarjan}{Tarjan}{1985}]%
        {kn:Tarjan85}
\bibfield{author}{\bibinfo{person}{R.~E. Tarjan}.}
  \bibinfo{year}{1985}\natexlab{}.
\newblock \showarticletitle{{Amortized Computational Complexity}}.
\newblock \bibinfo{journal}{\emph{SIAM J.\ Algebraic Discrete Methods}}
  \bibinfo{volume}{6} (\bibinfo{date}{August} \bibinfo{year}{1985}).
\newblock
Issue 2.


\bibitem[\protect\citeauthoryear{Tassarotti and Harper}{Tassarotti and
  Harper}{2019}]%
        {TassarottiH19}
\bibfield{author}{\bibinfo{person}{Joseph Tassarotti} {and}
  \bibinfo{person}{Robert Harper}.} \bibinfo{year}{2019}\natexlab{}.
\newblock \showarticletitle{{A Separation Logic for Concurrent Randomized
  Programs}}.
\newblock \bibinfo{journal}{\emph{Proc. ACM Program. Lang.}}
  \bibinfo{volume}{3}, \bibinfo{number}{POPL}, Article
  \bibinfo{articleno}{Article 64} (\bibinfo{date}{Jan.} \bibinfo{year}{2019}),
  \bibinfo{numpages}{30}~pages.
\newblock
\urldef\tempurl%
\url{https://doi.org/10.1145/3290377}
\showDOI{\tempurl}


\bibitem[\protect\citeauthoryear{Vasconcelos}{Vasconcelos}{2008}]%
        {phd:Vasconcelos08}
\bibfield{author}{\bibinfo{person}{P.~B. Vasconcelos}.}
  \bibinfo{year}{2008}\natexlab{}.
\newblock \emph{\bibinfo{title}{{Space Cost Analysis Using Sized Types}}}.
\newblock \bibinfo{thesistype}{Ph.D. Dissertation}. \bibinfo{school}{School of
  Computer Science, University of St Andrews}.
\newblock


\bibitem[\protect\citeauthoryear{Visser}{Visser}{1997}]%
        {visser1997using}
\bibfield{author}{\bibinfo{person}{Andre~W Visser}.}
  \bibinfo{year}{1997}\natexlab{}.
\newblock \showarticletitle{Using random walk models to simulate the vertical
  distribution of particles in a turbulent water column}.
\newblock \bibinfo{journal}{\emph{Marine Ecology Progress Series}}
  \bibinfo{volume}{158} (\bibinfo{year}{1997}), \bibinfo{pages}{275--281}.
\newblock


\bibitem[\protect\citeauthoryear{Walker}{Walker}{2002}]%
        {kn:Walker02}
\bibfield{author}{\bibinfo{person}{D. Walker}.}
  \bibinfo{year}{2002}\natexlab{}.
\newblock \showarticletitle{{Substructural Type Systems}}.
\newblock In \bibinfo{booktitle}{\emph{Advanced Topics in Types and Programming
  Languages}}. \bibinfo{publisher}{MIT Press}.
\newblock


\bibitem[\protect\citeauthoryear{Wang and Hoffmann}{Wang and Hoffmann}{2019}]%
        {POPL:WH19}
\bibfield{author}{\bibinfo{person}{D. Wang} {and} \bibinfo{person}{J.
  Hoffmann}.} \bibinfo{year}{2019}\natexlab{}.
\newblock \showarticletitle{{Type-Guided Worst-Case Input Generation}}. In
  \bibinfo{booktitle}{\emph{Princ.\ of Prog.\ Lang. (POPL'19)}}.
\newblock


\bibitem[\protect\citeauthoryear{Wang, Hoffmann, and Reps}{Wang
  et~al\mbox{.}}{2020}]%
        {WangHR20}
\bibfield{author}{\bibinfo{person}{Di Wang}, \bibinfo{person}{Jan Hoffmann},
  {and} \bibinfo{person}{Thomas Reps}.} \bibinfo{year}{2020}\natexlab{}.
\newblock \bibinfo{title}{{Tail Bound Analysis for Probabilistic Programs via
  Central Moments}}.
\newblock
\newblock
\showeprint[arxiv]{cs.PL/2001.10150}


\bibitem[\protect\citeauthoryear{Wang, Fu, Goharshady, Chatterjee, Qin, and
  Shi}{Wang et~al\mbox{.}}{2019}]%
        {PLDI:WFG19}
\bibfield{author}{\bibinfo{person}{P. Wang}, \bibinfo{person}{H. Fu},
  \bibinfo{person}{A.~K. Goharshady}, \bibinfo{person}{K. Chatterjee},
  \bibinfo{person}{X. Qin}, {and} \bibinfo{person}{W. Shi}.}
  \bibinfo{year}{2019}\natexlab{}.
\newblock \showarticletitle{{Cost Analysis of Nondeterministic Probabilistic
  Programs}}. In \bibinfo{booktitle}{\emph{Prog.\ Lang.\ Design and Impl.
  (PLDI'19)}}.
\newblock


\bibitem[\protect\citeauthoryear{Wang, Wang, and Chlipala}{Wang
  et~al\mbox{.}}{2017}]%
        {OOPSLA:WWC17}
\bibfield{author}{\bibinfo{person}{P. Wang}, \bibinfo{person}{D. Wang}, {and}
  \bibinfo{person}{A. Chlipala}.} \bibinfo{year}{2017}\natexlab{}.
\newblock \showarticletitle{{TiML: A Functional Language for Practical
  Complexity Analysis with Invariants}}. In
  \bibinfo{booktitle}{\emph{Object-Oriented Prog., Syst., Lang., and
  Applications (OOPSLA'17)}}.
\newblock


\bibitem[\protect\citeauthoryear{Williams}{Williams}{1991}]%
        {book:Williams91}
\bibfield{author}{\bibinfo{person}{D. Williams}.}
  \bibinfo{year}{1991}\natexlab{}.
\newblock \bibinfo{booktitle}{\emph{{Probability with Martingales}}}.
\newblock \bibinfo{publisher}{Cambridge University Press}.
\newblock


\bibitem[\protect\citeauthoryear{Wingate and Weber}{Wingate and Weber}{2013}]%
        {CoRR:abs/stat/1301.1299}
\bibfield{author}{\bibinfo{person}{D. Wingate} {and} \bibinfo{person}{T.
  Weber}.} \bibinfo{year}{2013}\natexlab{}.
\newblock \bibinfo{booktitle}{\emph{{Automated Variational Inference in
  Probabilistic Programming}}}.
\newblock \bibinfo{type}{{T}echnical {R}eport}. \bibinfo{institution}{Computing
  Research Repository}.
\newblock


\bibitem[\protect\citeauthoryear{Xi}{Xi}{2002}]%
        {kn:Xi02}
\bibfield{author}{\bibinfo{person}{H. Xi}.} \bibinfo{year}{2002}\natexlab{}.
\newblock \showarticletitle{{Dependent Types for Program Termination
  Verification}}.
\newblock \bibinfo{journal}{\emph{J.\ Higher-Order and Symbolic Comp.}}
  \bibinfo{volume}{15} (\bibinfo{year}{2002}).
\newblock
Issue 1.


\bibitem[\protect\citeauthoryear{Zuleger, Sinn, Gulwani, and Veith}{Zuleger
  et~al\mbox{.}}{2011}]%
        {SAS:ZSG11}
\bibfield{author}{\bibinfo{person}{F. Zuleger}, \bibinfo{person}{M. Sinn},
  \bibinfo{person}{S. Gulwani}, {and} \bibinfo{person}{H. Veith}.}
  \bibinfo{year}{2011}\natexlab{}.
\newblock \showarticletitle{{Bound Analysis of Imperative Programs with the
  Size-change Abstraction}}. In \bibinfo{booktitle}{\emph{Static Analysis Symp.
  (SAS'11)}}.
\newblock


\end{thebibliography}

\versioning{
\newpage
\appendix

\section{Proofs}
\label{sec:appendix:proofs}

\subsection{\cref{the:soundness}}

\begin{proof}
  It suffices to prove for every $n \in \bbN$, if $\steprel{V}{e}{\mu}{n}$, then
  \[
  \pot{V }{ \Gamma} + q \ge \sum_{(v_0,q_0) } \mu(v_0,q_0) \cdot (\pot{v_0 }{ A} + q_0).
  \]
  Proceed by induction on $n$ with inversion on $\steprel{V}{e}{\mu}{n}$ then inner induction on $\typingrel{\Gamma}{q}{e}{A}$. We show the interesting cases below.
  \begin{itemize}
    \item If $n=0$, then $\mu = \mathbf{0}$. Straightforward.
    \item Suppose the lemma holds for some $n \in \bbN$. Now we consider the case for $n+1$. Below are the proofs for \textsc{L:Let}, \textsc{L:Flip}, \textsc{L:Prob}, and \textsc{L:FlipS}.
    \begin{itemize}
      \item \textsc{(L:Let)}
      By assumption, we have $\typingrel{\Gamma_1}{q}{e_1}{\annoA{\tau}{p}}$, $\typingrel{\Gamma_2,x:\tau}{p}{e_2}{A}$, and $\Gamma = \Gamma_1,\Gamma_2$ for some $\Gamma_1,\Gamma_2$.
      By inversion, we have $\steprel{V}{e_1}{\mu_1}{n}$, for all $(v_1,q_1) \in \mathrm{supp}(\mu_1)$, $\steprel{V,x\mapsto v_1}{e_2}{\mu_{(v_1,q_1)}}{n}$, and $\mu = \sum_{(v_1,q_1)} \sum_{(v_2,q_2)}  \mu_1(v_1,q_1) \cdot \mu_{(v_1,q_1)}(v_2,q_2) \cdot \delta(v_2,q_1+q_2)$.
      By the induction hypothesis, we have
      \begin{small}\begin{equation*}
      \pot{V }{ \Gamma_1} + q \ge \sum_{(v_1,q_1)} \mu_1(v_1,q_1) \cdot (\pot{v_1 }{ \annoA{\tau}{p}} + q_1)  =  \sum_{(v_1,q_1) } \mu_1(v_1,q_1) \cdot (\pot{v_1}{ \tau} + q_1 + p),
      \end{equation*}\end{small}
      and also for all $(v_1,q_1) \in \mathrm{supp}(\mu_1)$,
      \begin{small}\[
      \pot{V,x \mapsto v_1 }{ \Gamma_2,x:\tau} +p \ge \sum_{(v_2,q_2) } \mu_{(v_1,q_1)}(v_2,q_2) \cdot (\pot{v_2 }{ A} + q_2).
      \]\end{small}
      This results in the following chain of inequalities.
      \begin{small}\begin{align*}
      \pot{V}{\Gamma} + q & = \pot{V}{\Gamma_1} + q + \pot{V}{\Gamma_2} \\
      & \ge  \sum_{(v_1,q_1) } \mu_1(v_1,q_1) \cdot (\pot{v_1 }{ \tau} + q_1 + p) + \pot{V}{\Gamma_2} \\
      & \ge \sum_{(v_1,q_1) } \mu_1(v_1,q_1) \cdot (\pot{v_1 }{ \tau} + q_1 + p + \pot{V}{\Gamma_2}) \\
      & = \sum_{(v_1,q_1) } \mu_1(v_1,q_1) \cdot ( q_1 + p + \pot{V, x\mapsto v_1 }{\Gamma_2,x:\tau})  \\
      & \ge \sum_{(v_1,q_1) } \mu_1(v_1,q_1) \cdot (q_1 + \sum_{(v_2,q_2)} \mu_{(v_1,q_1)}(v_2,q_2) \cdot (\pot{v_2 }{ A} + q_2)) \\
      & \ge \sum_{(v_1,q_1) } \mu_1(v_1,q_1) \cdot (\sum_{(v_2,q_2) } \mu_{(v_1,q_1)}(v_2,q_2) \cdot (\pot{v_2 }{ A} + q_2 + q_1)) \\
      & = \sum_{(v_1,q_1) }\sum_{(v_2,q_2)  } \mu_1(v_1,q_1) \cdot \mu_{(v_1,q_1)}(v_2,q_2) \cdot (\pot{v_2 }{ A} + q_2 + q_1).
      \end{align*}\end{small}
      Applying the following identities on the final term, we complete the case.
      \begin{small}\begin{align*}
          & \sum_{(v_0,q_0) } \mu(v_0,q_0) \cdot (\pot{v_0 }{ A} + q_0) \\
          ={} & \sum_{(v_0,q_0)} (\sum_{(v_1,q_1)}\sum_{(v_2,q_2)} \mu_1(v_1,q_1) \cdot \mu_{(v_1,q_1)}(v_2,q_2) \cdot \delta(v_2,q_1+q_2))(v_0,q_0) \cdot (\pot{v_0}{A} + q) \\
          ={} & \sum_{(v_0,q_0)} (\sum_{(v_1,q_1)}\sum_{(v_2,q_2)} \mu_1(v_1,q_1) \cdot \mu_{(v_1,q_1)}(v_2,q_2) \cdot [v_0 = v_2 \wedge q_0 = q_1+q_2]) \cdot (\pot{v_0}{A} + q) \\
          ={} & \sum_{(v_1,q_1)}\sum_{(v_2,q_2)} \mu_1(v_1,q_1) \cdot \mu_{(v_1,q_1)}(v_2,q_2) \cdot (\pot{v_2}{A} + q_1+q_2).
      \end{align*}\end{small}
      
      \item \textsc{(L:Flip)}
      By assumption, we have $\share{\Gamma}{p \times \Gamma_1}{(1-p) \times \Gamma_2}$, $q = p \cdot q_1 + (1-p) \cdot q_2$, $\typingrel{\Gamma_1}{q_1}{e_1}{A}$, and $\typingrel{\Gamma_2}{q_2}{e_2}{A}$.
      By inversion, we have $\steprel{V}{e_1}{\mu_1}{n}$, $\steprel{V}{e_2}{\mu_2}{n}$, and $\mu = p \cdot \mu_1 + (1-p) \cdot \mu_2$.
      By the induction hypothesis, we have
      \begin{small}\begin{align*}
          \pot{V }{ \Gamma_1} + q_1 &\ge \sum_{(v_0,q_0)} \mu_1(v_0,q_0) \cdot (\pot{v_0 }{ A} + q_0), & \pot{V }{ \Gamma_2} + q_2 & \ge \sum_{(v_0,q_0) } \mu_2(v_0,q_0) \cdot (\pot{v_0}{ A} + q_0).
      \end{align*}\end{small}
      Thus we conclude this case by
      \begin{small}\begin{align*}
          \pot{V }{\Gamma} + q & = \pot{V }{ p \times \Gamma_1} + \pot{V }{ (1-p)\times \Gamma_2} + p \cdot q_1 + (1-p) \cdot q_2 \\
          & = p \cdot \pot{V }{ \Gamma_1} + (1-p) \cdot \pot{V }{ \Gamma_2} + p \cdot q_1 + (1-p) \cdot q_2 \\
          & \ge p \cdot  (\sum_{(v_0,q_0)} \mu_1(v_0,q_0) \cdot (\pot{v_0 }{ A} + q_0)) + (1-p) \cdot (\sum_{(v_0,q_0)} \mu_2(v_0,q_0) \cdot (\pot{v_0}{ A} + q_0)) \\
          & = \sum_{(v_0,q_0) } (p \cdot \mu_1(v_0,q_0) + (1-p) \cdot \mu_2(v_0,q_0)) \cdot (\pot{v_0}{ A} + q_0) \\
          & = \sum_{(v_0,q_0)} \mu(v_0,q_0) \cdot (\pot{v_0 }{ A} + q_0).
      \end{align*}\end{small}
      
      \item \textsc{(L:Prob)}
      By assumption, we know that $\Gamma = \cdot$, $q = p \cdot q_H + (1-p) \cdot q_T$, and $A = \annoA{\probA{q_H}{q_T}}{0}$ for some $q_H,q_T \in \bbQ_{\ge 0}$.
      By inversion, we have $\mu = \delta(\probC{p},0)$.
      Thus we conclude this case by
      \begin{small}\begin{align*}
        \pot{V}{\Gamma} + q & = p \cdot q_H + (1-p) \cdot q_T = \pot{\probV{p}}{\probA{q_H}{q_T}} + 0 = \sum_{(v_0,q_0)} \mu(v_0,q_0) \cdot (\pot{v_0}{A} + q_0).
      \end{align*}  
      \end{small}
      
      \item \textsc{(L:FlipS)}
      By assumption, we have $\typingrel{\Gamma}{q+q_H}{e_1}{A}$ and $\typingrel{\Gamma}{q+q_T}{e_2}{A}$.
      By inversion, we have $V(x) = \probV{p}$, $\steprel{V}{e_1}{\mu_1}{n}$, $\steprel{V}{e_2}{\mu_2}{n}$, and $\mu = p \cdot \mu_1 + (1-p) \cdot \mu_2$ for some $p \in [0,1]$.
      By the induction hypothesis, we have
      \begin{small}\begin{align*}
        \pot{V}{\Gamma} + q + q_H & \ge \sum_{(v_0,q_0)} \mu_1(v_0,q_0) \cdot (\pot{v_0}{A} + q_0), & \pot{V}{\Gamma} + q + q_T & \ge \sum_{(v_0,q_0)} \mu_2(v_0,q_0) \cdot (\pot{v_0}{A} + q_0).
      \end{align*}  
      \end{small}
      Thus we conclude this case by
      \begin{small}\begin{align*}
        & \pot{V}{\Gamma,x:\probA{q_H}{q_T}} + q \\
         ={} & \pot{V}{\Gamma} + p \cdot q_H + (1-p) \cdot q_T + q \\
         ={}& (p \cdot \pot{V}{\Gamma} + (1-p) \cdot \pot{V}{\Gamma}) + (p \cdot q_H + (1-p) \cdot q_T) + (p \cdot q + (1-p) \cdot q) \\
        ={}&  p \cdot (\pot{V}{\Gamma} + q + q_H) + (1-p) \cdot (\pot{V}{\Gamma} + q + q_T) \\
        \ge{} &  p \cdot ( \sum_{(v_0,q_0)} \mu_1(v_0,q_0) \cdot (\pot{v_0}{A} + q_0) ) + (1-p) \cdot (\sum_{(v_0,q_0)} \mu_2(v_0,q_0) \cdot (\pot{v_0}{A} + q_0) ) \\
        ={}& \sum_{(v_0,q_0)} (p \cdot \mu_1(v_0,q_0) + (1-p) \cdot \mu_2(v_0,q_0)) \cdot (\pot{v_0}{A } + q_0) \\
        ={} & \sum_{(v_0,q_0)} \mu(v_0,q_0) \cdot (\pot{v_0}{A} + q_0).
      \end{align*}  
      \end{small}
    \end{itemize}
  \end{itemize}
\end{proof}

\subsection{\cref{the:soundness:improved}}

\begin{proof}
  By \cref{lemma:dist def,lemma:dist bound}
  it suffices to prove for every $n \in \bbN$, if $\steprel{V}{e}{\mu}{n}$, then
  \[
  \pot{V }{ \Gamma} + q \ge \sum_{q_0} \mu(\circ,q_0) \cdot q_0 + \sum_{(v_0,q_0) } \mu(v_0,q_0) \cdot (\pot{v_0 }{ A} + q_0).
  \]
  We can still proceed by induction on $n$ with inversion on $\steprel{V}{e}{\mu}{n}$ then inner induction on $\typingrel{\Gamma}{q}{e}{A}$.
  We illustrate cases \textsc{L:Let}, \textsc{L:Flip}, \textsc{L:Prob}, and \textsc{L:FlipS} below.
  \begin{itemize}
    \item \textsc{(L:Let)} 
    By assumption, we have $\typingrel{\Gamma_1}{q}{e_1}{\annoA{\tau}{p}}$, $\typingrel{\Gamma_2,x:\tau}{p}{e_2}{A}$, and $\Gamma = \Gamma_1,\Gamma_2$ for some $\Gamma_1,\Gamma_2$.
      By inversion, we have $\steprel{V}{e_1}{\mu_1}{n}$, for all $(v_1,q_1) \in \mathrm{supp}(\mu_1)$ such that $v_1\neq\circ$, $\steprel{V,x\mapsto v_1}{e_2}{\mu_{(v_1,q_1)}}{n}$, and $\mu = \sum_{q_1} \mu_1(\circ,q_1) \cdot \delta(\circ,q_1) + \sum_{(v_1,q_1):v_1 \neq \circ} \sum_{(v_2,q_2)}  \mu_1(v_1,q_1) \cdot \mu_{(v_1,q_1)}(v_2,q_2) \cdot \delta(v_2,q_1+q_2)$.
      By the induction hypothesis, we have
      \begin{small}\begin{align*}
      \pot{V }{ \Gamma_1} + q & \ge \sum_{q_1} \mu_1(\circ,q_1) \cdot q_1 +  \sum_{(v_1,q_1):v_1\neq\circ} \mu_1(v_1,q_1) \cdot (\pot{v_1 }{ \annoA{\tau}{p}} + q_1) \\
      & = \sum_{q_1} \mu_1(\circ,q_1) \cdot q_1 + \sum_{(v_1,q_1):v_1\neq\circ } \mu_1(v_1,q_1) \cdot (\pot{v_1}{ \tau} + q_1 + p),
      \end{align*}\end{small}
      and also for all $(v_1,q_1) \in \mathrm{supp}(\mu_1)$ such that $v_1 \neq \circ$,
      \begin{small}\[
      \pot{V,x \mapsto v_1 }{ \Gamma_2,x:\tau} +p \ge \sum_{q_2} \mu_{(v_1,q_1)}(\circ,q_2) \cdot q_2+ \sum_{(v_2,q_2):v_2 \neq \circ } \mu_{(v_1,q_1)}(v_2,q_2) \cdot (\pot{v_2 }{ A} + q_2).
      \]\end{small}
      Thus we find the following chain of inequalities.
      \begin{footnotesize}\begin{align*}
      &\pot{V}{\Gamma} + q  \\
      ={}& \pot{V}{\Gamma_1} + q + \pot{V}{\Gamma_2} 
      \\
       \ge{}&  \sum_{q_1} \mu_1(\circ,q_1) \cdot q_1 + \smashoperator{\sum_{(v_1,q_1):v_1\neq\circ }} \mu_1(v_1,q_1) \cdot (\pot{v_1}{ \tau} + q_1 + p) + \pot{V}{\Gamma_2} 
      \\
       \ge{} &\sum_{q_1} \mu_1(\circ,q_1) \cdot q_1 + \smashoperator{\sum_{(v_1,q_1):v_1\neq\circ }} \mu_1(v_1,q_1) \cdot (\pot{v_1}{ \tau} + q_1 + p + \pot{V}{\Gamma_2})
      \\
      = {} & \sum_{q_1} \mu_1(\circ,q_1) \cdot q_1 +\smashoperator{\sum_{(v_1,q_1):v_1\neq \circ }} \mu_1(v_1,q_1) \cdot ( q_1 + p + \pot{V, x\mapsto v_1 }{\Gamma_2,x:\tau})  
      \\
      \ge{} & \sum_{q_1} \mu_1(\circ,q_1) \cdot q_1 
       + \smashoperator{\sum_{(v_1,q_1):v_1\neq\circ }} \mu_1(v_1,q_1) \cdot (q_1 + \sum_{q_2} \mu_{(v_1,q_1)}(\circ,q_2) \cdot q_2+ \smashoperator{\sum_{(v_2,q_2):v_2 \neq \circ }} \mu_{(v_1,q_1)}(v_2,q_2) \cdot (\pot{v_2 }{ A} + q_2)) 
      \\
      = {} &\sum_{q_1} \mu_1(\circ,q_1) \cdot q_1 
       + \smashoperator{\sum_{(v_1,q_1):v_1\neq\circ }} \mu_1(v_1,q_1) \cdot ( \sum_{q_2} \mu_{(v_1,q_1)}(\circ,q_2) \cdot (q_1 + q_2) + \smashoperator{\sum_{(v_2,q_2):v_2 \neq \circ }} \mu_{(v_1,q_1)}(v_2,q_2) \cdot (\pot{v_2 }{ A} + q_1 + q_2)) 
      \\
       = {} & \sum_{q_1} \mu_1(\circ,q_1) \cdot q_1+\smashoperator{\sum_{\substack{(v_1,q_1):v_1\neq \circ \\ q_2 }}}  \mu_1(v_1,q_1) \cdot \mu_{(v_1,q_1)}(\circ,q_2) \cdot (q_1+q_2) 
       + \smashoperator{\sum_{\substack{(v_1,q_1):v_1\neq \circ \\ (v_2,q_2):v_2\neq \circ}}}  \mu_1(v_1,q_1) \cdot \mu_{(v_1,q_1)}(v_2,q_2) \cdot (\pot{v_2 }{ A} + q_1 + q_2).
      \end{align*}\end{footnotesize}
      The final line satisfies the following identities, completing the case.
      \begin{footnotesize}\begin{align*}
          & \sum_{q_0} \mu(\circ,q_0) \cdot q_0 +  \smashoperator{\sum_{(v_0,q_0) :v_0 \neq \circ}} \mu(v_0,q_0) \cdot (\pot{v_0 }{ A} + q_0)
           \\
           ={}& \sum_{q_0} \left( \sum_{q_1} \mu_1(\circ,q_1) \cdot \delta(\circ,q_1) + \smashoperator{\sum_{\substack{(v_1,q_1):v_1 \neq \circ \\ (v_2,q_2)}}}   \mu_1(v_1,q_1) \cdot \mu_{(v_1,q_1)}(v_2,q_2) \cdot \delta(v_2,q_1+q_2) \right) (\circ,q_0) \cdot q_0 
          \\
           & + \smashoperator[l]{\sum_{(v_0,q_0) : v_0 \neq \circ}} \left(\sum_{q_1} \mu_1(\circ,q_1) \cdot \delta(\circ,q_1) + \smashoperator{\sum_{\substack{(v_1,q_1):v_1 \neq \circ \\ (v_2,q_2)}} } \mu_1(v_1,q_1) \cdot \mu_{(v_1,q_1)}(v_2,q_2) \cdot \delta(v_2,q_1+q_2) \right)(v_0,q_0) \cdot (\pot{v_0}{A} + q_0)
           \\
          ={}& \sum_{q_0} \left( \sum_{q_1} \mu_1(\circ,q_1) \cdot \delta(\circ,q_1) + \smashoperator{\sum_{\substack{(v_1,q_1):v_1 \neq \circ \\ q_2}}}   \mu_1(v_1,q_1) \cdot \mu_{(v_1,q_1)}(\circ,q_2) \cdot \delta(\circ,q_1+q_2) \right)(\circ,q_0) \cdot q_0 
          \\
           & +\smashoperator[l]{\sum_{(v_0,q_0) : v_0 \neq \circ}} \left(\smashoperator[r]{\sum_{\substack{(v_1,q_1):v_1 \neq \circ \\(v_2,q_2):v_2 \neq \circ}} }  \mu_1(v_1,q_1) \cdot \mu_{(v_1,q_1)}(v_2,q_2) \cdot \delta(v_2,q_1+q_2) \right)(v_0,q_0) \cdot (\pot{v_0}{A} + q_0) 
          \\
          ={}& \sum_{q_1} \mu_1(\circ,q_1) \cdot q_1  + \smashoperator{\sum_{\substack{(v_1,q_1) : v_1 \neq \circ \\ q_2}}} \mu_1(v_1,q_1) \cdot \mu_{(v_1,q_1)}(\circ,q_2) \cdot (q_1+q_2)  +  \smashoperator{\sum_{\substack{(v_1,q_1):v_1 \neq \circ \\ (v_2,q_2) : v_2 \neq \circ}}} \mu_1(v_1,q_1) \cdot \mu_{(v_1,q_1)}(v_2,q_2) \cdot (\pot{v_2}{A} + q_1 + q_2). 
      \end{align*}\end{footnotesize}
      
      \item \textsc{(L:Flip)}
      By assumption, we have $\share{\Gamma}{p\times \Gamma_1}{(1-p)\times \Gamma_2}$, $q=p\cdot q_1+(1-p)\cdot q_2$, $\typingrel{\Gamma_1}{q_1}{e_1}{A}$, and $\typingrel{\Gamma_2}{q_2}{e_2}{A}$.
      By inversion, we have $\steprel{V}{e_1}{\mu_1}{n}$, $\steprel{V}{e_2}{\mu_2}{n}$, and $\mu = p \cdot \mu_1 + (1-p) \cdot \mu_2$.
      By the induction hypothesis, we have
      \begin{small}\begin{align*}
        \pot{V}{\Gamma_1} + q_1 & \ge \sum_{q_0} \mu_1(\circ,q_0) \cdot q_0 + \sum_{(v_0,q_0):v_0 \neq \circ} \mu_1(v_0,q_0) \cdot (\pot{v_0}{A} + q_0), \\
        \pot{V}{\Gamma_2} + q_2 & \ge \sum_{q_0} \mu_2(\circ,q_0) \cdot q_0 + \sum_{(v_0,q_0):v_0 \neq \circ} \mu_2(v_0,q_0) \cdot (\pot{v_0}{A}+q_0).
      \end{align*}
      \end{small}
      Thus we conclude this case by
      \begin{small}
        \begin{align*}
          \pot{V}{\Gamma}+q & = \pot{V}{p \times \Gamma_1} + \pot{V}{(1-p) \times \Gamma_2} + p \cdot q_1 + (1-p) \cdot q_2 \\
          & = p \cdot \pot{V}{\Gamma_1} + (1-p) \cdot \pot{V}{\Gamma_2} + p \cdot q_1 + (1-p) \cdot q_2 \\
          & = p \cdot (\sum_{q_0} \mu_1(\circ,q_0) \cdot q_0 + \sum_{(v_0,q_0):v_0 \neq \circ} \mu_1(v_0,q_0) \cdot (\pot{v_0}{A} + q_0)) \\
          & + (1-p) \cdot (\sum_{q_0} \mu_2(\circ,q_0) \cdot q_0 + \sum_{(v_0,q_0):v_0 \neq \circ} \mu_2(v_0,q_0) \cdot (\pot{v_0}{A}+q_0)) \\
          & = \sum_{q_0} q_0 \cdot (p \cdot \mu_1(\circ,q_0) + (1-p) \cdot \mu_2(\circ,q_0)) \\
          & + \sum_{(v_0,q_0):v_0\neq \circ} (\pot{v_0}{A} + q_0) \cdot (p \cdot \mu_1(v_0,q_0) + (1-p) \cdot \mu_2(v_0,q_0)) \\
          & = \sum_{q_0} \mu(\circ,q_0) \cdot q_0 + \sum_{(v_0,q_0):v_0\neq \circ} \mu(v_0,q_0) \cdot (\pot{v_0}{A} + q_0).
        \end{align*}
      \end{small}
      
      \item \textsc{(L:Prob)}
      By assumption, we know that $\Gamma = \cdot$, $q = p \cdot q_H + (1-p) \cdot q_T$, and $A = \annoA{\probA{q_H}{q_T}}{0}$ for some $q_H,q_T \in \bbQ_{\ge 0}$.
      By inversion, we have $\mu = \delta(\probC{p},0)$.
      Thus we conclude this case by
      \begin{small}\begin{align*}
        \pot{V}{\Gamma} + q & = p \cdot q_H + (1-p) \cdot q_T \\
        & = \pot{\probV{p}}{\probA{q_H}{q_T}} + 0 = \sum_{(v_0,q_0):v_0 \neq \circ} \mu(v_0,q_0) \cdot (\pot{v_0}{A} + q_0) \\
        & = \sum_{(v_0,q_0):v_0 \neq \circ} \mu(v_0,q_0) \cdot (\pot{v_0}{A} + q_0) + \sum_{q_0} \mu(\circ,q_0) \cdot q_0.
      \end{align*}  
      \end{small}
      
      \item \textsc{(L:FlipS)}
      By assumption, we have $\typingrel{\Gamma}{q+q_H}{e_1}{A}$ and $\typingrel{\Gamma}{q+q_T}{e_2}{A}$.
      By inversion, we have $V(x) = \probV{p}$, $\steprel{V}{e_1}{\mu_1}{n}$, $\steprel{V}{e_2}{\mu_2}{n}$, and $\mu = p \cdot \mu_1 + (1-p) \cdot \mu_2$ for some $p \in [0,1]$.
      By the induction hypothesis, we have
      \begin{small}\begin{align*}
        \pot{V}{\Gamma} + q + q_H & \ge \sum_{q_0} \mu_1(\circ,q_0) \cdot q_0 + \sum_{(v_0,q_0)} \mu_1(v_0,q_0) \cdot (\pot{v_0}{A} + q_0), \\
         \pot{V}{\Gamma} + q + q_T & \ge \sum_{q_0} \mu_2(\circ,q_0) \cdot q_0 + \sum_{(v_0,q_0)} \mu_2(v_0,q_0) \cdot (\pot{v_0}{A} + q_0).
      \end{align*}  
      \end{small}
      Thus we conclude this case by
      \begin{small}\begin{align*}
        & \pot{V}{\Gamma,x:\probA{q_H}{q_T}} + q \\
         ={} & \pot{V}{\Gamma} + p \cdot q_H + (1-p) \cdot q_T + q \\
         ={}& (p \cdot \pot{V}{\Gamma} + (1-p) \cdot \pot{V}{\Gamma}) + (p \cdot q_H + (1-p) \cdot q_T) + (p \cdot q + (1-p) \cdot q) \\
        ={}&  p \cdot (\pot{V}{\Gamma} + q + q_H) + (1-p) \cdot (\pot{V}{\Gamma} + q + q_T) \\
        \ge{} &  p \cdot ( \sum_{q_0} \mu_1(\circ,q_0) \cdot q_0 + \sum_{(v_0,q_0):v_0\neq \circ} \mu_1(v_0,q_0) \cdot (\pot{v_0}{A} + q_0) ) \\
        & + (1-p) \cdot ( \sum_{q_0} \mu_2(\circ,q_0) \cdot q_0 + \sum_{(v_0,q_0):v_0\neq \circ} \mu_2(v_0,q_0) \cdot (\pot{v_0}{A} + q_0) ) \\
        ={}& \sum_{q_0} (p \cdot \mu_1(\circ,q_0) + (1-p) \cdot \mu_2(\circ,q_0)) \cdot q_0 \\
        & + \sum_{(v_0,q_0):v_0\neq\circ} (p \cdot \mu_1(v_0,q_0) + (1-p) \cdot \mu_2(v_0,q_0)) \cdot (\pot{v_0}{A } + q_0) \\
        ={} & \sum_{q_0} \mu(\circ,q_0) \cdot q_0 + \sum_{(v_0,q_0):v_0\neq\circ} \mu(v_0,q_0) \cdot (\pot{v_0}{A} + q_0).
      \end{align*}  
      \end{small}
  \end{itemize}
\end{proof}

\subsection{\cref{lem:partialorder}}

\begin{proof}
  Let's consider the $\omega$-chain completeness.
  Let $\nu(v,q) \defeq \lim_{n \to \infty} \mu_n(v,q)$ for all $v \neq \circ$ and $q$.
  Let $P \defeq \sum_{(v,q) : v \neq \circ} \nu(v,q)$.
  We want to construct $\nu^\circ$ to be the ``limit'' of $\{\lambda q. \mu_n(\circ,q)\}_{n \in \bbN}$.
  Define $f_n^\circ(q) \defeq \mu_n(\{\circ\} \times [0,q])$ for all $n$ and $q$.
  Then for each $n \in \bbN$, $f_n^\circ$ is monotone and right-$\omega$-continuous and for all $q \in \bbQ_{\ge 0} \cup \{ \infty\}$, $\{f_n^\circ(q)\}_{n \in \bbN}$ is non-increasing.
  Let $f^\circ$ be the pointwise limit of $\{f_n^\circ\}_{n \in \bbN}$.
  Because both $\bbQ_{\ge 0} \cup \{\infty\}$ and $[0,1]$ are $\omega$-complete partially ordered sets, the right-$\omega$-continuous functions between them also form an $\omega$-complete partially ordered set.
  Therefore, $f^\circ$ is also right-$\omega$-continuous, and we can define $\nu^\circ(q) \defeq f^\circ(q) - \lim_{q' \to q^{-}} f^\circ(q')$.
  The final step is to prove that $f^\circ(\infty) = 1-P$.
  For each $n \in \bbN$, we have $f_n^\circ(\infty) = 1 - \sum_{(v,q) : v \neq \circ} \mu_n(v,q)$.
  Thus $f^\circ(\infty) = \lim_{n \to \infty} f^\circ_n(\infty) = \lim_{n \to \infty} (1 - \sum_{(v,q) : v \neq \circ} \mu_n(v,q)) = 1 - \lim_{n \to \infty} \sum_{(v,q) : v \neq \circ} \mu_n(v,q) = 1 - \sum_{(v,q):v\neq \circ} \lim_{n \to \infty} \mu_n(v,q) = 1 - P$ by the Monotone Convergence Theorem.
\end{proof}

\subsection{\cref{lemma:dist bound}}

\begin{proof}
  Let $\mu \defeq \bigsqcup_{n \in \bbN} \mu_n$ and define $f_n^\circ(q) \defeq \mu_n((\mathsf{Val} \cup \{\circ\}) \times [0,q])$.
  Similarly to the proof of $\omega$-chain completeness, $f_n^\circ$ is monotone and right-$\omega$-continuous for each $n \in \bbN$ and for all $q \in \bbQ_{\ge 0} \cup \{\infty\}$, $\{f_n^\circ(q)\}_{n \in \bbN}$ is non-increasing.
  Moreover, $f_n^\circ(\infty) = 1$ for all $n \in \bbN$.
  Now we extend the the domain of $f_n^\circ$ from $\bbQ_{\ge 0} \cup \{\infty\}$ to $\bbR_{\ge 0} \cup \{\infty\}$ as $g_n^\circ(r) = \lim_{q \to r^{+}} f_n^\circ(q)$.
  By the right-$\omega$-continuity, we know that $g_n^\circ(q) = f_n^\circ(q)$ for all $q \in \bbQ_{\ge 0} \cup \{\infty\}$.
  Therefore, $\sum_{(v,q)} \mu_n(v,q) \cdot q = \int (g_n^\circ(\infty) - g_n^\circ(r)) dr$.
  Let $g^\circ$ be the pointwise limit of $\{g^\circ_n\}_{n \in \bbN}$, so $g^\circ$ is also right-$\omega$-continuous. Thus
  \begin{small}\begin{align*}
  h(\mu) & =  \sum_q \mu(\circ,q) \cdot q + \sum_{(v,q):v \neq \circ} \mu(v,q) \cdot (\pot{v}{A} + q) \\
  & = \sum_{(v,q)} \mu(v,q) \cdot q + \sum_{(v,q):v\neq \circ} \mu(v,q) \cdot \pot{v}{A} \\
  & =  \int (g^\circ(\infty) - g^\circ(r)) dr + \sum_{(v,q):v\neq \circ} \mu(v,q) \cdot \pot{v}{A} \\
  & = \int (1-g^\circ(r)) dr + \sum_{(v,q):v\neq \circ} \mu(v,q) \cdot \pot{v}{A} \\
  & = \int (1 - \lim_{n \to \infty} g_n^\circ(r)) dr + \sum_{(v,q):v\neq \circ} \lim_{n \to \infty} \mu_n(v,q) \cdot \pot{v}{A} \\
  & = \int \lim_{n \to \infty} (1 - g_n^\circ(r)) dr + \sum_{(v,q):v\neq \circ} \lim_{n \to \infty} \mu_n(v,q) \cdot \pot{v}{A} \\
  & = \lim_{n \to \infty} \int (1 - g_n^\circ(r)) dr + \lim_{n \to \infty} \sum_{(v,q) : v\neq \circ}  \mu_n(v,q) \cdot \pot{v}{A}.
  \end{align*}\end{small}
  Since $h(\mu_n) = \int  (g_n^\circ(\infty) - g_n^\circ(r)) dr + \sum_{(v,q):v\neq \circ} \mu_n(v,q) \cdot \pot{v}{A} = \int (1-g_n^\circ(r)) dr + \sum_{(v,q):v\neq \circ} \mu_n(v,q) \cdot \pot{v}{A}$ and $h(\mu_n) \le M$ for all $n \in \bbN$, we conclude that $h(\mu) \le \sup_{n \in \bbN} h(\mu_n) \le M$.
\end{proof}

} {}

\end{document}